\pdfoutput=1 
\documentclass{article}%
\usepackage{amsmath}
\usepackage{amsfonts}
\usepackage{amssymb}
\usepackage{graphicx}
\providecommand{\U}[1]{\protect\rule{.1in}{.1in}}
\newtheorem{theorem}{Theorem}

\newtheorem{algorithm}[theorem]{Algorithm}

\newtheorem{corollary}[theorem]{Corollary}

\newenvironment{proof}[1][Proof]{\noindent\textbf{#1.} }{\ \rule{0.5em}{0.5em}}
\begin{document}

\title{\textbf{Phi-divergence statistics for the likelihood ratio order: an approach
based on log-linear models}}
\author{Martin, N.$^{1}$, Mata, R.$^{2}$ and Pardo, L.$^{2}$\thanks{Corresponding
author. E-mail: lpardo@mat.ucm.es, Tel: (+34) 91 394 4425, Fax: (+34) 91 394
4606.}\\$^{1}${\small Department of Statistics, Carlos III University of Madrid, 28903
Getafe (Madrid), Spain}\\$^{2}${\small Department of Statistics and O.R., Complutense University of
Madrid, 28040 Madrid, Spain} }
\date{\today}
\maketitle

\begin{abstract}
When some treatments are ordered according to the categories of an ordinal
categorical variable (e.g., extent of side effects) in a monotone order, one
might be interested in knowing wether the treatments are equally effective or
not. One way to do that is to test if the likelihood ratio order is strictly
verified. A method based on log-linear models is derived to make statistical
inference and phi-divergence test-statistics are proposed for the test of
interest. Focussed on loglinear modeling, the theory associated with the
asymptotic distribution of the phi-divergence test-statistics is developed. An
illustrative example motivates the procedure and a simulation study for small
and moderate sample sizes shows that it is possible to find phi-divergence
test-statistic with an exact size closer to nominal size and higher power in
comparison with the classical likelihood ratio.

\end{abstract}

\bigskip


\section{Introduction\label{Intro}}

In this paper we are interested in comparing $I$ treatments when the response
variable is ordinal with $J$ categories. We can consider each treatment type
to be each of the $I$ ordinal categories of a variable $X$. We shall denote by
$Y$ the response variable and its conditional probabilities by%
\[
\boldsymbol{\pi}_{i}=(\pi_{i1},...,\pi_{iJ})^{T},\quad i=1,..,I,
\]
with
\[
\pi_{ij}=\Pr\left(  Y=j|X=i\right)  ,\quad j=1,...,J.
\]
For the $i$-th treatment and for each individual taken independently from a
sample of size $n_{i}$ its response is classified to be $\{1,...,J\}$
according to the conditional distribution of $Y|X=i$, $\boldsymbol{\pi}_{i}$.
In this setting the $J$-dimensional random variable associated with the
observed frequencies,%
\[
\boldsymbol{N}_{i}=(N_{i1},...,N_{iJ})^{T},
\]
is multinomially distributed with parameters $n_{i}$\ and $\boldsymbol{\pi
}_{i}$. Assuming that the different treatments are independent, the
probability distribution of the $I\times J$ dimensional random variable
$\boldsymbol{N}=(\boldsymbol{N}_{1}^{T},...,\boldsymbol{N}_{I}^{T})^{T}$ is
product-multinomial. We are going to consider a motivation example, taken from
Section 5 in Dardanoni and Forcina (1998), in order to clarify the problem
considered in this paper. In Table \ref{tt1}, duodenal ulcer patients of a
hospital are cross-classified according to an increasing order of $I=4$
severity degrees of the operation, and the extent of side effects, categorized
as None, Slight and Moderate ($J=3$).%

\begin{table}[htbp]  \tabcolsep2.8pt  \centering
$%
\begin{tabular}
[c]{p{2.5cm}p{2.5cm}p{2.5cm}p{2.5cm}}\hline
& None & Slight & Moderate\\\hline
Treatment $1$ & $61$ & $28$ & $7$\\
Treatment $2$ & $68$ & $23$ & $13$\\
Treatment $3$ & $58$ & $40$ & $12$\\
Treatment $4$ & $53$ & $38$ & $16$\\\hline
\end{tabular}
\ \ \ \ \ \ \ \ \ $%
\caption{Extent of size effect of four treatments.\label{tt1}}%
\end{table}%

We shall consider that Treatment $i+1$ is as good as Treatment $i$, for
$i=1,...,I-1$ simultaneously, if $\tfrac{\Pr(Y=j|X=i+1)}{\Pr(Y=j|X=i)}$ is
non-decreasing for all $j\in\{1,...,J\}$, i.e.%
\begin{equation}
\tfrac{\Pr(Y=j|X=i+1)}{\Pr(Y=j|X=i)}\leq\tfrac{\Pr(Y=j+1|X=i+1)}%
{\Pr(Y=j+1|X=i)},\qquad\text{for every }(i,j)\in\{1,...,I-1\}\times
\{1,...,J-1\}. \label{eq1}%
\end{equation}
This is the so called \textquotedblleft likelihood ratio
ordering\textquotedblright, sometimes called also \textquotedblleft local
ordering\textquotedblright\ (see Silvapulle and Sen (2005), Chapter 6). It is
very important to clarify that the likelihood ratio ordering is more
thoroughly referred to $J$ independendent multinomial samples of sizes equal
to $1$, in such a way that $\pi_{ij}=\Pr\left(  Y=j|X=i\right)  =\Pr\left(
\mathcal{M}(1,\boldsymbol{\pi}_{i})=\boldsymbol{e}_{j}\right)  $, where
$\boldsymbol{e}_{j}$ is the $j$-th unit vector. In a similar way, Treatment
$i+1$ is better than Treatment $i$, for $i=1,...,I-1$ simultaneously, if
(\ref{eq1}) holds with at least one strict inequality. For testing that
Treatment $i+1$ is better than Treatment $i$, for $i=1,...,I-1$ we can
consider%
\begin{subequations}
\begin{align}
&  H_{0}:\;\tfrac{\Pr(Y=j|X=i+1)}{\Pr(Y=j|X=i)}=\tfrac{\Pr(Y=j+1|X=i+1)}%
{\Pr(Y=j+1|X=i)}\quad\text{for every }(i,j)\in\{1,...,I-1\}\times
\{1,...,J-1\}\text{,}\label{eq2}\\
&  H_{1}:\;\tfrac{\Pr(Y=j|X=i+1)}{\Pr(Y=j|X=i)}\leq\tfrac{\Pr(Y=j+1|X=i+1)}%
{\Pr(Y=j+1|X=i)}\quad\text{for every }(i,j)\in\{1,...,I-1\}\times
\{1,...,J-1\}\label{eq3}\\
&  \qquad\text{and}\quad\tfrac{\Pr(Y=j|X=i+1)}{\Pr(Y=j|X=i)}<\tfrac
{\Pr(Y=j+1|X=i+1)}{\Pr(Y=j+1|X=i)}\quad\text{for at least one }(i,j)\in
\{1,...,I-1\}\times\{1,...,J-1\}\text{.}\nonumber
\end{align}
For the motivation example, the null hypothesis means that all the treatments
have equal side effects, while the alternative hypothesis means that as the
more severe treatment is, greater is the probability of having side effects.
Note that if we multiply on the left and right hand side of (\ref{eq2}) and
(\ref{eq3}) by $\left(  \tfrac{\Pr(Y=j|X=i+1)}{\Pr(Y=j|X=i)}\right)  ^{-1}$ we
obtain%
\end{subequations}
\begin{subequations}
\begin{align}
&  H_{0}:\;\vartheta_{ij}=1,\quad\forall(i,j)\in\{1,...,I-1\}\times
\{1,...,J-1\},\label{eq4}\\
&  H_{1}:\;\vartheta_{ij}\geq1\quad\text{for every }(i,j)\in
\{1,...,I-1\}\times\{1,...,J-1\}\label{eq5}\\
&  \qquad\text{and}\quad\vartheta_{ij}>1\quad\text{for at least one }%
(i,j)\in\{1,...,I-1\}\times\{1,...,J-1\}\text{,}\nonumber
\end{align}
where $\vartheta_{ij}=\dfrac{\pi_{ij}\pi_{i+1,j+1}}{\pi_{i+1,j}\pi_{i,j+1}}%
$\ represent the \textquotedblleft local odds ratios\textquotedblright, also
called cross-product ratios.

If we denote by $n=%
{\textstyle\sum\nolimits_{i=1}^{I}}
n_{i}$ the total of the sample sizes, we can consider the joint distribution
to be%
\end{subequations}
\[
p_{ij}=\Pr\left(  X=i,Y=j\right)  =\Pr\left(  X=i\right)  \Pr\left(
Y=j|X=i\right)  =\frac{n_{i}}{n}\pi_{ij},\quad i=1,...,I,\;j=1,...,J.
\]
We display such a distribution in a rectangular table having $I$ rows for the
categories of $X$ and $J$ columns for the categories of $Y$, and we denote
$\boldsymbol{P}=(\mathbf{p}_{1},...,\mathbf{p}_{I})^{T}$, with $\mathbf{p}%
_{i}=(p_{i1},...,p_{iJ})^{T}$, $i=1,...,I$, the corresponding $I\times J$
matrix and
\begin{equation}
\boldsymbol{p}=\mathrm{vec}(\boldsymbol{P}^{T})=(\mathbf{p}_{1}^{T}%
,...,\mathbf{p}_{I}^{T})^{T} \label{0}%
\end{equation}
a vector obtained by stacking the columns of $\boldsymbol{P}^{T}$\ (i.e., the
rows of matrix $\boldsymbol{P}$). Note that the components of $\boldsymbol{P}$
are ordered in lexicographical order in $\boldsymbol{p}$. The local odds
ratios can be expressed only in terms of joint probabilities\
\begin{equation}
\vartheta_{ij}=\dfrac{p_{ij}p_{i+1,j+1}}{p_{i+1,j}p_{i,j+1}}=\dfrac{\pi
_{ij}\pi_{i+1,j+1}}{\pi_{i+1,j}\pi_{i,j+1}},\quad\forall(i,j)\in
\{1,...,I-1\}\times\{1,...,J-1\}. \label{2}%
\end{equation}

The likelihood ratio ordering has been extensely studied in order statistics.
In the literature related to order restricted inference for categorical data
analysis, the likelihood ratio ordering has received little attention. The
definition given in (\ref{eq1}) is not specific for multinomial random
variables, actually is very similar for any random variable, not neccesarily
discrete. In Bapar and Kochar (1994), it is mentioned that very important
families of random variables, such as the one-parameter exponential family of
distributions, have the likelihood ratio ordering property with respect to the
parameter. For two independent multinomial samples ($I=2$), Dykstra et al.
(1995) established the asymptotic distribution of the likelihood ratio
test-statistic and Dardanoni and Forcina extended it for a general problem of
$I$ independent multinomial samples. Recently, Davidov et al. (2010) have
highlighted its importance by considering it as a particular case of the power
bias model. This ordering is important not only in fields such as order
statistics and its consideration cannot be avoided in categorical data
analysis. Dykstra et al. (1995) argued that the likelihood ratio ordering,
applied in an adapted product-multinomial sampling context, is a useful method
for making statistical inference related to trend comparison of Poisson
processes. The merit of the work of Dardanoni and Forcina (1998) is not only
in the results obtained for a more general case, but also in the
parametrization used for the method developed to find the asymptotic
distribution of the likelihood ratio-test. They studied three kinds of
ordering in the same parametrization setting but any of them is considered to
be superior, for example the likelihood ratio ordering is stronger in
comparison with the stochastic one.

In Section \ref{Sec:LM} of this paper a new method is proposed, based on
log-linear modeling, for characterizing the likelihood ratio test for the
likelihood ratio ordering in several independent multinomial samples. None
paper has considered any alternative test-statistic to the likelihood ratio
one and hence it is interesting to study the performance of the phi-divergence
test-statistics. These test-statistics, which include as a particular case the
likelihood ratio one, are introduced in Section \ref{Sec:PDM}. It is proven,
in Section \ref{Sec2}, that the asymptotic distribution of the phi-divergence
test-statistics is chi-bar and an algorithm is also provided to find its
weights in a simple way. An illustrative example is given in Section \ref{Ex},
as well as an algorithm, to clarify the method and the computational aspects.
We consider that the likelihood ratio ordering, as strong ordering type and
nested model within other ordering types, is a useful order since it should
have asymptotically, under an alternative hypothesis of likelihood ratio
order, much power. One of our interests in this regard, is to study though
simulation the performance of the likelihood ratio test-statistic for small
and moderate multinomial sample sizes, as Wang (1996) did in relation to the
stochastic order. This is done in Section \ref{MC}.

\section{Modeling local odds ratios through loglinear models \label{Sec:LM}}

Our first aim in this paper is to formulate the hypothesis testing problem
making a reparametrization using the saturated loglinear model associated with
$\boldsymbol{p}$, so that the restrictions are linear with respect to the new
parameters. Focussed on $\boldsymbol{p}$, the saturated loglinear model with
canonical parametrization is defined by
\begin{equation}
\log p_{ij}=u+u_{1(i)}+\theta_{2(j)}+\theta_{12(ij)}, \label{3}%
\end{equation}
with%
\begin{equation}
u_{1(I)}=0,\quad\theta_{2(J)}=0,\quad\theta_{12(iJ)}=0,i=1,...,I-1,\quad
\theta_{12(Ij)}=0,j=1,...,J. \label{ident}%
\end{equation}
It is important to clarify that we have used the identifiability constraints
(\ref{ident}) in order to make easier the calculations. Similar conditions
have been used for instance in Lang (1996, examples of Section 7) and Sen and
Silvapulle (2005, exercise 6.25 in page 345). Let $\boldsymbol{\theta}%
_{2}=(\theta_{2(1)},...,\theta_{2(J-1)})^{T}$, $\boldsymbol{\theta}%
_{12(i)}=(\theta_{12(i1)},...,\theta_{12(i,J-1)})^{T}$, $i=1,...,I-1$, be the
subvector of unknown parameters of , $\boldsymbol{\theta}=(\boldsymbol{\theta
}_{2}^{T},\boldsymbol{\theta}_{12(1)}^{T},...,\boldsymbol{\theta}%
_{12(I-1)}^{T})^{T}$ and the vector of redundant components of the model
$\boldsymbol{u}=(u,u_{1(1)},...,u_{1(I-1)})^{T}$, since its components can be
expressed in terms of $\boldsymbol{\theta}$\ as follows%
\begin{align}
u  &  =u(\boldsymbol{\theta})=\log n_{I}-\log n-\log%
{\textstyle\sum\nolimits_{j=1}^{J}}
\exp\{\theta_{2(j)}\}\nonumber\\
&  =\log n_{I}-\log n-\log\left(  1+\boldsymbol{1}_{J-1}^{T}\exp
\{\boldsymbol{\theta}_{2}\}\right)  , \label{u}%
\end{align}%
\begin{align}
u_{1(i)}  &  =u_{1(i)}(\boldsymbol{\theta})=\log n_{i}-\log n-\log%
{\textstyle\sum\nolimits_{j=1}^{J}}
\exp\{\theta_{2(j)}+\theta_{12(ij)}\}-u(\boldsymbol{\theta}).\nonumber\\
&  =\log n_{i}-\log n_{I}+\log\left(  1+%
{\textstyle\sum\nolimits_{j=1}^{J-1}}
\exp\{\theta_{2(j)}\}\right)  -\log\left(  1+%
{\textstyle\sum\nolimits_{j=1}^{J-1}}
\exp\{\theta_{2(j)}+\theta_{12(ij)}\}\right) \nonumber\\
&  =\log\frac{n_{i}}{n_{I}}+\log\frac{1+\boldsymbol{1}_{J-1}^{T}%
\exp\{\boldsymbol{\theta}_{2}\}}{1+\boldsymbol{1}_{J-1}^{T}\exp
\{\boldsymbol{\theta}_{2}+\boldsymbol{\theta}_{12(i)}\}},\quad i=1,...,I-1.
\label{theta1}%
\end{align}
Note that the expressions of $u_{1(i)}=u_{1(i)}(\boldsymbol{\theta})$,
$i=1,..,I-1$ and $u=u(\boldsymbol{\theta})$ are obtained taking into account
that for product-multinomial sampling $%
{\textstyle\sum\nolimits_{j=1}^{J}}
p_{ij}(\boldsymbol{\theta})=\frac{n_{i}}{n}$, $i=1,...,I$. In matrix notation,
(\ref{3}) is given by%
\begin{align}
\log\boldsymbol{p}(\boldsymbol{\theta})  &  =\left(
\begin{pmatrix}
\boldsymbol{1}_{I-1} & \boldsymbol{I}_{I-1}\\
1 & \boldsymbol{0}_{I-1}^{T}%
\end{pmatrix}
\otimes%
\begin{pmatrix}
\boldsymbol{1}_{J-1} & \boldsymbol{I}_{J-1}\\
1 & \boldsymbol{0}_{J-1}^{T}%
\end{pmatrix}
\right)
\begin{pmatrix}
\boldsymbol{u}\\
\boldsymbol{\theta}%
\end{pmatrix}
\nonumber\\
&  =\boldsymbol{W}_{0}\boldsymbol{u}+\boldsymbol{W\theta}, \label{loglin}%
\end{align}
where $\otimes$ is the Kronecker product (see Chapter 16 of Harville (2008)),
$\boldsymbol{I}_{a}$\ is the the identity matrix of order $a$, $\boldsymbol{0}%
_{a}$\ is the $a$-vector of zeros, $\boldsymbol{p}(\boldsymbol{\theta})$ is
$\boldsymbol{p}$\ such that the components are defined by (\ref{3}) and%
\[
\boldsymbol{W}_{0}=%
\begin{pmatrix}
\boldsymbol{1}_{I-1} & \boldsymbol{I}_{I-1}\\
1 & \boldsymbol{0}_{I-1}^{T}%
\end{pmatrix}
\otimes\boldsymbol{1}_{J},\qquad\boldsymbol{W}=%
\begin{pmatrix}
\boldsymbol{1}_{I-1} & \boldsymbol{I}_{I-1}\\
1 & \boldsymbol{0}_{I-1}^{T}%
\end{pmatrix}
\otimes%
\begin{pmatrix}
\boldsymbol{I}_{J-1}\\
\boldsymbol{0}_{J-1}^{T}%
\end{pmatrix}
.
\]
Condition (\ref{eq1}) can be expressed by the linear constraint
\begin{equation}
\theta_{12(ij)}-\theta_{12(i+1,j)}-\theta_{12(i,j+1)}+\theta_{12(i+1,j+1)}%
\geq0,\text{ }\forall(i,j)\in\{1,...,I-1\}\times\{1,...,J-1\}, \label{6}%
\end{equation}
because%
\[
\log\vartheta_{ij}=\log p_{ij}-\log p_{i+1,j}-\log p_{i,j+1}+\log
p_{i+1,j+1}=\theta_{12(ij)}-\theta_{12(i+1,j)}-\theta_{12(i,j+1)}%
+\theta_{12(i+1,j+1)}.
\]
Let us consider $\boldsymbol{R\theta}\geq\boldsymbol{0}_{(I-1)(J-1)}$, with
$\boldsymbol{R}=\boldsymbol{e}_{J}\otimes\left(  \boldsymbol{G}_{I-1}%
\otimes\boldsymbol{G}_{J-1}\right)  =(\boldsymbol{0}_{(I-1)(J-1)\times
(J-1)},\boldsymbol{G}_{I-1}\otimes\boldsymbol{G}_{J-1})$, $\boldsymbol{0}%
_{a\times b}$ is the $a\times b$\ matrix of zeros and $\boldsymbol{G}_{h}$ is
a $h\times h$\ matrix with $1$-s in the main diagonal and $-1$-s in the upper
superdiagonal. Such restriction are equivalent to condition (\ref{6}). Observe
that the restrictions can be expressed also as $\boldsymbol{R}_{12}%
\boldsymbol{\theta}_{12}\geq\boldsymbol{0}_{(I-1)(J-1)}$, and
$\boldsymbol{\theta}_{2}$\ is a nuisance parameter vector because it does not
take part actively in the restrictions. The kernel of the likelihood function
with the new parametrization is obtained replacing $\boldsymbol{p}$\ by
$\boldsymbol{p}(\boldsymbol{\theta})$ in (\ref{0b}), i.e.
\[
\ell(\boldsymbol{N};\boldsymbol{\theta})=\boldsymbol{N}^{T}\log\boldsymbol{p}%
(\boldsymbol{\theta})=\boldsymbol{N}^{T}\boldsymbol{W}_{0}(\boldsymbol{u}%
(\widetilde{\boldsymbol{\theta}})-\boldsymbol{u}(\widehat{\boldsymbol{\theta}%
}))+\boldsymbol{N}^{T}\boldsymbol{W}(\widetilde{\boldsymbol{\theta}%
}-\widehat{\boldsymbol{\theta}})=nu(\boldsymbol{\theta})+%
{\textstyle\sum\nolimits_{i=1}^{I-1}}
n_{i}u_{1(i)}(\boldsymbol{\theta})+\boldsymbol{N}^{T}\boldsymbol{W\theta}.
\]
Hypotheses (\ref{eq2})-(\ref{eq3}) or (\ref{eq4})-(\ref{eq5}) can be now
formulated as%
\begin{equation}
H_{0}:\boldsymbol{R\theta}=\boldsymbol{0}_{(I-1)(J-1)}\text{ versus }%
H_{1}:\boldsymbol{R\theta}\geq\boldsymbol{0}_{(I-1)(J-1)}\text{ and
}\boldsymbol{R\theta}\neq\boldsymbol{0}_{(I-1)(J-1)}\text{.} \label{4b}%
\end{equation}

\section{Test-statistics based on phi-divergence measures\label{Sec:PDM}}

The likelihood function in our model is $\mathcal{L}(\boldsymbol{N}%
;\boldsymbol{p})=%
{\textstyle\prod\nolimits_{i=1}^{I}}
k_{i}%
{\textstyle\prod\nolimits_{j=1}^{J}}
\pi_{ij}^{N_{ij}}$, where $k_{i}=n_{i}!/%
{\textstyle\prod\nolimits_{j=1}^{J}}
N_{ij}!$, and the kernel of the loglikelihood function, in terms of
$\boldsymbol{p}$, is%
\begin{equation}
\ell(\boldsymbol{N};\boldsymbol{p})=%
{\displaystyle\sum\limits_{i=1}^{I}}
{\displaystyle\sum\limits_{j=1}^{J}}
N_{ij}\log p_{ij}. \label{0b}%
\end{equation}
Under $H_{0}$, the parameter space is $\Theta_{0}=\left\{  \boldsymbol{\theta
}\in%
\mathbb{R}
^{I(J-1)}:\boldsymbol{R\theta}=\boldsymbol{0}_{(I-1)(J-1)}\right\}  $ and the
maximum likelihood estimator of $\boldsymbol{\theta}$ in $\Theta_{0}$ is
$\widehat{\boldsymbol{\theta}}=\arg\max_{\boldsymbol{\theta\in}\Theta_{0}}%
\ell(\boldsymbol{N};\boldsymbol{\theta})$. Under either $H_{0}$\ or $H_{1}$,
the overall parameter space is $\Theta_{1}=\left\{  \boldsymbol{\theta}\in%
\mathbb{R}
^{I(J-1)}:\boldsymbol{R\theta}\geq\boldsymbol{0}_{(I-1)(J-1)}\right\}  $ and
the maximum likelihood estimator of $\boldsymbol{\theta}$ in $\Theta_{1}$ is
$\widetilde{\boldsymbol{\theta}}=\arg\max_{\boldsymbol{\theta\in}\Theta_{1}%
}\ell(\boldsymbol{N};\boldsymbol{\theta})$.

The likelihood ratio test-statistic for testing (\ref{4b}) is%
\begin{equation}
G^{2}=2(\ell(\boldsymbol{N};\widetilde{\boldsymbol{\theta}})-\ell
(\boldsymbol{N};\widehat{\boldsymbol{\theta}}%
))=2n(u(\widetilde{\boldsymbol{\theta}})-u(\widehat{\boldsymbol{\theta}}))+2n%
{\textstyle\sum\nolimits_{i=1}^{I-1}}
(u_{1(i)}(\widetilde{\boldsymbol{\theta}})-u_{1(i)}%
(\widehat{\boldsymbol{\theta}}))+2\boldsymbol{N}^{T}\boldsymbol{W}\left(
\widetilde{\boldsymbol{\theta}}-\widehat{\boldsymbol{\theta}}\right)  ,
\label{LRT}%
\end{equation}
and the chi-square test-statistic for testing (\ref{4b}) is%
\begin{equation}
X^{2}=n%
{\displaystyle\sum\limits_{i=1}^{I}}
{\displaystyle\sum\limits_{j=1}^{J}}
\frac{(p_{ij}(\widehat{\boldsymbol{\theta}})-p_{ij}%
(\widetilde{\boldsymbol{\theta}}))^{2}}{p_{ij}(\widehat{\boldsymbol{\theta}}%
)}. \label{CS}%
\end{equation}
Let $\overline{\boldsymbol{p}}=\boldsymbol{N}/n$ the vector of relative
frequencies,%
\[
d_{Kull}(\boldsymbol{p},\boldsymbol{q})=\sum_{i=1}^{I}\sum_{j=1}^{J}p_{ij}%
\log\frac{p_{ij}}{q_{ij}}%
\]
the Kullback-Leibler divergence measure between two $IJ$-dimensional
probability vectors $\boldsymbol{p}$ and $\boldsymbol{q}$, and
\[
d_{Pearson}(\boldsymbol{p},\boldsymbol{q})=\frac{1}{2}\sum_{i=1}^{I}\sum
_{j=1}^{J}\frac{(p_{ij}-q_{ij})^{2}}{q_{ij}}%
\]
the Pearson divergence measure. It is not difficult to check that
$G^{2}=2n(d_{Kull}(\overline{\boldsymbol{p}},\boldsymbol{p}%
(\widehat{\boldsymbol{\theta}}))-d_{Kull}(\overline{\boldsymbol{p}%
},\boldsymbol{p}(\widetilde{\boldsymbol{\theta}})))$ and $X^{2}=2nd_{Pearson}%
(\boldsymbol{p}(\widetilde{\boldsymbol{\theta}}),\boldsymbol{p}%
(\widehat{\boldsymbol{\theta}}))$. More general than the Kullback-Leibler
divergence and Pearson divergence measures are $\phi$-divergence measures,
defined as%
\begin{equation}
d_{\phi}(\boldsymbol{p},\boldsymbol{q})=\sum_{i=1}^{I}\sum_{j=1}^{J}q_{ij}%
\phi\left(  \frac{p_{ij}}{q_{ij}}\right)  , \label{div}%
\end{equation}
where $\phi:%
\mathbb{R}
_{+}\longrightarrow%
\mathbb{R}
$ is a convex function such that $\phi(1)=\phi^{\prime}(1)=0$, $\phi
^{\prime\prime}(1)>0$, $0\phi(\frac{0}{0})=0$, $0\phi(\frac{p}{0}%
)=p\lim_{u\rightarrow\infty}\frac{\phi(u)}{u}$, for $p\neq0$. For more details
about $\phi$-divergence measures see Pardo (2005).

Our second aim in this paper is to formulate test statistics valid for testing
(\ref{4b}). Apart from the likelihood ratio statistic (\ref{LRT}) and the
chi-square (\ref{CS}) statistic, we shall consider two family of
test-statistics based on $\phi$-divergence measures,%
\begin{equation}
T_{\phi}(\overline{\boldsymbol{p}},\boldsymbol{p}%
(\widetilde{\boldsymbol{\theta}}),\boldsymbol{p}(\widehat{\boldsymbol{\theta}%
}))=\frac{2n}{\phi^{\prime\prime}(1)}(d_{\phi}(\overline{\boldsymbol{p}%
},\boldsymbol{p}(\widehat{\boldsymbol{\theta}}))-d_{\phi}(\overline
{\boldsymbol{p}},\boldsymbol{p}(\widetilde{\boldsymbol{\theta}}))) \label{5a}%
\end{equation}
and%
\begin{equation}
S_{\phi}(\boldsymbol{p}(\widetilde{\boldsymbol{\theta}}),\boldsymbol{p}%
(\widehat{\boldsymbol{\theta}}))=\frac{2n}{\phi^{\prime\prime}(1)}d_{\phi
}(\boldsymbol{p}(\widetilde{\boldsymbol{\theta}}),\boldsymbol{p}%
(\widehat{\boldsymbol{\theta}})). \label{5b}%
\end{equation}
The two previous families of test-statistics can be considered as a natural
extension of likelihood ratio test-statistic and chi-square test-statistic
respectively. More thorouhly, for $\phi(x)=x\log x-x+1$ in (\ref{5a}), we get
the likelihood ratio test-statistic and for $\phi(x)=\frac{1}{2}(x-1)^{2}$ in
(\ref{5b}) we get the chi-square test-statistic.

Section \ref{Sec2} is devoted to present the main theoretical results in the
paper. In Section \ref{Ex}, an example illustrates the results of Section
\ref{Sec2}. A simulation study is developed in Section \ref{MC}, in order to
study the behaviour of the families of test-statistics introduced in
(\ref{5a}) and (\ref{5b}). Finally, we present an appendix in which we
establish some parts of the proofs of the results given in Section \ref{Sec2}.

\section{Asymptotic distribution\label{Sec2}}

For product-multinomial sampling, if we consider the partitioned matrix
$\boldsymbol{W}^{T}=(\boldsymbol{W}_{1}^{T},...,\boldsymbol{W}_{I}^{T})$, such
that $\log\boldsymbol{p}_{i}(\boldsymbol{\theta})=u\boldsymbol{1}_{J}%
+u_{1(i)}\boldsymbol{1}_{J}+\boldsymbol{W}_{i}\boldsymbol{\theta}$,
$i=1,...,I-1$, $\log\boldsymbol{p}_{I}(\boldsymbol{\theta})=u\boldsymbol{1}%
_{J}+\boldsymbol{W}_{I}\boldsymbol{\theta}$, and $\mathcal{I}_{F,i}%
^{(n_{1},...,n_{I})}(\boldsymbol{\theta})=\boldsymbol{W}_{i}^{T}%
(\boldsymbol{D}_{\boldsymbol{\pi}_{i}(\boldsymbol{\theta})}-\boldsymbol{\pi
}_{i}(\boldsymbol{\theta})\boldsymbol{\pi}_{i}^{T}(\boldsymbol{\theta
}))\boldsymbol{W}_{i}$, we have%
\begin{equation}
\mathcal{I}_{F}^{(n_{1},...,n_{I})}(\boldsymbol{\theta})=\sum_{i=1}^{I}%
\frac{n_{i}}{n}\mathcal{I}_{F,i}^{(n_{1},...,n_{I})}(\boldsymbol{\theta
})=\boldsymbol{W}^{T}\left(  \bigoplus_{i=1}^{I}\frac{n_{i}}{n}\left(
\boldsymbol{D}_{\boldsymbol{\pi}_{i}(\boldsymbol{\theta})}-\boldsymbol{\pi
}_{i}(\boldsymbol{\theta})\boldsymbol{\pi}_{i}^{T}(\boldsymbol{\theta
})\right)  \right)  \boldsymbol{W}, \label{FIM1}%
\end{equation}
where $\bigoplus_{h=1}^{a}\boldsymbol{A}_{h}$ denotes the direct sum between
the matrices $\{\boldsymbol{A}_{h}\}_{h=1}^{a}$. Our interest is to establish
the Fisher information matrix, for $n\rightarrow\infty$, under the assumption%
\[
\boldsymbol{\pi}_{1}(\boldsymbol{\theta}_{0})=\cdots=\boldsymbol{\pi}%
_{I}(\boldsymbol{\theta}_{0})=(\pi_{1}(\boldsymbol{\theta}_{0}),...,\pi
_{J}(\boldsymbol{\theta}_{0}))^{T}=\boldsymbol{\pi}(\boldsymbol{\theta}%
_{0}),\quad i=1,...,I,
\]
which is equivalent to the null hypothesis (\ref{eq2}) or (\ref{eq4}). Let
$\boldsymbol{\pi}^{\ast}(\boldsymbol{\theta}_{0})=(\pi_{1}(\boldsymbol{\theta
}_{0}),...,\pi_{J-1}(\boldsymbol{\theta}_{0}))^{T}$\ be the subvector of
$\boldsymbol{\pi}(\boldsymbol{\theta}_{0})$ obtained deleting the $J$-th
element, $\pi_{J}(\boldsymbol{\theta}_{0})$, from $\boldsymbol{\pi
}(\boldsymbol{\theta}_{0})$. If we consider the probability vector%
\[
\boldsymbol{\nu}=(\nu_{1},...,\nu_{I})^{T},
\]
such that $\nu_{i}=\lim_{n\rightarrow\infty}\frac{n_{i}}{n}$, $i=1,...,I$, we
denote by $\boldsymbol{\nu}^{\ast}=(\nu_{1},...,\nu_{I-1})^{T}$\ its subvector
obtained deleting the $I$-th element, $\nu_{I}$, from $\boldsymbol{\nu}$.

\begin{theorem}
\label{th0}If we denote $\mathcal{I}_{F}(\boldsymbol{\theta})=\lim
_{n\rightarrow\infty}\mathcal{I}_{F}^{(n_{1},...,n_{I})}(\boldsymbol{\theta})$
when $\boldsymbol{\theta}\in\Theta_{0}$, we have%
\begin{equation}
\mathcal{I}_{F}(\boldsymbol{\theta}_{0})=%
\begin{pmatrix}
1 & \boldsymbol{\nu}^{\ast T}\\
\boldsymbol{\nu}^{\ast} & \bigoplus_{i=1}^{I-1}\nu_{i}%
\end{pmatrix}
\otimes\left(  \boldsymbol{D}_{\boldsymbol{\pi}^{\ast}(\boldsymbol{\theta}%
_{0})}-\boldsymbol{\pi}^{\ast}(\boldsymbol{\theta}_{0})\boldsymbol{\pi}^{\ast
T}(\boldsymbol{\theta}_{0})\right)  . \label{FIM2}%
\end{equation}

\end{theorem}

\begin{proof}
Replacing $\boldsymbol{\theta}$ by $\boldsymbol{\theta}_{0}$ and the explicit
expression of $\boldsymbol{W}$\ in the general expression of the finite sample
size Fisher information matrix for two independent multinomial samples,
(\ref{FIM1}), we obtain through the property of the Kronecker product given in
(1.22) of Harville (2008, page 341) that%
\begin{align*}
\mathcal{I}_{F}^{(n_{1},...,n_{I})}(\boldsymbol{\theta}_{0})  &  =\left(
\begin{pmatrix}
\boldsymbol{1}_{I-1} & \boldsymbol{I}_{I-1}\\
1 & \boldsymbol{0}_{I-1}^{T}%
\end{pmatrix}
^{T}\otimes%
\begin{pmatrix}
\boldsymbol{I}_{J-1}\\
\boldsymbol{0}_{J-1}^{T}%
\end{pmatrix}
^{T}\right)  \left(  \left(
{\textstyle\bigoplus\limits_{i=1}^{I}}
\tfrac{n_{i}}{n}\right)  \otimes(\boldsymbol{D}_{\boldsymbol{\pi
}(\boldsymbol{\theta}_{0})}-\boldsymbol{\pi}(\boldsymbol{\theta}%
_{0})\boldsymbol{\pi}^{T}(\boldsymbol{\theta}_{0}))\right) \\
&  \times\left(
\begin{pmatrix}
\boldsymbol{1}_{I-1} & \boldsymbol{I}_{I-1}\\
1 & \boldsymbol{0}_{I-1}^{T}%
\end{pmatrix}
\otimes%
\begin{pmatrix}
\boldsymbol{I}_{J-1}\\
\boldsymbol{0}_{J-1}^{T}%
\end{pmatrix}
\right) \\
&  =\left(
\begin{pmatrix}
\boldsymbol{1}_{I-1} & \boldsymbol{I}_{I-1}\\
1 & \boldsymbol{0}_{I-1}^{T}%
\end{pmatrix}
^{T}\left(
{\textstyle\bigoplus\limits_{i=1}^{I}}
\tfrac{n_{i}}{n}\right)
\begin{pmatrix}
\boldsymbol{1}_{I-1} & \boldsymbol{I}_{I-1}\\
1 & \boldsymbol{0}_{I-1}^{T}%
\end{pmatrix}
\right)  \otimes\left(
\begin{pmatrix}
\boldsymbol{I}_{J-1}\\
\boldsymbol{0}_{J-1}^{T}%
\end{pmatrix}
^{T}(\boldsymbol{D}_{\boldsymbol{\pi}(\boldsymbol{\theta}_{0})}%
-\boldsymbol{\pi}(\boldsymbol{\theta}_{0})\boldsymbol{\pi}^{T}%
(\boldsymbol{\theta}_{0}))%
\begin{pmatrix}
\boldsymbol{I}_{J-1}\\
\boldsymbol{0}_{J-1}^{T}%
\end{pmatrix}
\right) \\
&  =%
\begin{pmatrix}
1 & (\tfrac{n_{1}}{n},...,\tfrac{n_{I-1}}{n})\\
(\tfrac{n_{1}}{n},...,\tfrac{n_{I-1}}{n})^{T} & \left(  \bigoplus_{i=1}%
^{I-1}\frac{n_{i}}{n}\right)
\end{pmatrix}
\otimes\left(  \boldsymbol{D}_{\boldsymbol{\pi}^{\ast}(\boldsymbol{\theta}%
_{0})}-\boldsymbol{\pi}^{\ast}(\boldsymbol{\theta}_{0})\boldsymbol{\pi}^{\ast
T}(\boldsymbol{\theta}_{0})\right)  ,
\end{align*}
and then (\ref{FIM2}).
\end{proof}

In the following theorem we present the asymptotic distribution of all of the
proposed test-statistics under the null hypothesis. Let
$E=\{1,...,(I-1)(J-1)\}$ the whole set of all row-indices of matrix
$\boldsymbol{R}$, $\mathcal{F}(E)$ the family of all possible subsets of $E$,
and $\boldsymbol{R}(S\mathbf{)}$ is a submatrix of $\boldsymbol{R}$\ with
row-indices belonging to $S\in\mathcal{F}(E)$.

\begin{theorem}
\label{Th1}Under $H_{0}$, the asymptotic distribution of $S_{\phi
}(\boldsymbol{p}(\widetilde{\boldsymbol{\theta}}),\boldsymbol{p}%
(\widehat{\boldsymbol{\theta}}))$\ and $T_{\phi}(\overline{\boldsymbol{p}%
},\boldsymbol{p}(\widetilde{\boldsymbol{\theta}}),\boldsymbol{p}%
(\widehat{\boldsymbol{\theta}}))$\ is%
\[
\lim_{n\rightarrow\infty}\Pr\left(  S_{\phi}(\boldsymbol{p}%
(\widetilde{\boldsymbol{\theta}}),\boldsymbol{p}(\widehat{\boldsymbol{\theta}%
}))\leq x\right)  =\lim_{n\rightarrow\infty}\Pr\left(  T_{\phi}(\overline
{\boldsymbol{p}},\boldsymbol{p}(\widetilde{\boldsymbol{\theta}}%
),\boldsymbol{p}(\widehat{\boldsymbol{\theta}}))\leq x\right)  =\sum
_{h=0}^{(I-1)(J-1)}w_{h}(\boldsymbol{\theta}_{0})\Pr\left(  \chi
_{(I-1)(J-1)-h}^{2}\leq x\right)
\]
where $\boldsymbol{\theta}_{0}$\ is the true value of the unknown parameter,
$\chi_{0}^{2}\equiv0$,%
\begin{equation}
w_{j}(\boldsymbol{\theta}_{0})=\sum_{S\in\mathcal{F}(E),\mathrm{card}(S)=h}%
\Pr\left(  \boldsymbol{Z}_{1}(S)\geq\boldsymbol{0}_{h}\right)  \Pr\left(
\boldsymbol{Z}_{2}(S)\geq\boldsymbol{0}_{(I-1)(J-1)-h}\right)  , \label{eqw}%
\end{equation}
$\boldsymbol{Z}_{1}(S)\sim\mathcal{N}\left(  \boldsymbol{0}_{\mathrm{card}%
(S)},\boldsymbol{\Sigma}_{1}(\boldsymbol{\theta}_{0},S)\right)  $,
$\boldsymbol{Z}_{2}(S)\sim\mathcal{N}\left(  \boldsymbol{0}%
_{(I-1)(J-1)-\mathrm{card}(S)},\boldsymbol{\Sigma}_{2}(\boldsymbol{\theta}%
_{0},S)\right)  $, with
\begin{equation}
\boldsymbol{\Sigma}_{1}(\boldsymbol{\theta}_{0},S)=\mathbf{H}^{-1}%
(S,S,\boldsymbol{\theta}_{0}), \label{matriz1}%
\end{equation}%
\begin{equation}
\boldsymbol{\Sigma}_{2}(\boldsymbol{\theta}_{0},S)=\mathbf{H}(S^{C}%
,S^{C},\boldsymbol{\theta}_{0})-\mathbf{H}(S^{C},S,\boldsymbol{\theta}%
_{0})\mathbf{H}^{-1}(S,S,\boldsymbol{\theta}_{0})\mathbf{H}^{T}(S^{C}%
,S,\boldsymbol{\theta}_{0}), \label{matriz2}%
\end{equation}
$S^{C}=E-S$, $\mathbf{H}(S_{1},S_{2},\boldsymbol{\theta}_{0})$ is the matrix
obtained deleting from $\boldsymbol{H}(\boldsymbol{\theta}_{0})$ the row
indices not contained in $S_{1}$, the column indices not contained in $S_{2}$,%
\begin{equation}
\boldsymbol{H}(\boldsymbol{\theta}_{0})=\boldsymbol{K}(\boldsymbol{\nu
})\otimes\boldsymbol{K}(\boldsymbol{\pi}(\boldsymbol{\theta}_{0})) \label{H}%
\end{equation}
is the $(I-1)(J-1)\times(I-1)(J-1)$ matrix which depends on the symmetric
tridiagonal matrices $\boldsymbol{K}(\boldsymbol{\nu})$ and $\boldsymbol{K}%
(\boldsymbol{\pi}(\boldsymbol{\theta}_{0}))$ defined as%
\begin{equation}
\boldsymbol{K}(\boldsymbol{q})=\boldsymbol{G}_{K-1}\boldsymbol{D}%
_{\boldsymbol{q}^{\ast}}^{-1}\boldsymbol{G}_{K-1}^{T}+\tfrac{1}{q_{K}%
}\boldsymbol{e}_{K-1}\boldsymbol{e}_{K-1}^{T}=%
\begin{pmatrix}
\frac{q_{1}+q_{2}}{q_{1}q_{2}} & -\frac{1}{q_{2}} &  &  & \\
-\frac{1}{q_{2}} & \frac{q_{2}+q_{3}}{q_{2}q_{3}} & -\frac{1}{q_{3}} &  & \\
& -\frac{1}{q_{3}} & \frac{q_{3}+q_{4}}{q_{3}q_{4}} & \ddots & \\
&  & \ddots & \ddots & -\frac{1}{q_{K-1}}\\
&  &  & -\frac{1}{q_{K-1}} & \frac{q_{K-1}+q_{K}}{q_{K-1}q_{K}}%
\end{pmatrix}
, \label{H-b}%
\end{equation}
where $\boldsymbol{q}^{\ast}=(q_{1},...,q_{K-1})^{T}$, the subvector of a
probability vector $\boldsymbol{q}=(q_{1},...,q_{K})^{T}$.
\end{theorem}

\begin{proof}
By following similar arguments of Mart\'{\i}n and Balakrishnan we obtain
$\boldsymbol{H}(S,S,\boldsymbol{\theta}_{0})=\boldsymbol{R}(S\mathbf{)}%
\mathcal{I}_{F}^{-1}(\boldsymbol{\theta}_{0})\boldsymbol{R}^{T}(S\mathbf{)}$
(see Section \ref{ProofTh1ContrA} of the Appendix). We shall here obtain the
expression of
\begin{align*}
\boldsymbol{H}(\boldsymbol{\theta}_{0})  &  =\boldsymbol{H}%
(E,E,\boldsymbol{\theta}_{0})=\boldsymbol{R}(E\mathbf{)}\mathcal{I}_{F}%
^{-1}(\boldsymbol{\theta}_{0})\boldsymbol{R}^{T}(E\mathbf{)}\\
&  \mathbf{=}\left(  \boldsymbol{e}_{J}\otimes\left(  \boldsymbol{G}%
_{I-1}\otimes\boldsymbol{G}_{J-1}\right)  \right)  \mathcal{I}_{F}%
^{-1}(\boldsymbol{\theta}_{0})\left(  \boldsymbol{e}_{J}^{T}\otimes\left(
\boldsymbol{G}_{I-1}\otimes\boldsymbol{G}_{J-1}\right)  ^{T}\right)  ,
\end{align*}
where%
\begin{align*}
\mathcal{I}_{F}^{-1}(\boldsymbol{\theta}_{0})  &  =%
\begin{pmatrix}
1 & \boldsymbol{\nu}^{\ast T}\\
\boldsymbol{\nu}^{\ast} & \bigoplus_{i=1}^{I-1}\nu_{i}%
\end{pmatrix}
^{-1}\otimes\left(  \boldsymbol{D}_{\boldsymbol{\pi}^{\ast}(\boldsymbol{\theta
}_{0})}-\boldsymbol{\pi}^{\ast}(\boldsymbol{\theta}_{0})\boldsymbol{\pi}^{\ast
T}(\boldsymbol{\theta}_{0})\right)  ^{-1}\\
&  =%
\begin{pmatrix}
\frac{1}{\nu_{I}} & -\frac{1}{\nu_{I}}\boldsymbol{1}_{I-1}^{T}\\
-\frac{1}{\nu_{I}}\boldsymbol{1}_{I-1} & \boldsymbol{D}_{\boldsymbol{\nu
}^{\ast}}^{-1}+\frac{1}{\nu_{I}}\boldsymbol{1}_{I-1}\boldsymbol{1}_{I-1}^{T}%
\end{pmatrix}
\otimes\left(  \boldsymbol{D}_{\boldsymbol{\pi}^{\ast}(\boldsymbol{\theta}%
_{0})}^{-1}+\frac{1}{\pi_{J}(\boldsymbol{\theta}_{0})}\boldsymbol{1}%
_{J-1}\boldsymbol{1}_{J-1}^{T}\right)  ,
\end{align*}
and then%
\begin{align*}
\boldsymbol{H}(\boldsymbol{\theta}_{0})  &  =\left(  \boldsymbol{e}_{J}%
\otimes\left(  \boldsymbol{G}_{I-1}\otimes\boldsymbol{G}_{J-1}\right)
\right)  \left(
\begin{pmatrix}
\frac{1}{\nu_{I}} & -\frac{1}{\nu_{I}}\boldsymbol{1}_{I-1}^{T}\\
-\frac{1}{\nu_{I}}\boldsymbol{1}_{I-1} & \boldsymbol{D}_{\boldsymbol{\nu
}^{\ast}}^{-1}+\frac{1}{\nu_{I}}\boldsymbol{1}_{I-1}\boldsymbol{1}_{I-1}^{T}%
\end{pmatrix}
\otimes\left(  \boldsymbol{D}_{\boldsymbol{\pi}^{\ast}(\boldsymbol{\theta}%
_{0})}^{-1}+\frac{1}{\pi_{J}(\boldsymbol{\theta}_{0})}\boldsymbol{1}%
_{J-1}\boldsymbol{1}_{J-1}^{T}\right)  \right) \\
&  \times\left(  \boldsymbol{e}_{J}^{T}\otimes\left(  \boldsymbol{G}%
_{I-1}\otimes\boldsymbol{G}_{J-1}\right)  ^{T}\right) \\
&  =\left(  \boldsymbol{G}_{I-1}\otimes\boldsymbol{G}_{J-1}\right)  \left(
\left(  \boldsymbol{D}_{\boldsymbol{\nu}^{\ast}}^{-1}+\frac{1}{\nu_{I}%
}\boldsymbol{1}_{I-1}\boldsymbol{1}_{I-1}^{T}\right)  \otimes\left(
\boldsymbol{D}_{\boldsymbol{\pi}^{\ast}(\boldsymbol{\theta}_{0})}^{-1}%
+\frac{1}{\pi_{J}(\boldsymbol{\theta}_{0})}\boldsymbol{1}_{J-1}\boldsymbol{1}%
_{J-1}^{T}\right)  \right)  \left(  \boldsymbol{G}_{I-1}^{T}\otimes
\boldsymbol{G}_{J-1}^{T}\right) \\
&  =\left(  \boldsymbol{G}_{I-1}\boldsymbol{D}_{\boldsymbol{\nu}^{\ast}}%
^{-1}\boldsymbol{G}_{I-1}^{T}+\frac{1}{\nu_{I}}\boldsymbol{e}_{I-1}%
\boldsymbol{e}_{I-1}^{T}\right)  \otimes\left(  \boldsymbol{G}_{J-1}%
\boldsymbol{D}_{\boldsymbol{\pi}^{\ast}(\boldsymbol{\theta}_{0})}%
^{-1}\boldsymbol{G}_{J-1}^{T}+\frac{1}{\pi_{J}(\boldsymbol{\theta}_{0}%
)}\boldsymbol{e}_{J-1}\boldsymbol{e}_{J-1}^{T}\right)  ,
\end{align*}
and thus it holds (\ref{H}).\medskip
\end{proof}

We must take into account that even thought there is an equality in (\ref{4b})
which is effective only for $\boldsymbol{\theta}_{12}$, the rest of the
components of $\boldsymbol{\theta}$ are nuisance parameters, and hence we have
a composite null hypothesis which require of estimation of $\boldsymbol{\theta
}$, through $\widehat{\boldsymbol{\theta}}$.

The following result is very useful in order to calculate the weights of the
chi-bar distribution by using simulation experiments.

\begin{corollary}
\label{Cor}Under $H_{0}$, the weights $w_{j}(\boldsymbol{\theta}_{0})$\ of the
asymptotic distribution of $S_{\phi}(\boldsymbol{p}%
(\widetilde{\boldsymbol{\theta}}),\boldsymbol{p}(\widehat{\boldsymbol{\theta}%
}))$\ and $T_{\phi}(\overline{\boldsymbol{p}},\boldsymbol{p}%
(\widetilde{\boldsymbol{\theta}}),\boldsymbol{p}(\widehat{\boldsymbol{\theta}%
}))$, given in Theorem \ref{Th1}, can be expressed as%
\begin{align}
w_{(I-1)(J-1)-j}(\boldsymbol{\theta}_{0})  &  =w_{j}(\boldsymbol{\theta}%
_{0};(I-1)(J-1),\boldsymbol{H}^{-1}(\boldsymbol{\theta}_{0}),%
\mathbb{R}
_{+}^{(I-1)(J-1)})\label{pesos}\\
&  =\Pr\left(  \arg\min_{\boldsymbol{\zeta\in}%
\mathbb{R}
_{+}^{(I-1)(J-1)}}(\boldsymbol{Z}-\boldsymbol{\zeta})^{T}\boldsymbol{H}%
(\boldsymbol{\theta}_{0})(\boldsymbol{Z}-\boldsymbol{\zeta})\in%
\mathbb{R}
_{+}^{(I-1)(J-1)}(j)\right)  ,\nonumber
\end{align}
with $\boldsymbol{H}(\boldsymbol{\theta}_{0})$\ given by (\ref{H}) and%
\begin{equation}
\boldsymbol{H}^{-1}(\boldsymbol{\theta}_{0})=\boldsymbol{K}^{-1}%
(\boldsymbol{\nu})\otimes\boldsymbol{K}^{-1}(\boldsymbol{\pi}%
(\boldsymbol{\theta}_{0})), \label{H-1}%
\end{equation}
which depends on%
\begin{equation}
\boldsymbol{K}^{-1}(\boldsymbol{q})=\boldsymbol{T}_{K-1}^{T}\left(
\boldsymbol{D}_{\boldsymbol{q}^{\ast}}-\boldsymbol{q}^{\ast}\boldsymbol{q}%
^{\ast T}\right)  \boldsymbol{T}_{K-1}, \label{H-1-b}%
\end{equation}
$\boldsymbol{T}_{h}=\boldsymbol{G}_{h}^{-1}$ is an upper triangular matrix of
$1$-s, $\boldsymbol{Z\sim}\mathcal{N}_{(I-1)(J-1)}\left(  \boldsymbol{0}%
_{(I-1)(J-1)},\boldsymbol{H}^{-1}(\boldsymbol{\theta}_{0})\right)  $, and $%
\mathbb{R}
_{+}^{(I-1)(J-1)}(j)$\ is the subset of $%
\mathbb{R}
_{+}^{(I-1)(J-1)}=\{\boldsymbol{\zeta}\in%
\mathbb{R}
^{(I-1)(J-1)}:\boldsymbol{\zeta}\geq\boldsymbol{0}_{(I-1)(J-1)}\}$, such that
$j$\ components of the $(I-1)(J-1)$-dimensional\ vectors are strictly positive
and $(I-1)(J-1)-j$ components are null.
\end{corollary}

\begin{proof}
It is well known that the weights of a chi-bar distribution can be interpreted
in terms of the projection of a $p$-dimensional central normal distribution
$\boldsymbol{Z}_{p}$ with a non-singular variance-covariance matrix
$\boldsymbol{V}$, on a closed convex cone in $%
\mathbb{R}
^{p}$, $C$, as $w_{j}(p,\boldsymbol{V},C)=\Pr(\Pi(\boldsymbol{Z}_{p}|C)\in%
\mathbb{R}
_{+}^{p}(j))$\ where%
\[
\Pi(\boldsymbol{Z}_{p}|C)=\arg\min_{\boldsymbol{\zeta\in}C}(\boldsymbol{Z}%
_{p}-\boldsymbol{\zeta})^{T}\boldsymbol{V}^{-1}(\boldsymbol{Z}_{p}%
-\boldsymbol{\zeta}).
\]
Now, focussed on $w_{j}(\boldsymbol{\theta}_{0})=w_{j}(\boldsymbol{\theta}%
_{0};p,\boldsymbol{V},C)$ in Theorem \ref{Th1}, we must identify the value of
$p$, expression of matrix $\boldsymbol{V}$ and the set $C$. In Kud\^{o} (1963,
p.414) and also in Shapiro (1988, p.54) it is shown that
\[
w_{j}(p,\boldsymbol{V},%
\mathbb{R}
_{+}^{p})=\sum_{S\in\mathcal{F}(F),\mathrm{card}(S)=p-j}\Pr\left(
\boldsymbol{Z}_{1,p}(S)\geq\boldsymbol{0}_{j}\right)  \Pr\left(
\boldsymbol{Z}_{2,p}(S)\geq\boldsymbol{0}_{(I-1)(J-1)-j}\right)  ,
\]
where $F=\{1,...,p\}$, $\boldsymbol{Z}_{1,p}(S)\boldsymbol{\sim}%
\mathcal{N}_{\mathrm{card}(S)}\left(  \boldsymbol{0}_{\mathrm{card}%
(S)},\boldsymbol{V}^{-1}(S)\right)  $ and $\boldsymbol{Z}_{2,p}%
(S)\boldsymbol{\sim}\mathcal{N}_{\mathrm{card}(F)-\mathrm{card}(S)}%
(\boldsymbol{0}_{\mathrm{card}(F)-\mathrm{card}(S)}$,\linebreak$\boldsymbol{V}%
(F-S,S))$, where $\boldsymbol{V}(S)$\ is the variance-covariance matrix of the
random vector obtained by considering only from $\boldsymbol{Z}_{p}$ the
components belonging to $S$ and $\boldsymbol{V}(S,F-S)$ is the same but rather
than ignoring the components out from $S$ they are considered equals zero.
This enunciate can be also seen in Silvapulle and Sen (2005, page 83). Note
that we can identify $p=(I-1)(J-1)$, $\boldsymbol{V=H}(\boldsymbol{\theta}%
_{0})$ and $C=%
\mathbb{R}
_{+}^{(I-1)(J-1)}$\ and
\begin{align*}
&  w_{(I-1)(J-1)-j}(\boldsymbol{\theta}_{0};(I-1)(J-1),\boldsymbol{H}%
(\boldsymbol{\theta}_{0}),%
\mathbb{R}
_{+}^{(I-1)(J-1)})\\
&  =\sum_{S\in\mathcal{F}(F),\mathrm{card}(S)=j}\Pr\left(  \boldsymbol{Z}%
_{1,J-1}(S)\geq\boldsymbol{0}_{j}\right)  \Pr\left(  \boldsymbol{Z}%
_{2,J-1}(S)\geq\boldsymbol{0}_{(I-1)(J-1)-j}\right)  ,
\end{align*}
which is equal to
\begin{align*}
&  w_{j}(\boldsymbol{\theta}_{0};(I-1)(J-1),\boldsymbol{H}^{-1}%
(\boldsymbol{\theta}_{0}),%
\mathbb{R}
_{+}^{(I-1)(J-1)})\\
&  =\sum_{S\in\mathcal{F}(F),\mathrm{card}(S)=j}\Pr\left(  \boldsymbol{Z}%
_{1,J-1}(S)\geq\boldsymbol{0}_{j}\right)  \Pr\left(  \boldsymbol{Z}%
_{2,J-1}(S)\geq\boldsymbol{0}_{(I-1)(J-1)-j}\right)  ,
\end{align*}
according to Proposition 3.6.1(7) in Silvapulle and Sen (2005, page 82). This
expression match (\ref{eqw}).\medskip
\end{proof}

Since $\boldsymbol{\theta}_{0}$\ is unknown, we cannot use directly the
results based on Theorem \ref{Th1} or Corollary \ref{Cor}. However, the
unknown parameter $\boldsymbol{\theta}_{0}$\ can be replaced by its estimator
under the null hypothesis, $\widehat{\boldsymbol{\theta}}$. The tests
performed replacing $\boldsymbol{\theta}_{0}$\ by $\widehat{\boldsymbol{\theta
}}$ are called \textquotedblleft local tests\textquotedblright\ (see Dardanoni
and Forcina (1998)) and they are usually considered to be good approximations
of the theoretical tests. It is worthwhile to mention that $\boldsymbol{p}%
(\widehat{\boldsymbol{\theta}})$ has an explicit expression,%
\begin{equation}
p_{ij}(\widehat{\boldsymbol{\theta}})=\widehat{\nu}_{i}\pi_{j}%
(\widehat{\boldsymbol{\theta}}),\qquad\widehat{\nu}_{i}=\frac{n_{i}}{n}%
,\qquad\pi_{j}(\widehat{\boldsymbol{\theta}})=\frac{1}{n}N_{\bullet j},\qquad
N_{\bullet j}=\sum_{h=1}^{I}N_{hj}. \label{ind}%
\end{equation}
Now, based on Corollary \ref{Cor}, and taking into account that (\ref{pesos})
is equal to%
\begin{equation}
w_{j}(\boldsymbol{\theta}_{0})=\Pr\left(  \arg\min_{\boldsymbol{\zeta\in}%
\mathbb{R}
_{+}^{(I-1)(J-1)}}\tfrac{1}{2}\boldsymbol{\zeta}^{T}\boldsymbol{H}%
(\boldsymbol{\theta}_{0})\boldsymbol{\zeta}-\left(  \boldsymbol{H}%
(\boldsymbol{\theta}_{0})\boldsymbol{Z}\right)  ^{T}\boldsymbol{\zeta}\in%
\mathbb{R}
_{+}^{(I-1)(J-1)}(j)\right)  , \label{pesos2}%
\end{equation}
where $\boldsymbol{H}(\boldsymbol{\theta}_{0})$ is (\ref{H}) and
$\boldsymbol{Z\sim}\mathcal{N}_{(I-1)(J-1)}\left(  \boldsymbol{0}%
_{(I-1)(J-1)},\boldsymbol{H}^{-1}(\boldsymbol{\theta}_{0})\right)  $, we shall
consider$\ $an algorithm for obtaining the weights associated with a sample.

\begin{algorithm}
[Estimation of weights]\label{AlgorW}The weights of the local tests,
$w_{j}(\widehat{\boldsymbol{\theta}})$, are obtained by Monte Carlo, once we
have a realization $\boldsymbol{n}$ of $\boldsymbol{N}$ in the following way

\noindent\texttt{STEP 1: Using }$\boldsymbol{n}$,\texttt{\ calculate
}$\boldsymbol{\nu}$\texttt{ and }$\boldsymbol{\pi}(\widehat{\boldsymbol{\theta
}})$\texttt{\ taking into account (\ref{ind}).}\newline\texttt{STEP 2: Compute
}$\boldsymbol{H}(\widehat{\boldsymbol{\theta}})$\texttt{ by following
(\ref{H}), in terms of }$\boldsymbol{K}(\widehat{\boldsymbol{\nu}})$,
$\boldsymbol{K}(\boldsymbol{\pi}(\widehat{\boldsymbol{\theta}}))$, given by
(\ref{H-b}))\texttt{.}\newline\texttt{STEP 3: Compute }$\boldsymbol{H}%
^{-1}(\widehat{\boldsymbol{\theta}})$\texttt{ by following (\ref{H-1-b}), in
terms of }$\boldsymbol{K}^{-1}(\widehat{\boldsymbol{\nu}})$, $\boldsymbol{K}%
^{-1}(\boldsymbol{\pi}(\widehat{\boldsymbol{\theta}}))$, given by
(\ref{H-1-b})\texttt{.\newline STEP 4: For }$j=0,...,(I-1)(J-1)$\texttt{, do
}$N(j):=0$\texttt{.}\newline\texttt{STEP 5: Repeat the following steps }%
$R$\texttt{ (say }$R=1,000,000$\texttt{) times:}\newline\hspace*{0.75cm}%
\texttt{STEP 5.1: Generate an observation, }$\boldsymbol{z}$\texttt{, from
}$\boldsymbol{Z\sim}\mathcal{N}_{(I-1)(J-1)}\left(  \boldsymbol{0}%
_{(I-1)(J-1)},\boldsymbol{H}^{-1}(\boldsymbol{\theta}_{0})\right)  $\texttt{.
E.g., the\newline\hspace*{1cm}\hspace*{1cm}\hspace*{0.75cm}NAG Fortran library
subroutines G05CBF, G05EAF, and G05EZF can be useful.\newline\hspace
*{0.75cm}STEP 5.2: Compute }$\widehat{\boldsymbol{\zeta}}(\boldsymbol{z}%
)\boldsymbol{=}\arg\min_{\boldsymbol{\zeta\in}%
\mathbb{R}
_{+}^{(I-1)(J-1)}}\tfrac{1}{2}\boldsymbol{\zeta}^{T}\boldsymbol{H}%
(\widehat{\boldsymbol{\theta}})\boldsymbol{\zeta}-(\boldsymbol{H}%
(\widehat{\boldsymbol{\theta}})\boldsymbol{z})^{T}\boldsymbol{\zeta}$.\texttt{
E.g., the IMSL Fortran \newline\hspace*{1cm}\hspace*{1cm}\hspace
*{0.75cm}library subroutine DQPROG can be useful.\newline\hspace*{0.75cm}STEP
5.3: Count }$j^{\ast}$\texttt{, the number of strictly positive components
contained in }$\widehat{\boldsymbol{\zeta}}(\boldsymbol{z})$\texttt{, and
\newline\hspace*{1cm}\hspace*{1cm}\hspace*{0.75cm}do }$N(j^{\ast}):=N(j^{\ast
})+1$.\newline\texttt{STEP 6: Do }$w_{j}(\widehat{\boldsymbol{\theta}}%
):=\frac{N(j)}{R}$ for $j=0,...,(I-1)(J-1)$.
\end{algorithm}

See
\hyperref{http://www.nag.co.uk/numeric/fl/FLdescription.asp}{}{}%
{http://www.nag.co.uk/numeric/fl/FLdescription.asp}%
, for details about subroutines of the\ \texttt{NAG} Fortran library, and
\hyperref{http://www.roguewave.com/Portals/0/products/imsl-numerical-libraries/fortran-library/docs/7.0/math/math.htm}%
{}{}%
{http://www.roguewave.com/Portals/0/products/imsl-numerical-libraries/fortran-library/docs/7.0/\allowbreak
math/math.htm}
for the\ \texttt{IMSL} Fortran library. It is worthwhile to mention that these
values can be also computed using \texttt{mvtnorm} R package (see
\hyperref{http://CRAN.R-project.org/package=mvtnorm}{}{}%
{http://CRAN.R-project.org/package=mvtnorm}%
, for details), however this method based on numerical integration tends to
provide less accurate values.

\section{Example\label{Ex}}

In this section we are going to analyze the data set of Section \ref{Intro}%
,\ using the proposed test-statistics. By following the specific notation of
our paper, we are considering two ordinal variables associated with $n=417$
duodenal ulcer patients in a hospital, $X=$severity of the operation,
classified in an increasing order from $1$ to $I=4$, and $Y=$extent of side
effects, categorized as None ($1$), Slight ($2$) and Moderate ($J=3$). The
sample, a realization of $\boldsymbol{N}$, is summarized in%
\begin{align*}
\boldsymbol{n}  &  =(n_{11},n_{12},n_{13},n_{21},n_{22},n_{23},n_{31}%
,n_{32},n_{33},n_{41},n_{42},n_{43})^{T}\\
&  =(61,28,7,68,23,13,58,40,12,53,38,16)^{T}.
\end{align*}

The order restricted maximum likelihood estimator (MLE) of $\boldsymbol{\theta
}=(\boldsymbol{\theta}_{2}^{T},\boldsymbol{\theta}_{12}^{T})^{T}$\ under
likelihood ratio order, obtained through \texttt{E04UCF} subroutine
of\ \texttt{NAG} Fortran library (%
\hyperref{http://www.nag.co.uk/numeric/fl/FLdescription.asp}{}{}%
{http://www.nag.co.uk/numeric/fl/FLdescription.asp}%
), is $\widetilde{\boldsymbol{\theta}}=(\widetilde{\boldsymbol{\theta}}%
_{2}^{T},\widetilde{\boldsymbol{\theta}}_{12}^{T})^{T}$, with%
\[
\widetilde{\boldsymbol{\theta}}_{2}=(1.1977,0.8650)^{T},\quad
\widetilde{\boldsymbol{\theta}}_{12}%
=(0.9983,0.4501,0.6376,0.0894,0.1916,0.0894)^{T}.
\]
The estimated probability vectors of interest are%
\begin{align*}
\widehat{\boldsymbol{\nu}}  &  =(\tfrac{n_{1}}{n},\tfrac{n_{2}}{n}%
,\tfrac{n_{3}}{n},\tfrac{n_{4}}{n})^{T}=(\tfrac{96}{417},\tfrac{104}%
{417},\tfrac{110}{417},\tfrac{107}{417})^{T},\\
\boldsymbol{\pi}(\widehat{\boldsymbol{\theta}})  &  =(\tfrac{n_{\bullet1}}%
{n},\tfrac{n_{\bullet2}}{n},\tfrac{n_{\bullet3}}{n})^{T}=(\tfrac{240}%
{417},\tfrac{129}{417},\tfrac{48}{417})^{T},
\end{align*}%
\begin{align*}
\overline{\boldsymbol{p}}  &
=(0.1463,0.0671,0.0168,0.1631,0.0552,0.0312,0.1391,0.0959,0.0288,0.1271,0.0911,0.0384)^{T}%
,\\
\boldsymbol{p}(\widetilde{\boldsymbol{\theta}})  &
=(0.1509,0.0625,0.0168,0.1585,0.0657,0.0253,0.1391,0.0900,0.0347,0.1271,0.0911,0.0384)^{T}%
,\\
\boldsymbol{p}(\widehat{\boldsymbol{\theta}})  &
=(0.1325,0.0712,0.0265,0.1435,0.0772,0.0287,0.1518,0.0816,0.0304,0.1477,0.0794,0.0295)^{T}%
,
\end{align*}
and the weights%
\begin{equation}
\{w_{j}(\widehat{\boldsymbol{\theta}})\}_{j=0}^{6}%
=\{0.0006103103,0.009753533,0.06122672,0.1953851,0.3353725,0.2949007,0.1028136\},
\label{ww}%
\end{equation}
were obtained using Theorem \ref{Th1} and the \texttt{R} package
\texttt{mvtnorm} (located at
\hyperref{http://cran.r-project.org/web/packages/mvtnorm/index.html}{}%
{}{\texttt{http://cran.r-project.org/web/\allowbreak
packages/mvtnorm/index.html}}%
), once we knew%
\begin{align*}
\boldsymbol{K}(\widehat{\boldsymbol{\nu}})  &  =%
\begin{pmatrix}
n\frac{n_{1}+n_{2}}{n_{1}n_{2}} & -\frac{n}{n_{2}} & 0\\
-\frac{n}{n_{2}} & n\frac{n_{2}+n_{3}}{n_{2}n_{3}} & -\frac{n}{n_{3}}\\
0 & -\frac{n}{n_{3}} & n\frac{n_{3}+n_{4}}{n_{3}n_{4}}%
\end{pmatrix}
=%
\begin{pmatrix}
\frac{3475}{416} & -\frac{417}{104} & 0\\
-\frac{417}{104} & \frac{44\,619}{5720} & -\frac{417}{110}\\
0 & -\frac{417}{110} & \frac{90\,489}{11\,770}%
\end{pmatrix}
,\\
\boldsymbol{K}(\boldsymbol{\pi}(\widehat{\boldsymbol{\theta}}))  &  =%
\begin{pmatrix}
n\frac{n_{\bullet1}+n_{\bullet2}}{n_{\bullet1}n_{\bullet2}} & -\frac
{n}{n_{\bullet2}}\\
-\frac{n}{n_{\bullet2}} & n\frac{n_{\bullet2}+n_{\bullet3}}{n_{\bullet
2}n_{\bullet3}}%
\end{pmatrix}
=%
\begin{pmatrix}
\frac{17\,097}{3440} & -\frac{417}{129}\\
-\frac{417}{129} & \frac{8201}{688}%
\end{pmatrix}
,\\
\boldsymbol{H}(\widehat{\boldsymbol{\theta}})  &  =\boldsymbol{K}%
(\widehat{\boldsymbol{\nu}})\otimes\boldsymbol{K}(\boldsymbol{\pi
}(\widehat{\boldsymbol{\theta}}))=%
\begin{pmatrix}
\frac{11\,882\,415}{286\,208} & -\frac{483\,025}{17\,888} & -\frac
{7129\,449}{357\,760} & \frac{57\,963}{4472} & 0 & 0\\
-\frac{483\,025}{17\,888} & \frac{28\,498\,475}{286\,208} & \frac
{57\,963}{4472} & -\frac{3419\,817}{71\,552} & 0 & 0\\
-\frac{7129\,449}{357\,760} & \frac{57\,963}{4472} & \frac{762\,851\,043}%
{19\,676\,800} & -\frac{6202\,041}{245\,960} & -\frac{7129\,449}{378\,400} &
\frac{57\,963}{4730}\\
\frac{57\,963}{4472} & -\frac{3419\,817}{71\,552} & -\frac{6202\,041}%
{245\,960} & \frac{365\,920\,419}{3935\,360} & \frac{57\,963}{4730} &
-\frac{3419\,817}{75\,680}\\
0 & 0 & -\frac{7129\,449}{378\,400} & \frac{57\,963}{4730} & \frac
{1547\,090\,433}{40\,488\,800} & -\frac{12\,577\,971}{506\,110}\\
0 & 0 & \frac{57\,963}{4730} & -\frac{3419\,817}{75\,680} & -\frac
{12\,577\,971}{506\,110} & \frac{742\,100\,289}{8097\,760}%
\end{pmatrix}
,
\end{align*}
Along the current section, we are trying to express the matrices as precise as
possible in order to highlight that the proposed method provide very simple
accurate way for obtaining the weights even for big dimensions. In the
notation we understand that $n_{\bullet j}$, $j=1,2,3$, are realizations of
the r.v. $N_{\bullet j}$ defined in (\ref{ind}).The output of the code based
on such package, provides normal orthant probabilities based on numerical
integration, as well as the precision error. Taking into account Proposition
3.6.1(3) in Silvapulle and Sen (2005, page 82), $%
{\textstyle\sum_{i=0}^{6}}
(-1)^{i}w_{i}(\widehat{\boldsymbol{\theta}})=0$ should be held theoretically,
and for (\ref{ww}) we obtained $%
{\textstyle\sum_{i=0}^{6}}
(-1)^{i}w_{i}(\widehat{\boldsymbol{\theta}})=\,\allowbreak-1.6203\times
10^{-5}$, and this means that $1.6203\times10^{-5}$ could be considered as an
overall measure of precision error for the weights. Using Algorithm
\ref{AlgorW} and taking into account%
\begin{align*}
\boldsymbol{K}^{-1}(\widehat{\boldsymbol{\nu}})  &  =\boldsymbol{T}_{3}%
^{T}\left(  \boldsymbol{D}_{\widehat{\nu}^{\ast}}-\widehat{\boldsymbol{\nu}%
}^{\ast}\widehat{\boldsymbol{\nu}}^{\ast T}\right)  \boldsymbol{T}_{3}=%
\begin{pmatrix}
\frac{3424}{19\,321} & \frac{6944}{57\,963} & \frac{3424}{57\,963}\\
\frac{6944}{57\,963} & \frac{43\,400}{173\,889} & \frac{21\,400}{173\,889}\\
\frac{3424}{57\,963} & \frac{21\,400}{173\,889} & \frac{33\,170}{173\,889}%
\end{pmatrix}
\\
\boldsymbol{K}^{-1}(\boldsymbol{\pi}(\widehat{\boldsymbol{\theta}}))  &
=\boldsymbol{T}_{2}^{T}\left(  \boldsymbol{D}_{\boldsymbol{\pi}^{\ast
}(\widehat{\boldsymbol{\theta}})}-\boldsymbol{\pi}^{\ast}%
(\widehat{\boldsymbol{\theta}})\boldsymbol{\pi}^{\ast T}%
(\widehat{\boldsymbol{\theta}})\right)  \boldsymbol{T}_{2}=%
\begin{pmatrix}
\frac{4720}{19\,321} & \frac{1280}{19\,321}\\
\frac{1280}{19\,321} & \frac{1968}{19\,321}%
\end{pmatrix}
\\
\boldsymbol{H}^{-1}(\widehat{\boldsymbol{\theta}})  &  =\boldsymbol{K}%
^{-1}(\widehat{\boldsymbol{\nu}}^{\ast})\otimes\boldsymbol{K}^{-1}%
(\boldsymbol{\pi}(\widehat{\boldsymbol{\theta}}))=\\
&  =%
\begin{pmatrix}
\frac{16\,161\,280}{373\,301\,041} & \frac{4382\,720}{373\,301\,041} &
\frac{32\,775\,680}{1119\,903\,123} & \frac{8888\,320}{1119\,903\,123} &
\frac{16\,161\,280}{1119\,903\,123} & \frac{4382\,720}{1119\,903\,123}\\
\frac{4382\,720}{373\,301\,041} & \frac{6738\,432}{373\,301\,041} &
\frac{8888\,320}{1119\,903\,123} & \frac{4555\,264}{373\,301\,041} &
\frac{4382\,720}{1119\,903\,123} & \frac{2246\,144}{373\,301\,041}\\
\frac{32\,775\,680}{1119\,903\,123} & \frac{8888\,320}{1119\,903\,123} &
\frac{204\,848\,000}{3359\,709\,369} & \frac{55\,552\,000}{3359\,709\,369} &
\frac{101\,008\,000}{3359\,709\,369} & \frac{27\,392\,000}{3359\,709\,369}\\
\frac{8888\,320}{1119\,903\,123} & \frac{4555\,264}{373\,301\,041} &
\frac{55\,552\,000}{3359\,709\,369} & \frac{28\,470\,400}{1119\,903\,123} &
\frac{27\,392\,000}{3359\,709\,369} & \frac{14\,038\,400}{1119\,903\,123}\\
\frac{16\,161\,280}{1119\,903\,123} & \frac{4382\,720}{1119\,903\,123} &
\frac{101\,008\,000}{3359\,709\,369} & \frac{27\,392\,000}{3359\,709\,369} &
\frac{156\,562\,400}{3359\,709\,369} & \frac{42\,457\,600}{3359\,709\,369}\\
\frac{4382\,720}{1119\,903\,123} & \frac{2246\,144}{373\,301\,041} &
\frac{27\,392\,000}{3359\,709\,369} & \frac{14\,038\,400}{1119\,903\,123} &
\frac{42\,457\,600}{3359\,709\,369} & \frac{21\,759\,520}{1119\,903\,123}%
\end{pmatrix}
,
\end{align*}
very similar weights were obtained: $\{w_{j}(\widehat{\boldsymbol{\theta}%
})\}_{j=0}^{6}=\{0.000613$, $0.009627$, $0.060873$, $0.195312$, $0.335389$,
$0.295527$, $0.102659\}$. From these weights the quantile of order $0.05$,
which defines the rejection region, was find to be $6.34$.

If we take, for (\ref{div}), $\phi_{\lambda}(x)=\frac{1}{\lambda(1+\lambda
)}(x^{\lambda+1}-x-\lambda(x-1))$, where for each $\lambda\in%
\mathbb{R}
-\{-1,0\}$ a different divergence measure is constructed, a very important
subfamily called \textquotedblleft power divergence family of
measures\textquotedblright\ is obtained%
\begin{equation}
d_{\lambda}(\boldsymbol{p},\boldsymbol{q})=\frac{1}{\lambda(\lambda+1)}\left(
%
{\displaystyle\sum\limits_{i=1}^{I}}
{\displaystyle\sum\limits_{j=1}^{J}}
\frac{p_{ij}^{\lambda+1}}{q_{ij}^{\lambda}}-1\right)  \text{, for each
}\lambda\in%
\mathbb{R}
-\{-1,0\}\text{.} \label{CR}%
\end{equation}
It is also possible to cover the real line for $\lambda$, by defining
$d_{\lambda}(\boldsymbol{p},\boldsymbol{q})=\lim_{t\rightarrow\lambda}%
d_{t}(\boldsymbol{p},\boldsymbol{q})$, for $\lambda\in\{-1,0\}$. It is well
known that $d_{0}(\boldsymbol{p},\boldsymbol{q})=d_{Kull}(\boldsymbol{p}%
,\boldsymbol{q})$ and $d_{1}(\boldsymbol{p},\boldsymbol{q})=d_{Pearson}%
(\boldsymbol{p},\boldsymbol{q})$, which is very interesting because the power
divergence based family of test-statistics, which contains as special cases
$G^{2}$ and $X^{2}$, can be created. It is also worthwhile to mention that
$d_{-1}(\boldsymbol{p},\boldsymbol{q})=d_{Kull}(\boldsymbol{q},\boldsymbol{p}%
)$.

When the test-statistic (\ref{5a}) and (\ref{5b}), based on power-divergences
(\ref{CR}),\ are applied we get%
\begin{equation}
T_{\lambda}(\overline{\boldsymbol{p}},\boldsymbol{p}%
(\widetilde{\boldsymbol{\theta}}),\boldsymbol{p}(\widehat{\boldsymbol{\theta}%
}))=2n(d_{\lambda}(\overline{\boldsymbol{p}},\boldsymbol{p}%
(\widehat{\boldsymbol{\theta}}))-d_{\lambda}(\overline{\boldsymbol{p}%
},\boldsymbol{p}(\widetilde{\boldsymbol{\theta}})))=\frac{2n}{\lambda
(\lambda+1)}\left(
{\displaystyle\sum\limits_{i=1}^{I}}
{\displaystyle\sum\limits_{j=1}^{J}}
\frac{\overline{p}_{ij}^{\lambda+1}}{p_{ij}^{\lambda}%
(\widehat{\boldsymbol{\theta}})}-%
{\displaystyle\sum\limits_{i=1}^{I}}
{\displaystyle\sum\limits_{j=1}^{J}}
\frac{\overline{p}_{ij}^{\lambda+1}}{p_{ij}^{\lambda}%
(\widetilde{\boldsymbol{\theta}})}\right)  \label{PD1}%
\end{equation}
and%
\begin{equation}
S_{\lambda}(\boldsymbol{p}(\widetilde{\boldsymbol{\theta}}),\boldsymbol{p}%
(\widehat{\boldsymbol{\theta}}))=2nd_{\lambda}(\boldsymbol{p}%
(\widetilde{\boldsymbol{\theta}}),\boldsymbol{p}(\widehat{\boldsymbol{\theta}%
}))=\frac{2n}{\lambda(\lambda+1)}\left(
{\displaystyle\sum\limits_{i=1}^{I}}
{\displaystyle\sum\limits_{j=1}^{J}}
\frac{p_{ij}^{\lambda+1}(\widetilde{\boldsymbol{\theta}})}{p_{ij}^{\lambda
}(\widehat{\boldsymbol{\theta}})}-1\right)  , \label{PD2}%
\end{equation}
for $\lambda\in%
\mathbb{R}
-\{0,-1\}$, and $T_{\lambda}(\overline{\boldsymbol{p}},\boldsymbol{p}%
(\widetilde{\boldsymbol{\theta}}),\boldsymbol{p}(\widehat{\boldsymbol{\theta}%
}))=\lim_{\lambda\rightarrow\ell}T_{\ell}(\overline{\boldsymbol{p}%
},\boldsymbol{p}(\widetilde{\boldsymbol{\theta}}),\boldsymbol{p}%
(\widehat{\boldsymbol{\theta}}))$, $S_{\lambda}(\boldsymbol{p}%
(\widetilde{\boldsymbol{\theta}}),\boldsymbol{p}(\widehat{\boldsymbol{\theta}%
}))=\lim_{\lambda\rightarrow\ell}S_{\ell}(\boldsymbol{p}%
(\widetilde{\boldsymbol{\theta}}),\boldsymbol{p}(\widehat{\boldsymbol{\theta}%
}))$, $\lambda\in\{0,-1\}$, i.e.%
\begin{align}
T_{0}(\overline{\boldsymbol{p}},\boldsymbol{p}(\widetilde{\boldsymbol{\theta}%
}),\boldsymbol{p}(\widehat{\boldsymbol{\theta}}))  &  =2n(d_{Kull}%
(\overline{\boldsymbol{p}},\boldsymbol{p}(\widehat{\boldsymbol{\theta}%
}))-d_{Kull}(\overline{\boldsymbol{p}},\boldsymbol{p}%
(\widetilde{\boldsymbol{\theta}})))=2n%
{\displaystyle\sum\limits_{i=1}^{I}}
{\displaystyle\sum\limits_{j=1}^{J}}
\overline{p}_{ij}\log\left(  \frac{p_{ij}(\widetilde{\boldsymbol{\theta}}%
)}{p_{ij}(\widehat{\boldsymbol{\theta}})}\right)  ,\label{PD3}\\
T_{-1}(\overline{\boldsymbol{p}},\boldsymbol{p}(\widetilde{\boldsymbol{\theta
}}),\boldsymbol{p}(\widehat{\boldsymbol{\theta}}))  &  =2n(d_{Kull}%
(\boldsymbol{p}(\widehat{\boldsymbol{\theta}}),\overline{\boldsymbol{p}%
})-d_{Kull}(\boldsymbol{p}(\widetilde{\boldsymbol{\theta}}),\overline
{\boldsymbol{p}}))\nonumber\\
&  =2n\left(
{\displaystyle\sum\limits_{i=1}^{I}}
{\displaystyle\sum\limits_{j=1}^{J}}
p_{ij}(\widehat{\boldsymbol{\theta}})\log\left(  \frac{p_{ij}%
(\widehat{\boldsymbol{\theta}})}{\overline{p}_{ij}}\right)  -%
{\displaystyle\sum\limits_{i=1}^{I}}
{\displaystyle\sum\limits_{j=1}^{J}}
p_{ij}(\widetilde{\boldsymbol{\theta}})\log\left(  \frac{p_{ij}%
(\widetilde{\boldsymbol{\theta}})}{\overline{p}_{ij}}\right)  \right)
\label{PD4}%
\end{align}
and%
\begin{align}
S_{0}(\boldsymbol{p}(\widetilde{\boldsymbol{\theta}}),\boldsymbol{p}%
(\widehat{\boldsymbol{\theta}}))  &  =2nd_{Kull}(\boldsymbol{p}%
(\widetilde{\boldsymbol{\theta}}),\boldsymbol{p}(\widehat{\boldsymbol{\theta}%
}))=2n%
{\displaystyle\sum\limits_{i=1}^{I}}
{\displaystyle\sum\limits_{j=1}^{J}}
p_{ij}(\widetilde{\boldsymbol{\theta}})\log\left(  \frac{p_{ij}%
(\widetilde{\boldsymbol{\theta}})}{p_{ij}(\widehat{\boldsymbol{\theta}}%
)}\right)  ,\label{PD5}\\
S_{-1}(\boldsymbol{p}(\widetilde{\boldsymbol{\theta}}),\boldsymbol{p}%
(\widehat{\boldsymbol{\theta}}))  &  =2nd_{Kull}(\boldsymbol{p}%
(\widehat{\boldsymbol{\theta}}),\boldsymbol{p}(\widetilde{\boldsymbol{\theta}%
}))=2n%
{\displaystyle\sum\limits_{i=1}^{I}}
{\displaystyle\sum\limits_{j=1}^{J}}
p_{ij}(\widehat{\boldsymbol{\theta}})\log\left(  \frac{p_{ij}%
(\widehat{\boldsymbol{\theta}})}{p_{ij}(\widetilde{\boldsymbol{\theta}}%
)}\right)  . \label{PD6}%
\end{align}
Suppose we want to consider a set of values for the parameter $\lambda$,
$\Lambda$. The power divergence based test-statistics cover as special cases
the classical ones (\ref{LRT}), (\ref{CS}), actually $T_{0}(\overline
{\boldsymbol{p}},\boldsymbol{p}(\widetilde{\boldsymbol{\theta}}%
),\boldsymbol{p}(\widehat{\boldsymbol{\theta}}))=G^{2}$ and $S_{1}%
(\boldsymbol{p}(\widetilde{\boldsymbol{\theta}}),\boldsymbol{p}%
(\widehat{\boldsymbol{\theta}}))=X^{2}$. The power divergence based
test-statistics with $\lambda=\frac{2}{3}$ are commonly considered for
analysis because their performance is usually quite good. At this setting, a
possible choice for studying its $p$-values is $\lambda\in\Lambda
=\{-1.5,-1,-0.5,0,\frac{2}{3},1,1.5,2\}$. We shall consider$\ $an algorithm
for obtaining the $p$-values associated with hypothesis testing (\ref{4b}) for
a given sample.

\begin{algorithm}
[Calculation of p-value]\label{AlgPval}Let $T\in\{T_{\lambda}(\overline
{\boldsymbol{p}},\boldsymbol{p}(\widetilde{\boldsymbol{\theta}}%
),\boldsymbol{p}(\widehat{\boldsymbol{\theta}})),S_{\lambda}(\boldsymbol{p}%
(\widetilde{\boldsymbol{\theta}}),\boldsymbol{p}(\widehat{\boldsymbol{\theta}%
}))\}_{\lambda\in\Lambda}$ be the test-statistic associated with (\ref{4b}).
In the following steps the corresponding asymptotic $p$-value is calculated
once it is suppose we have $\{w_{j}(\widehat{\boldsymbol{\theta}}%
)\}_{j=0}^{(I-1)(J-1)}$:

\noindent\texttt{STEP 1: Using }$\boldsymbol{n}$\texttt{\ calculate
}$\boldsymbol{p}(\widehat{\boldsymbol{\theta}})$\texttt{\ taking into account
(\ref{ind}).}\newline\texttt{STEP 2: Using }$\boldsymbol{p}%
(\widehat{\boldsymbol{\theta}})$\texttt{\ calculate value }$t$\texttt{ of
test-statistic }$T$\texttt{ using the corresponding expression in
(\ref{PD1})-(\ref{PD6}).}\newline\texttt{STEP 3: Compute }$p$\textrm{-}%
$\mathrm{value}(T):=0$.\newline\texttt{STEP 4: }If $t\leq0$\texttt{, do }%
$p$\textrm{-}$\mathrm{value}(T):=1$, \texttt{otherwise (}if $t>0$%
\texttt{)\newline\hspace*{1.6cm}}for $h=0,...,(I-1)(J-1)-1$,\texttt{ do }%
$p$\textrm{-}$\mathrm{value}(T):=p$\textrm{-}$\mathrm{value}(T)+w_{h}%
(\widehat{\boldsymbol{\theta}})\Pr\left(  \chi_{(I-1)(J-1)-h}^{2}>t\right)
$.\texttt{\newline\hspace*{1.6cm}E.g., the NAG Fortran library subroutine
G01ECF can be useful.}\newline\texttt{\hspace*{1.6cm}(Remark: for small sample
sizes and for values of }$T_{\lambda}(\overline{\boldsymbol{p}},\boldsymbol{p}%
(\widetilde{\boldsymbol{\theta}}),\boldsymbol{p}(\widehat{\boldsymbol{\theta}%
}))$\texttt{, sometimes }$t<0$\texttt{).}
\end{algorithm}

In Table \ref{t1}, the power divergence based test-statistics and their
corresponding asymptotic $p$-values are shown. For all the power divergence
based test-statistics, $T\in\{T_{\lambda}(\overline{\boldsymbol{p}%
},\boldsymbol{p}(\widetilde{\boldsymbol{\theta}}),\boldsymbol{p}%
(\widehat{\boldsymbol{\theta}})),S_{\lambda}(\boldsymbol{p}%
(\widetilde{\boldsymbol{\theta}}),\boldsymbol{p}(\widehat{\boldsymbol{\theta}%
}))\}_{\lambda\in\Lambda}$, the order restricted hypothesis cannot be rejected
with a significance level $0.05$. So, it is accepted that the probability of
having side effects increases when the severity degree of the operation increases.%

\begin{table}[htbp]  \tabcolsep2.8pt  \centering
$%
\begin{tabular}
[c]{ccccccccc}\hline\hline
test-statistic & $\lambda=-1.5$ & $\lambda=-1$ & $\lambda=-0.5$ & $\lambda=0$
& $\lambda=\frac{2}{3}$ & $\lambda=1$ & $\lambda=1.5$ & $\lambda=2$\\\hline
\multicolumn{1}{l}{$\overset{}{T}_{\lambda}(\overline{\boldsymbol{p}%
},\boldsymbol{p}(\widetilde{\boldsymbol{\theta}}),\boldsymbol{p}%
(\widehat{\boldsymbol{\theta}}))$} & $9.4681$ & $9.1918$ & $8.9535$ & $8.7497$
& $8.5262$ & $8.4334$ & $8.3160$ & $8.2230$\\
$p\mathrm{-value}(T_{\lambda})$ & $0.0123$ & $0.0139$ & $0.0155$ & $0.0170$ &
$0.0188$ & $0.0196$ & $0.0206$ & $0.0215$\\
\multicolumn{1}{l}{$S_{\lambda}(\boldsymbol{p}(\widetilde{\boldsymbol{\theta}%
}),\boldsymbol{p}(\widehat{\boldsymbol{\theta}}))$} & $9.1282$ & $8.9774$ &
$8.8520$ & $8.7497$ & $8.6463$ & $8.6076$ & $8.5650$ & $8.5399$\\
$p\mathrm{-value}(S_{\lambda})$ & $0.0143$ & $0.0153$ & $0.0162$ & $0.0170$ &
$0.0178$ & $0.0181$ & $0.0184$ & $0.0186$\\\hline\hline
\end{tabular}
\ \ \ $%
\caption{Power divergence based test-statistics and asymptotic p-values.\label{t1}}%
\end{table}%

\section{Monte Carlo Study\label{MC}}

Taking five cases, depending on $\delta\in\{0,0.1,0.5,1,1.5,15\}$, we
considered $I=4$ independent trinomial samples ($J=3$) with a vector of
theoretical probabilities,%
\begin{align*}
\boldsymbol{\pi}_{i}(\boldsymbol{\theta}(\delta))  &  =(\pi_{i1}%
(\boldsymbol{\theta}(\delta)),\pi_{i2}(\boldsymbol{\theta}(\delta)),\pi
_{i3}(\boldsymbol{\theta}(\delta)))^{T}\\
\pi_{ij}(\boldsymbol{\theta}(\delta))  &  =\frac{1}{3}\frac{1+i(j-1)\delta
}{1+i\delta},\quad i=1,...,4,\quad j=1,...,3,
\end{align*}
for each of the $I=4$ independent multinomial samples, in four scenarios:
\newline$\hspace*{0.5cm}\ast$ Scenario 1: $n=28$, $n_{1}=4$, $n_{2}=6$,
$n_{3}=8$, $n_{4}=10$;\newline$\hspace*{0.5cm}\ast$ Scenario 2: $n=56$,
$n_{1}=8$, $n_{2}=12$, $n_{3}=16$, $n_{4}=20$;\newline$\hspace*{0.5cm}\ast$
Scenario 3: $n=84$, $n_{1}=12$, $n_{2}=18$, $n_{3}=24$, $n_{4}=30$%
;\newline$\hspace*{0.5cm}\ast$ Scenario 4: $n=112$, $n_{1}=16$, $n_{2}=24$,
$n_{3}=32$, $n_{4}=40$.\newline It is worthwhile to mention that we have
chosen either equal or unequal sample sizes and we did not find any different
performance as it was found for the stochastic ordering in Wang (1996). When
$\delta=0$, the null hypothesis is held, $\boldsymbol{\pi}_{i}%
(\boldsymbol{\theta}(0))=\boldsymbol{\pi}(\boldsymbol{\theta}_{0})=(\frac
{1}{3},\frac{1}{3},\frac{1}{3})^{T}$, $i=1,2,3,4$, while in the rest of the
values of $\delta$\ the alternative hypothesis is held.

Let $R=10,000$ be the number of replications considered for the Monte Carlo
study. Once the nominal size of the test is prefixed to be $\alpha=0.05$, the
exact size of the test associated with $T\in\{T_{\lambda}(\overline
{\boldsymbol{p}},\boldsymbol{p}(\widetilde{\boldsymbol{\theta}}%
),\boldsymbol{p}(\widehat{\boldsymbol{\theta}})),S_{\lambda}(\boldsymbol{p}%
(\widetilde{\boldsymbol{\theta}}),\boldsymbol{p}(\widehat{\boldsymbol{\theta}%
}))\}_{\lambda\in\Lambda}$, $\lambda\in\Lambda=\{-1.5,-1,-0.5,0,\frac{2}%
{3},1,1.5,2\}$, can be estimated through%
\[
\widehat{\alpha}_{T}=\frac{\sum_{h=1}^{R}I(p\text{\textrm{-}}\mathrm{value}%
(T_{h})<\alpha)}{R},
\]
taking into account that $\mathrm{p}$\textrm{-}$\mathrm{value}(T_{h})$ is the
$p$-value obtained in the $h$-th replication by using Algorithm \ref{AlgPval}
and $I(\bullet)$ is the indicator function, which takes value $1$ if $\bullet
$\ is true and $0$ otherwise. It is expected a more precise value of
$\widehat{\alpha}_{T}$ with respect to the nominal size $\alpha$, as $n$ is
greater (the least precise nominal sizes in Scenario 1 and the most precise
nominal sizes in Scenario 4). The first interest of the simulation study is
focussed on identifying which test-statistic has the best approximation of
$\widehat{\alpha}_{T}$\ with respect to $\alpha$, for all the scenarios.

In Table \ref{t2} the local odds ratios,%
\[
\vartheta_{ij}=\vartheta_{ij}(\delta)=\frac{1+i(j-1)\delta}{1+(i+1)(j-1)\delta
}\frac{1+(i+1)j\delta}{1+ij\delta},
\]
$(i,j)\in\{1,2,3\}\times\{1,2\}$, are shown for $\delta\in\{0.1,0.5,1,1.5\}$.
Notice that in $\boldsymbol{\vartheta}=\boldsymbol{\vartheta}(\delta
)=(\vartheta_{11}(\delta),\vartheta_{12}(\delta),\vartheta_{21}(\delta
),\vartheta_{22}(\delta),$\allowbreak$\vartheta_{31}(\delta),\vartheta
_{32}(\delta))^{T}$\ some of the components are further from
$\boldsymbol{\vartheta}(0)=\boldsymbol{1}_{6}$ (null hypothesis), as the value
of $\delta>0$ is further from $0$. This means that a greater value of the
estimation of the power function might be obtained,%
\[
\widehat{\beta}_{T}(\delta)=\frac{\sum_{h=1}^{R}I(p\text{\textrm{-}%
}\mathrm{value}(T_{h})<\alpha)}{R},
\]
as $\delta>0$ is greater. This claim is supported by the fact that some values
of the components of $\boldsymbol{\vartheta}=\boldsymbol{\vartheta}(\delta)$
decrease as $\delta$ increases but more slowly than the others increase. In
addition, for a fixed value of $\delta>0$, it is expected a greater value of
$\widehat{\beta}_{T}(\delta)$, as $n$ is greater (the worst powers in Scenario
1 and the best powers in Scenario 4). It is also worthwhile to mention that as
$\delta$ increases $\pi_{i+1,1}(\boldsymbol{\theta}(\delta))/\pi
_{i1}(\boldsymbol{\theta}(\delta))$ remains constant, $\pi_{i+1,2}%
(\boldsymbol{\theta}(\delta))/\pi_{i2}(\boldsymbol{\theta}(\delta))$\ is not
constant for $i=1,2,3$, and $\pi_{i+1,2}(\boldsymbol{\theta}(\delta))/\pi
_{i2}(\boldsymbol{\theta}(\delta))$ is approaching the limit ($\pi
_{i+1,2}(\boldsymbol{\theta}(\infty))/\pi_{i2}(\boldsymbol{\theta}(\infty))$)
on the right for $i=1,2,3$. The second interest of the simulation study is
focussed in identifying which test-statistic has the best performance in
powers and at the same time in approximating $\widehat{\alpha}_{T}$\ by
$\alpha$, in all the scenarios.%

\begin{table}[htbp]  \tabcolsep2.8pt  \centering
\begin{tabular}
[c]{ccccccccccccc}\hline
&  & $\delta=0$ &  & $\delta=0.1$ &  & $\delta=0.5$ &  & $\delta=1$ &  &
$\delta=1.5$ &  & $\delta=\infty$\\\hline
$\vartheta_{11}=\vartheta_{11}(\delta)$ &  & $1.000$ &  & $1.091$ &  & $1.333$
&  & $1.500$ &  & $1.600$ &  & $2.00$\\
$\vartheta_{12}=\vartheta_{12}(\delta)$ &  & $1.000$ &  & $1.069$ &  & $1.125$
&  & $1.111$ &  & $1.094$ &  & $1.00$\\
$\vartheta_{21}=\vartheta_{21}(\delta)$ &  & $1.000$ &  & $1.083$ &  & $1.250$
&  & $1.333$ &  & $1.375$ &  & $1.50$\\
$\vartheta_{22}=\vartheta_{22}(\delta)$ &  & $1.000$ &  & $1.055$ &  & $1.066$
&  & $1.050$ &  & $1.039$ &  & $1.00$\\
$\vartheta_{31}=\vartheta_{31}(\delta)$ &  & $1.000$ &  & $1.077$ &  & $1.200$
&  & $1.250$ &  & $1.273$ &  & $1.33$\\
$\vartheta_{32}=\vartheta_{32}(\delta)$ &  & $1.000$ &  & $1.045$ &  & $1.042$
&  & $1.029$ &  & $1.021$ &  & $1.000$\\\hline
$\pi_{21}(\boldsymbol{\theta}(\delta))/\pi_{11}(\boldsymbol{\theta}(\delta))$
&  & $0.33/0.33$ &  & $0.28/0.30$ &  & $0.17/0.22$ &  & $0.11/0.17$ &  &
$0.08/0.13$ &  & $0.50$\\
$\pi_{22}(\boldsymbol{\theta}(\delta))/\pi_{12}(\boldsymbol{\theta}(\delta))$
&  & $0.33/0.33$ &  & $0.33/0.33$ &  & $0.33/0.33$ &  & $0.33/0.33$ &  &
$0.33/0.33$ &  & $1.00$\\
$\pi_{31}(\boldsymbol{\theta}(\delta))/\pi_{21}(\boldsymbol{\theta}(\delta))$
&  & $0.33/0.33$ &  & $0.26/0.28$ &  & $0.13/0.17$ &  & $0.08/0.11$ &  &
$0.06/0.08$ &  & $0.67$\\
$\pi_{32}(\boldsymbol{\theta}(\delta))/\pi_{22}(\boldsymbol{\theta}(\delta))$
&  & $0.33/0.33$ &  & $0.33/0.33$ &  & $0.33/0.33$ &  & $0.33/0.33$ &  &
$0.33/0.33$ &  & $1.00$\\
$\pi_{41}(\boldsymbol{\theta}(\delta))/\pi_{31}(\boldsymbol{\theta}(\delta))$
&  & $0.33/0.33$ &  & $0.24/0.26$ &  & $0.11/0.13$ &  & $0.07/0.08$ &  &
$0.05/0.06$ &  & $0.75$\\
$\pi_{42}(\boldsymbol{\theta}(\delta))/\pi_{32}(\boldsymbol{\theta}(\delta))$
&  & $0.33/0.33$ &  & $0.33/0.33$ &  & $0.33/0.33$ &  & $0.33/0.33$ &  &
$0.33/0.33$ &  & $1.00$\\\hline
\end{tabular}
\caption{Theoretical local odd ratios for the Monte Carlo study.\label{t2}}%
\end{table}%

Once a nominal size $\alpha=0.05$ is established, Table \ref{alfas} summarizes
the simulated exact sizes in all the scenarios for the test-statistic
${T\in\{T_{\lambda},S_{\lambda}\}}_{\lambda\in\Lambda}$,with $\Lambda
=\{-1.5,-1,-\frac{1}{2},0,\frac{2}{3},1,1.5,2\}$. We have plotted $3\times2$
graphs in Figures \ref{fig1}-\ref{fig4} and we refer them as plots in three
rows. In the first row of Figures \ref{fig1}-\ref{fig4} we can see on the left
the exact power in all the scenarios for the test-statistic $\{{T_{\lambda}%
\}}_{\lambda\in\lbrack-1.5,3]}$ and on the right for the test-statistic
$\{{S_{\lambda}\}}_{\lambda\in\lbrack-1.5,3]}$. In order to make a comparison
of exact powers, we cannot directly proceed without considering the exact
sizes. For this reason we are going to give a procedure based on two steps.

\noindent\textit{Step 1}: We are going to check for all the power divergence
based test-statistics the criterion given by Dale (1986), i.e.,
\begin{equation}
|\,\text{logit}(1-{\widehat{\alpha}}_{T})-\text{logit}(1-\alpha)\,|\leq e
\label{con1}%
\end{equation}
with $\mathrm{logit}\left(  p\right)  =\log\left(  \frac{p}{1-p}\right)  $. We
only consider the values of $\lambda$\ such that ${\widehat{\alpha}}_{T}%
$\ verifies (\ref{con1}) with $e=0.35$, then we shall only consider the
test-statistics such that ${\widehat{\alpha}}_{T}\in\left[
0.0357,0.0695\right]  $, in all the scenarios. This criterion has been
considered for some authors, see for instance Cressie et al. (2003) and
Mart\'{\i}n and Pardo (2012). The cases satisfying the criterion are marked in
bold in Table \ref{alfas}, and comprise those values in the abscissa of the
plot between the dashed band (the dashed line in the middle represents the
nominal size), and we can conclude that we must not consider in our study
${T\in\{T_{\lambda},S_{\lambda}\}}_{\lambda\in\lbrack-1.5,-0.4)}$.

\noindent\textit{Step 2}: We compare all the test statistics obtained in Step
1 with the classical likelihood ratio test ($G^{2}=T_{0}$) as well as the
Pearson test statistic ($X^{2}=S_{1}$). To do so, we have calculated the
relative local efficiencies%
\[
\widehat{\rho}_{T}=\widehat{\rho}_{T}(\delta)=\frac{({\widehat{\beta}}%
_{T}(\delta)-{\widehat{\alpha}}_{T})-({\widehat{\beta}}_{T_{0}}(\delta
)-{\widehat{\alpha}}_{T_{0}})}{{\widehat{\beta}}_{T_{0}}(\delta
)-{\widehat{\alpha}}_{T_{0}}},\qquad\widehat{\rho}_{T}^{\ast}=\widehat{\rho
}_{T}^{\ast}(\delta)=\frac{({\widehat{\beta}}_{T}(\delta)-{\widehat{\alpha}%
}_{T})-({\widehat{\beta}}_{S_{1}}(\delta)-{\widehat{\alpha}}_{S_{1}}%
)}{{\widehat{\beta}}_{S_{1}}(\delta)-{\widehat{\alpha}}_{S_{1}}}.
\]
It is important to mention that we are comparing the proposed test-statistics
with respect to the classical likelihood ratio test ($G^{2}=T_{0}$), which is
the only asymptotic test-statistic considered in the literature of hypothesis
testing (\ref{eq2}) against (\ref{eq3}), however we are considering also the
comparisons with respect to the chi-square test statistic ($X^{2}=S_{1}$)
since this is well-known in other ordering types for having good asymptotic
performance (see Mart\'{\i}n and Balakrishnan (2013) and references therein).

In Figures \ref{fig1}-\ref{fig4} the powers and the relative local
efficiencies are summarized. The second rows of the figures represent
$\widehat{\rho}_{T}$, while in the third row is plotted $\widehat{\rho}%
_{T}^{\ast}$, on the left it is considered ${T=T_{\lambda}}$ and
${T=S_{\lambda}}$ on the right.

In all the scenarios a similar pattern is observed when plotting the exact
power, ${\widehat{\beta}}_{T}$, for $\lambda\in(-1,3)$ since a U shaped curve
is obtained. This means that the exact power is higher in the corners of the
interval in comparison with the classical likelihood ratio test ($G^{2}=T_{0}%
$) as well as the classical Pearson test statistic ($X^{2}=S_{1}$), contained
in the middle. The likelihood ratio test has very bad performance in relation
to the simulated exact size, and we restrict ourselves to $(0,2]$, taking into
account the simulated exact sizes. The Cressie-Read test-statistic ($T_{2/3}$)
and the chi-square one ($X^{2}=S_{1}$) have very good performance in regards
to the simulated exact size since it is very close to nominal size
$\alpha=0.05$. If we pay attention on the local efficiencies with respect to
$G^{2}$ and $X^{2}$, $\widehat{\rho}_{T}$ and $\widehat{\rho}_{T}^{\ast}$,
${T_{\lambda}}$ and ${S_{\lambda}}$\ with $\lambda$ close to $2$ have big
values since their powers are greater in comparison with ${T_{\lambda}}$ and
${S_{\lambda}}$\ with $\lambda$ close to $0$. For $\lambda$ close to $2$, the
values of $\widehat{\rho}_{T}$ are a slighly superior in comparison with
$\widehat{\rho}_{T}^{\ast}$. Taking into account the plots we conclude that
${T_{-2}}$ and ${S_{-2}}$ have clearly the best performance for moderate
sample sizes (scenarios 3 and 4) and for small sample sizes (scenarios 1 and
2) the same test-statistics have good performance according to $\widehat{\rho
}_{T}$ and $\widehat{\rho}_{T}^{\ast}$, however with the Cressie-Read
test-statistic ($T_{2/3}$) and the chi-square one ($X^{2}=S_{1}$) a better
simulated exact sizes were obtained.%

\begin{table}[htbp]  \tabcolsep2.8pt  \centering
\begin{tabular}
[c]{l}%
$%
\begin{tabular}
[c]{ccccccccc}\hline
Scenario & ${\widehat{\alpha}}_{T_{-1.5}}$ & ${\widehat{\alpha}}_{T_{-1}}$ &
${\widehat{\alpha}}_{T_{-0.5}}$ & ${\widehat{\alpha}}_{T_{0}}$ &
${\widehat{\alpha}}_{T_{\frac{2}{3}}}$ & ${\widehat{\alpha}}_{T_{1}}$ &
${\widehat{\alpha}}_{T_{1.5}}$ & ${\widehat{\alpha}}_{T_{2}}$\\\hline
Scenario 1 & $0.0111$ & $0.0079$ & $0.1626$ & $0.0706$ & $\boldsymbol{0.0488}$
& $\boldsymbol{0.0481}$ & $\boldsymbol{0.0533}$ & $\boldsymbol{0.0651}$\\
Scenario 2 & $\boldsymbol{0.0623}$ & $\boldsymbol{0.0480}$ & $0.0888$ &
$0.0665$ & $\boldsymbol{0.0501}$ & $\boldsymbol{0.0492}$ &
$\boldsymbol{0.0502}$ & $\boldsymbol{0.0542}$\\
Scenario 3 & $0.0911$ & $0.0702$ & $\boldsymbol{0.0648}$ &
$\boldsymbol{0.0529}$ & $\boldsymbol{0.0477}$ & $\boldsymbol{0.0474}$ &
$\boldsymbol{0.0485}$ & $\boldsymbol{0.0531}$\\
Scenario 4 & $0.0827$ & $0.0708$ & $\boldsymbol{0.0620}$ &
$\boldsymbol{0.0550}$ & $\boldsymbol{0.0494}$ & $\boldsymbol{0.0485}$ &
$\boldsymbol{0.0487}$ & $\boldsymbol{0.0534}$\\\hline
\end{tabular}
\ \ \ \ $\\
$%
\begin{tabular}
[c]{ccccccccc}\hline
Scenario & ${\widehat{\alpha}}_{S_{-1.5}}$ & ${\widehat{\alpha}}_{S_{-1}}$ &
${\widehat{\alpha}}_{S_{-0.5}}$ & ${\widehat{\alpha}}_{S_{0}}$ &
${\widehat{\alpha}}_{S_{\frac{2}{3}}}$ & ${\widehat{\alpha}}_{S_{1}}$ &
${\widehat{\alpha}}_{S_{1.5}}$ & ${\widehat{\alpha}}_{S_{2}}$\\\hline
Scenario 1 & $0.2356$ & $0.2299$ & $0.1409$ & $0.0706$ & $\boldsymbol{0.0514}$
& $\boldsymbol{0.0498}$ & $\boldsymbol{0.0527}$ & $\boldsymbol{0.0599}$\\
Scenario 2 & $0.0966$ & $0.0887$ & $0.0809$ & $\boldsymbol{0.0665}$ &
$\boldsymbol{0.0521}$ & $\boldsymbol{0.0514}$ & $\boldsymbol{0.0512}$ &
$\boldsymbol{0.0532}$\\
Scenario 3 & $0.0835$ & $0.0697$ & $\boldsymbol{0.0597}$ &
$\boldsymbol{0.0529}$ & $\boldsymbol{0.0485}$ & $\boldsymbol{0.0479}$ &
$\boldsymbol{0.0490}$ & $\boldsymbol{0.0516}$\\
Scenario 4 & $0.0726$ & $\boldsymbol{0.0654}$ & $\boldsymbol{0.0590}$ &
$\boldsymbol{0.0550}$ & $\boldsymbol{0.0516}$ & $\boldsymbol{0.0508}$ &
$\boldsymbol{0.0505}$ & $\boldsymbol{0.0533}$\\\hline
\end{tabular}
\ \ \ \ $%
\end{tabular}
\caption{${\widehat{\alpha}}_{T}$, for ${T\in\{T_{\lambda},S_{\lambda}\}}_{\lambda\in\Lambda}$ in the four scenarios. \label{alfas}}%
\end{table}%
%

\begin{figure}[htbp]  \tabcolsep2.8pt  \centering
\begin{tabular}
[c]{cc}%
${T_{\lambda}}$ & ${S_{\lambda}}$\\%
\raisebox{-0cm}{\includegraphics[
trim=0.000000in 0.000000in -0.107791in -0.004856in,
height=6.1066cm,
width=7.7211cm
]%
{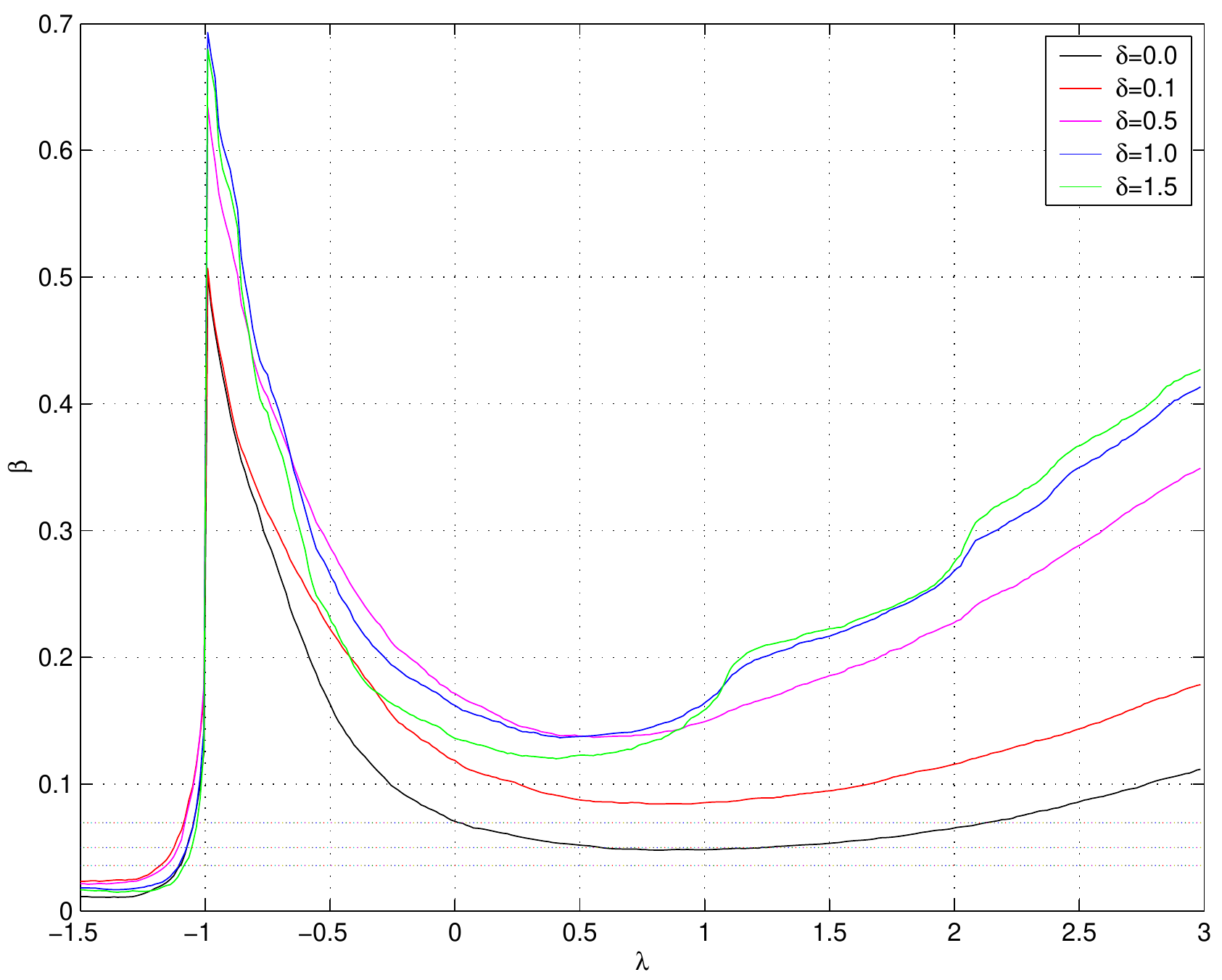}%
}%
&
{\includegraphics[
height=2.4016in,
width=2.9922in
]%
{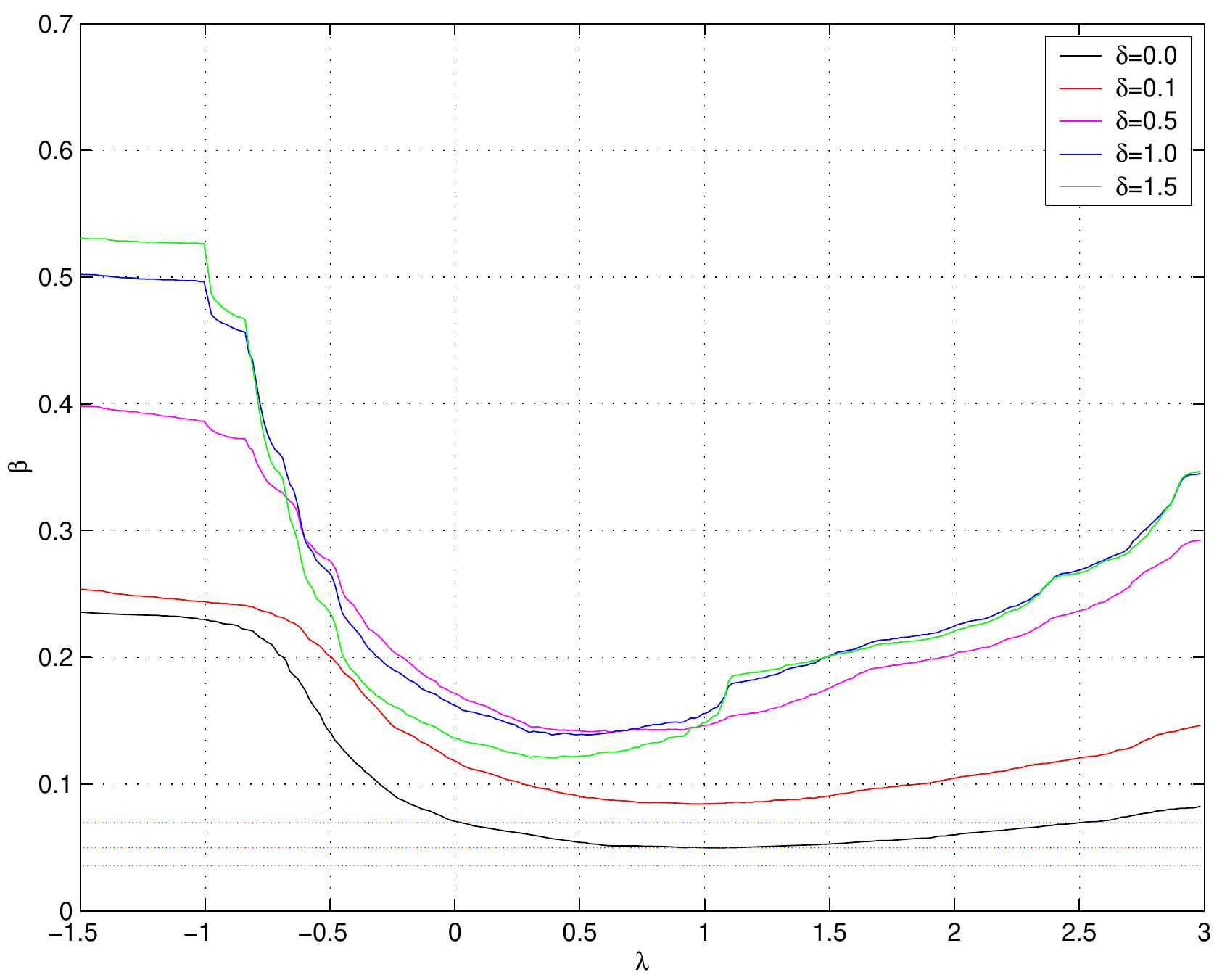}%
}%
\\%
{\includegraphics[
height=2.4016in,
width=2.9784in
]%
{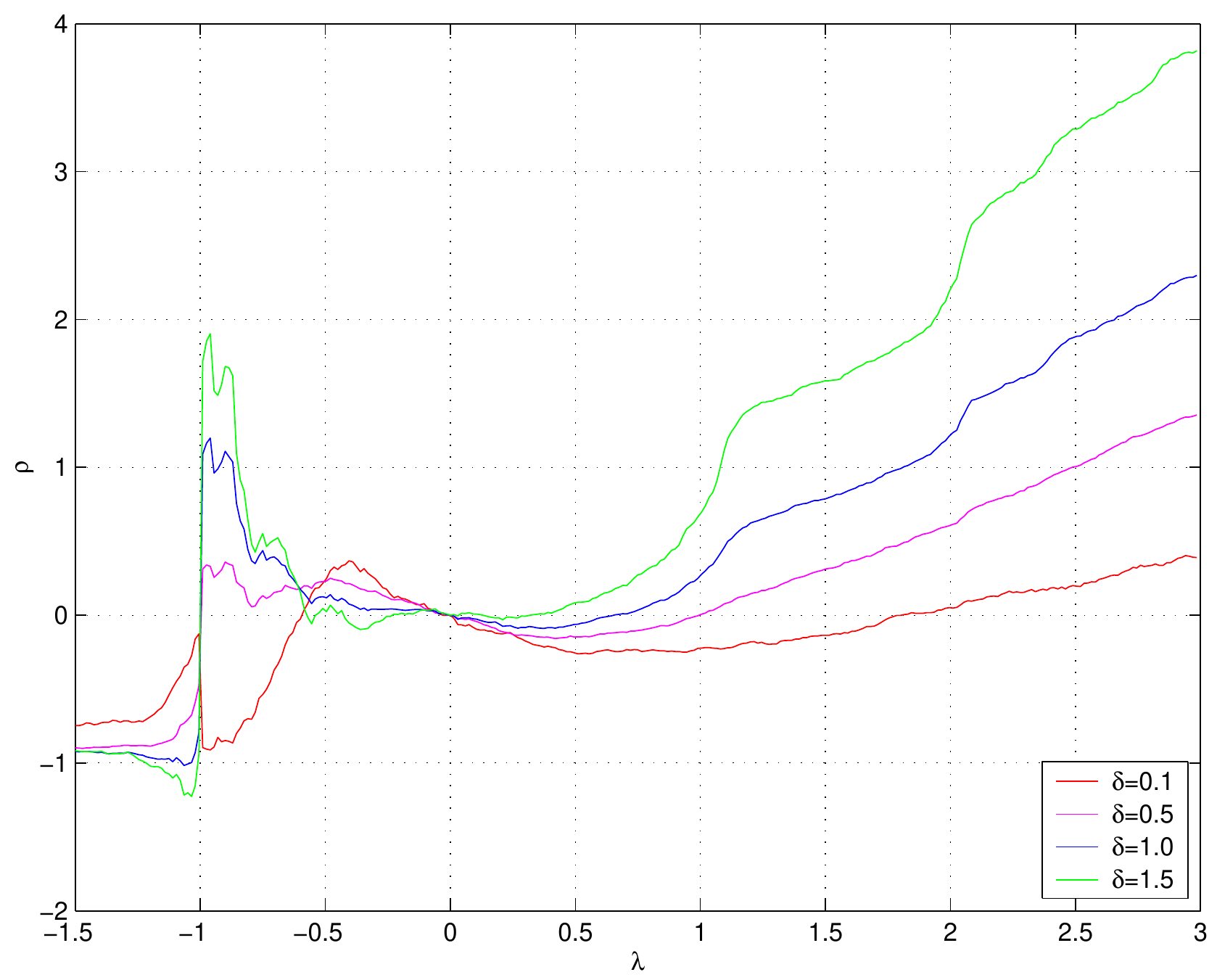}%
}%
&
{\includegraphics[
height=2.4016in,
width=3.0286in
]%
{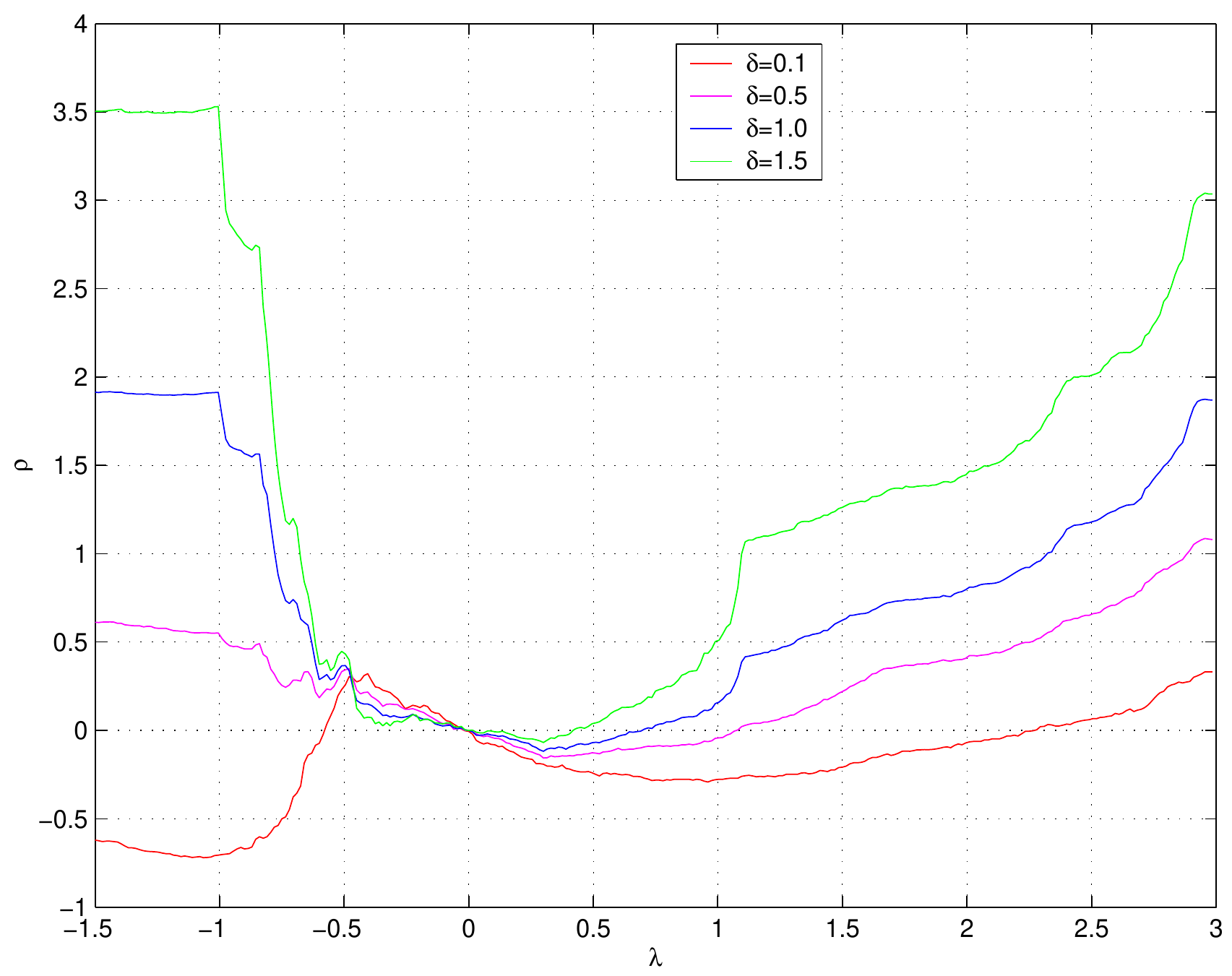}%
}%
\\%
{\includegraphics[
height=2.4016in,
width=3.0286in
]%
{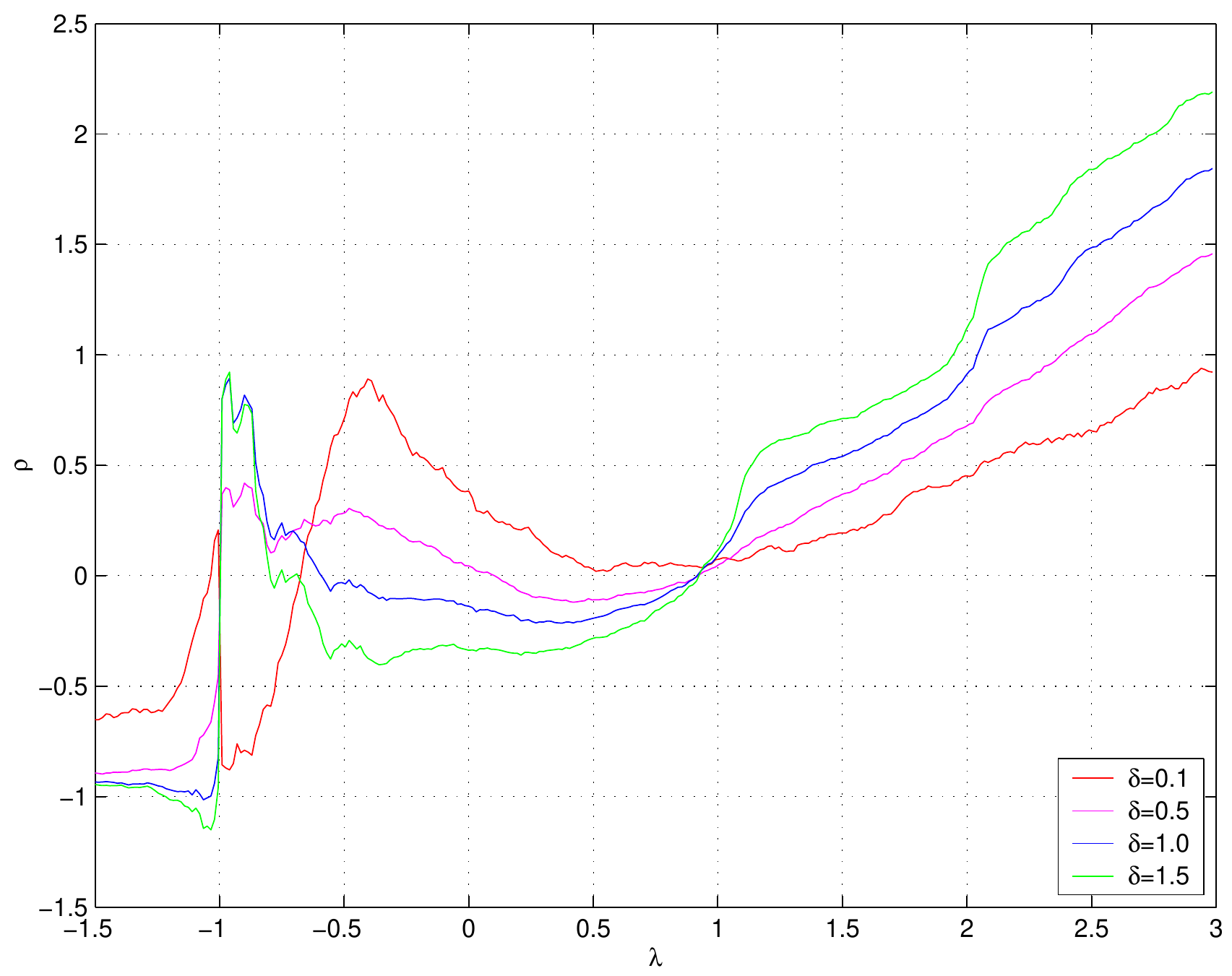}%
}%
&
{\includegraphics[
height=2.4016in,
width=3.0286in
]%
{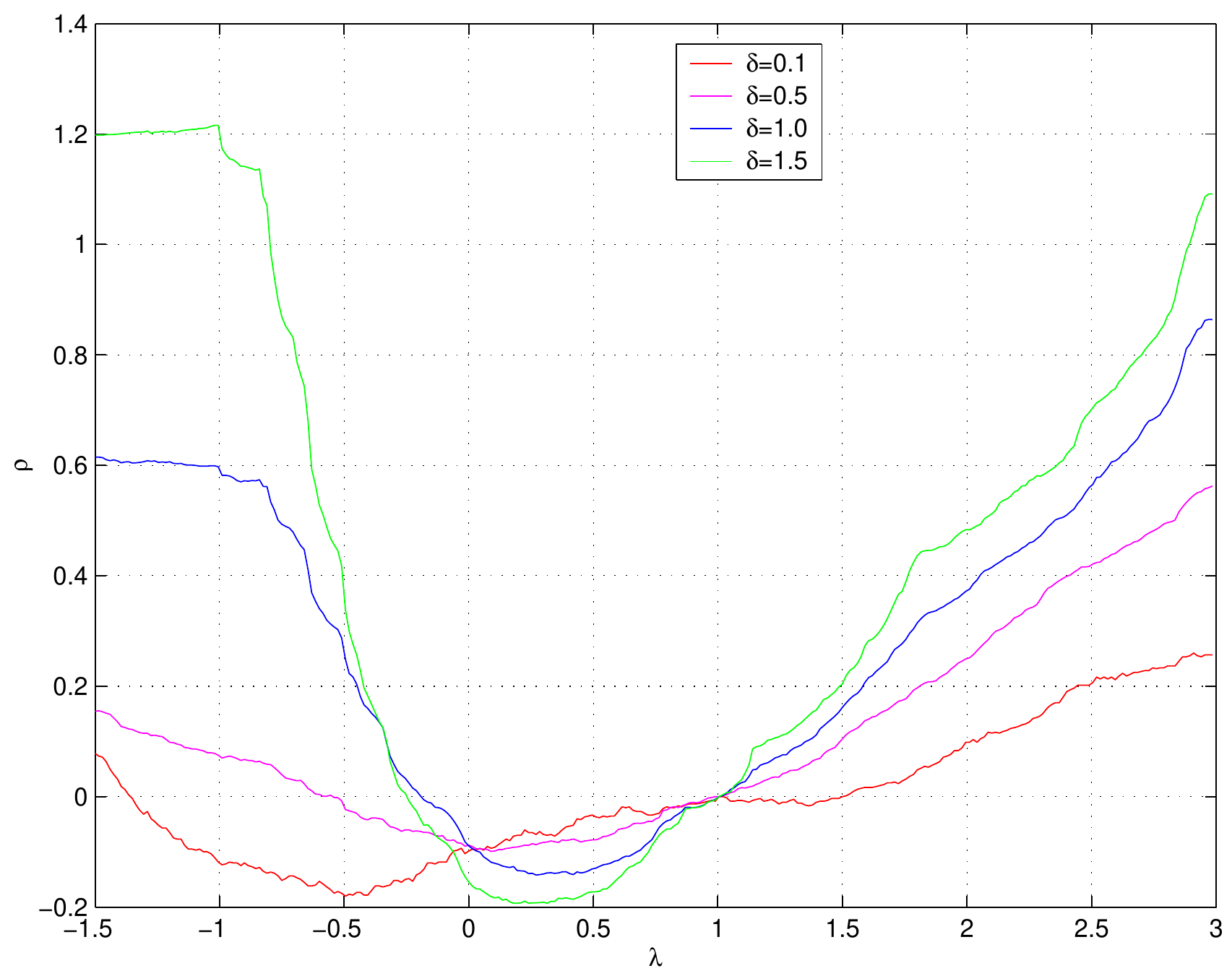}%
}%
\end{tabular}
\caption{Power and relative local efficiencies for $T_{\lambda}$ and $S_{\lambda}$ in scenario 1. \label{fig1}}%
\end{figure}%
%

\begin{figure}[htbp]  \tabcolsep2.8pt  \centering
\begin{tabular}
[c]{cc}%
${T_{\lambda}}$ & ${S_{\lambda}}$\\%
{\includegraphics[
height=2.4016in,
width=2.9922in
]%
{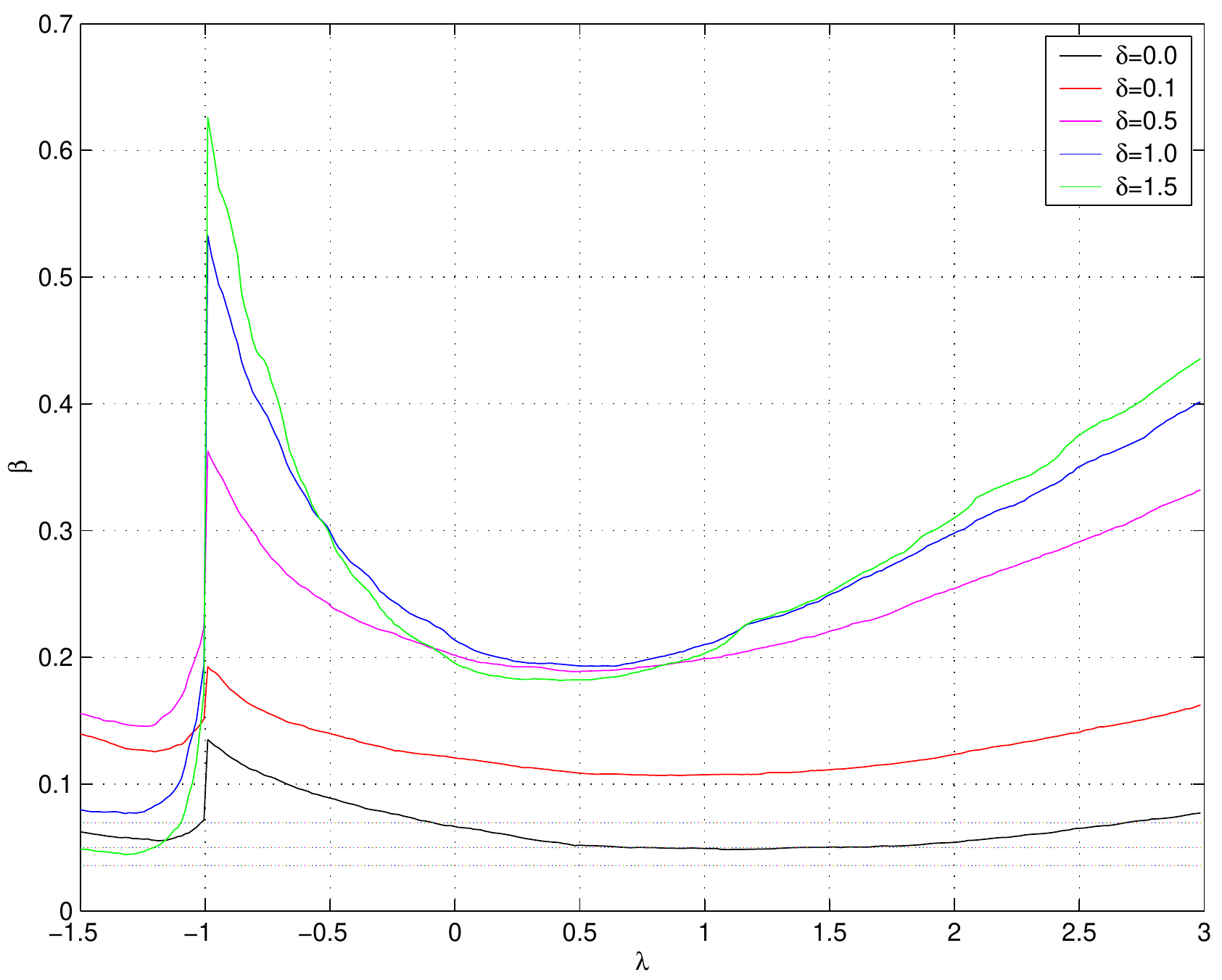}%
}%
&
{\includegraphics[
height=2.4016in,
width=3.0286in
]%
{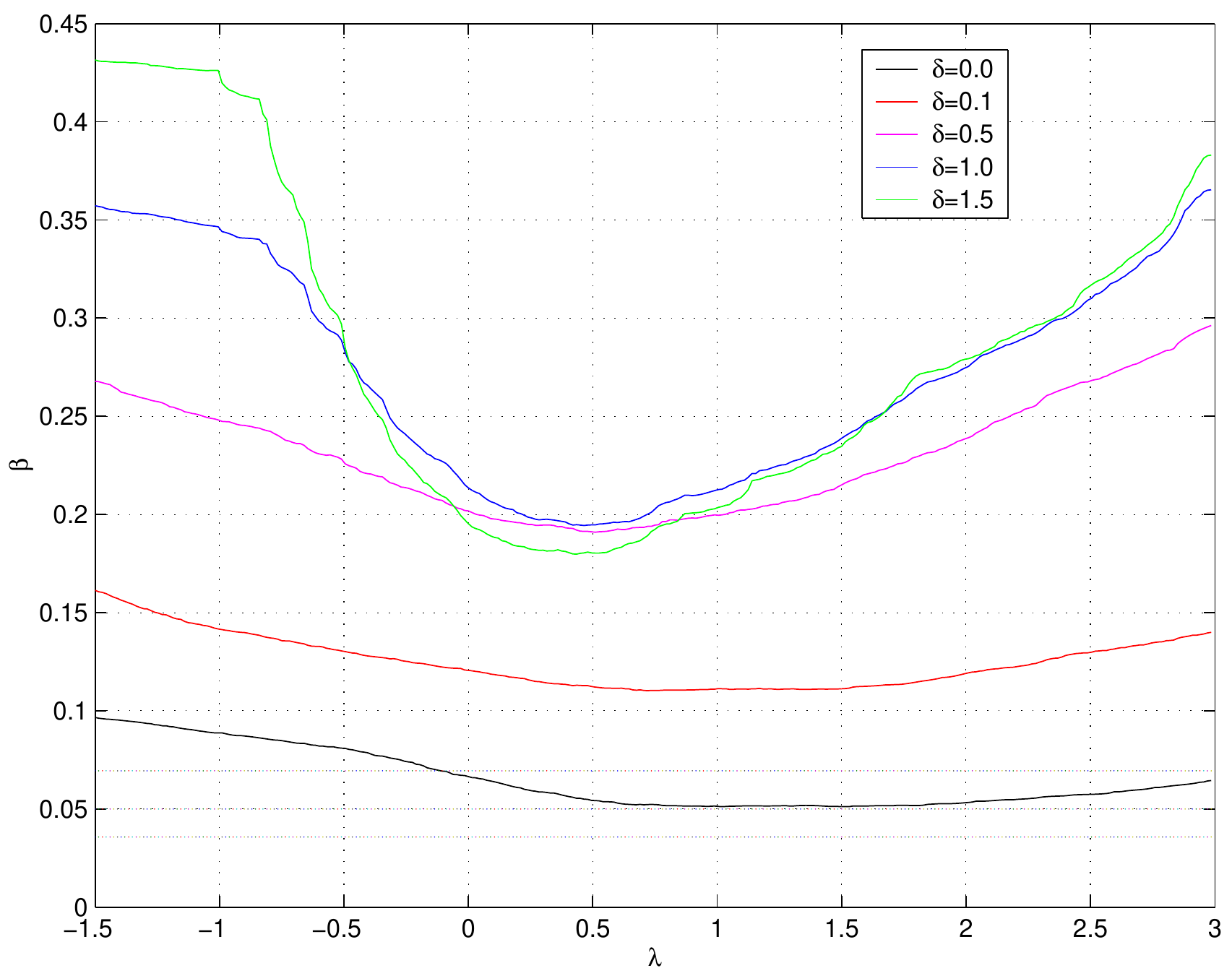}%
}%
\\%
{\includegraphics[
height=2.4016in,
width=3.0286in
]%
{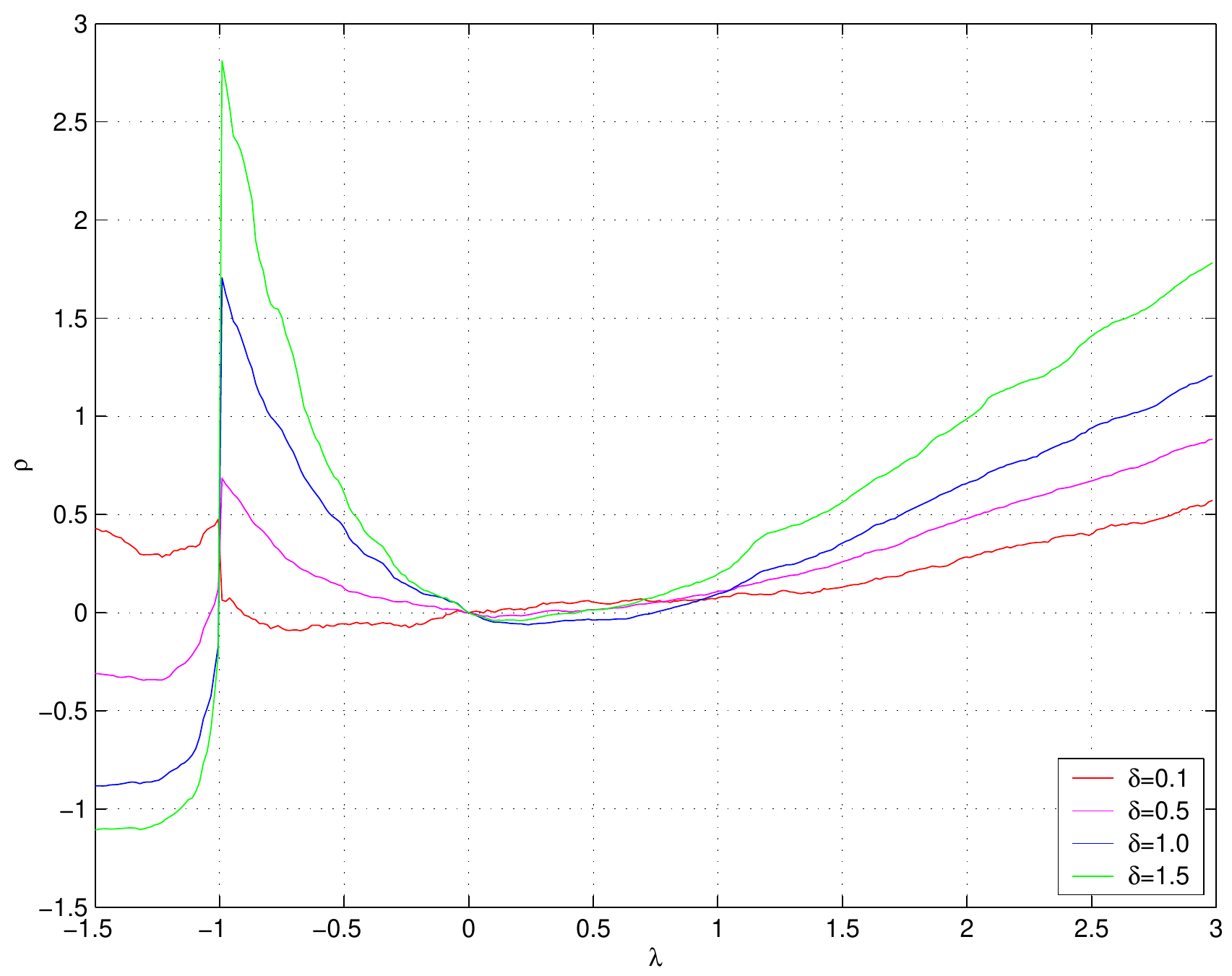}%
}%
&
{\includegraphics[
height=2.4016in,
width=3.0286in
]%
{EscB_Eficiencias_asterisco_S.pdf}%
}%
\\%
{\includegraphics[
height=2.4016in,
width=3.0286in
]%
{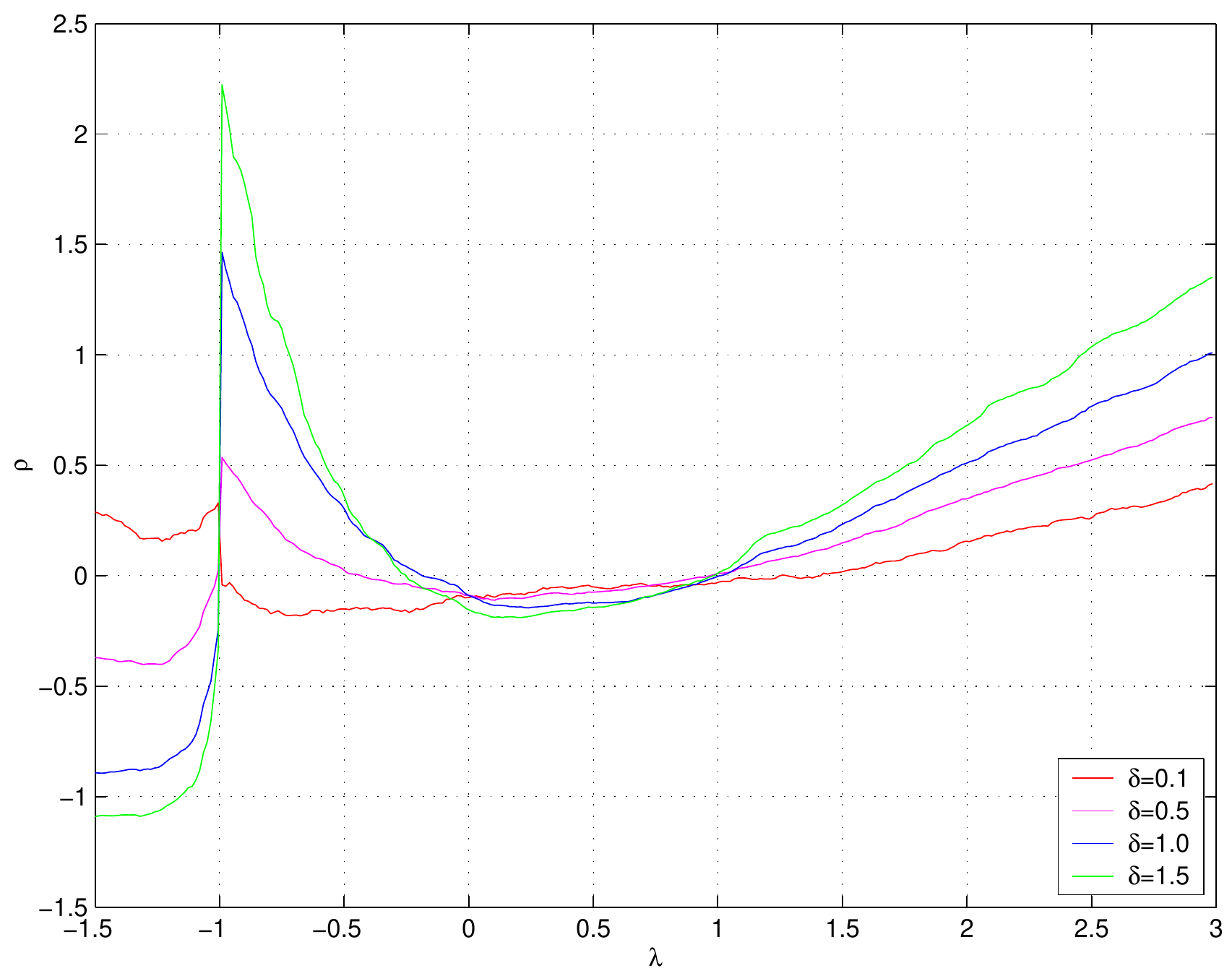}%
}%
&
{\includegraphics[
height=2.4016in,
width=3.0286in
]%
{EscB_Eficiencias_asterisco_S.pdf}%
}%
\end{tabular}
\caption{Power and relative local efficiencies for $T_{\lambda}$ and $S_{\lambda}$ in scenario 2. \label{fig2}}%
\end{figure}%
%

\begin{figure}[htbp]  \tabcolsep2.8pt  \centering
\begin{tabular}
[c]{cc}%
${T_{\lambda}}$ & ${S_{\lambda}}$\\%
{\includegraphics[
height=2.4016in,
width=2.9922in
]%
{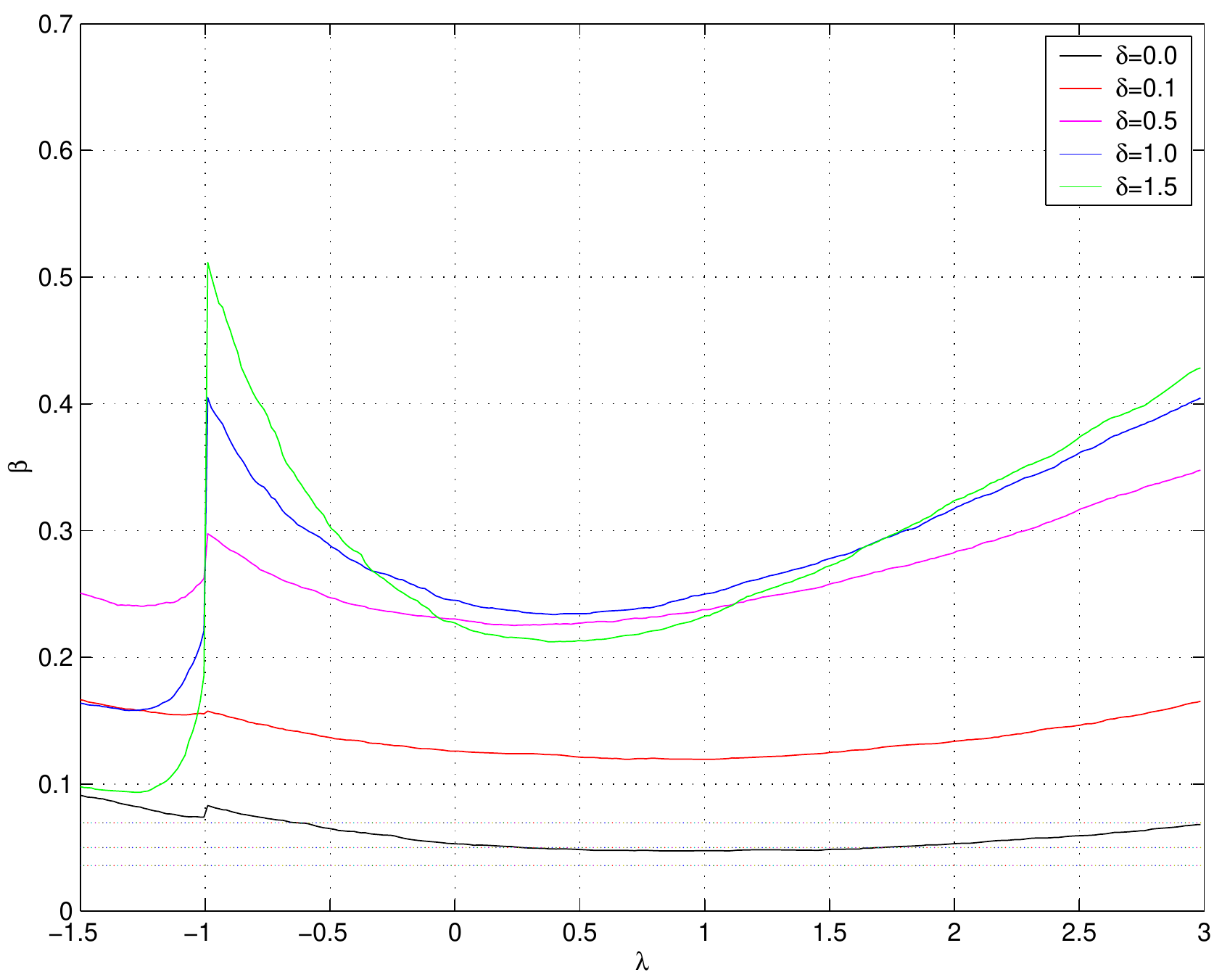}%
}%
&
{\includegraphics[
height=2.4016in,
width=3.0286in
]%
{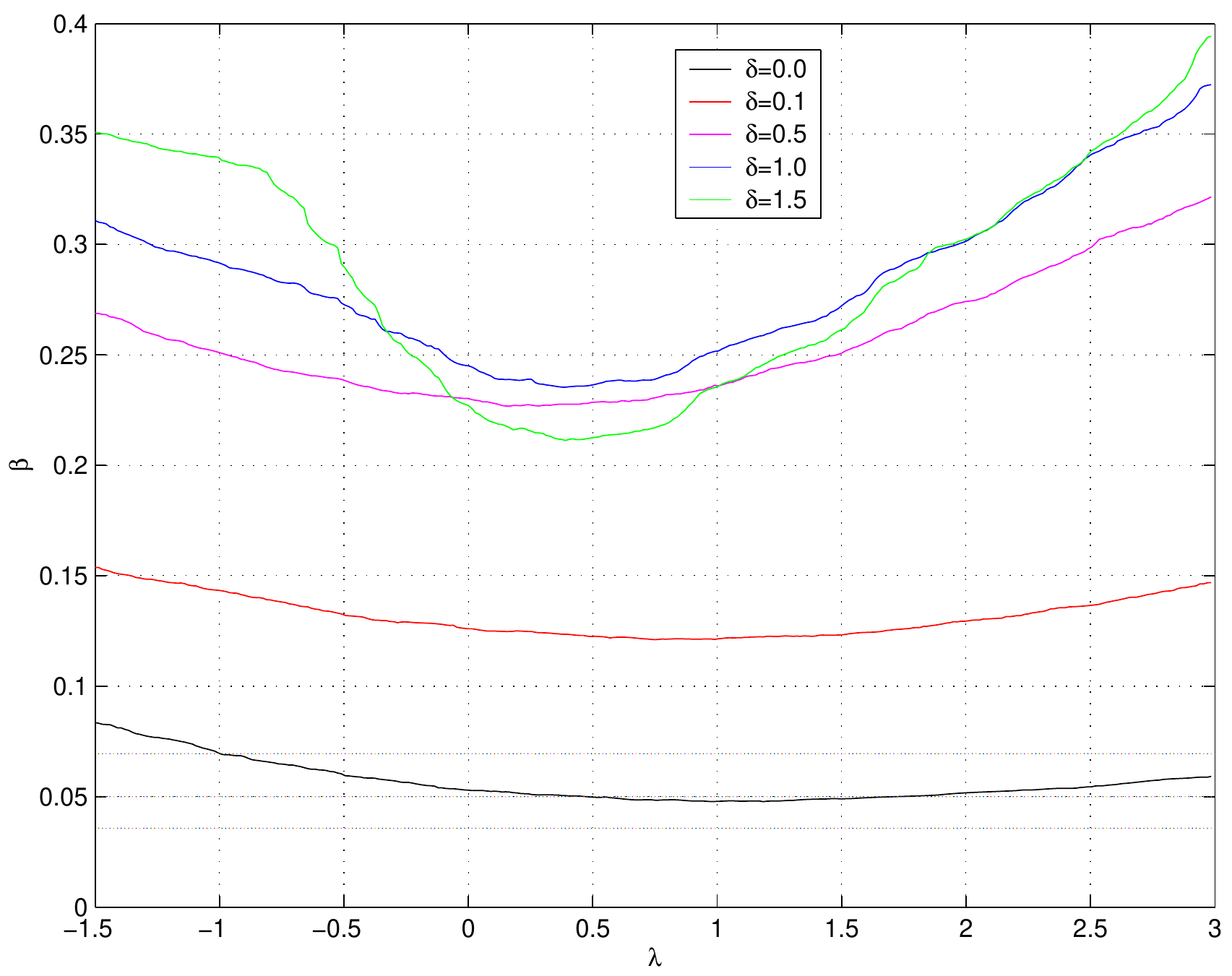}%
}%
\\%
{\includegraphics[
height=2.4016in,
width=3.0286in
]%
{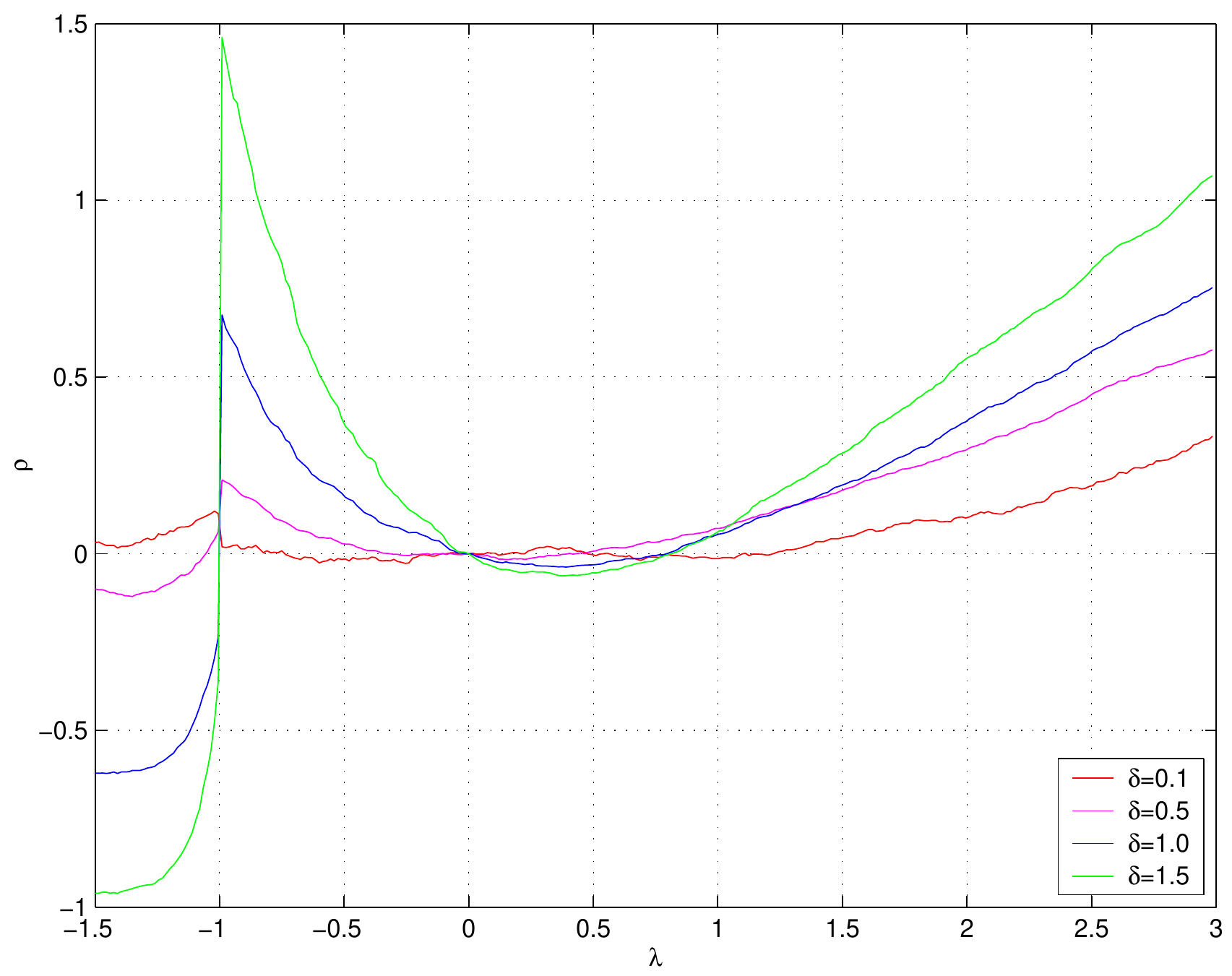}%
}%
&
{\includegraphics[
height=2.4016in,
width=3.0286in
]%
{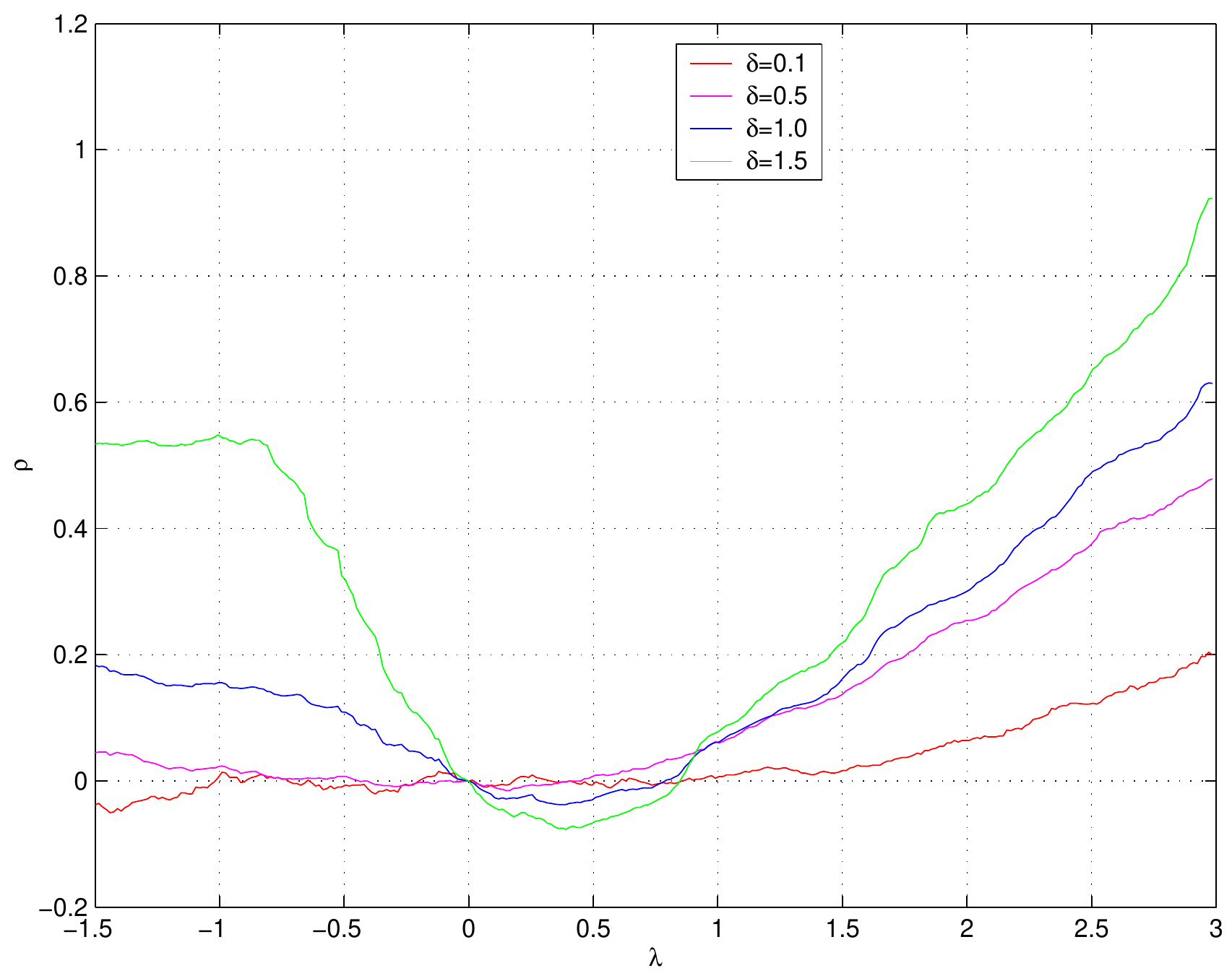}%
}%
\\%
{\includegraphics[
height=2.4016in,
width=3.0286in
]%
{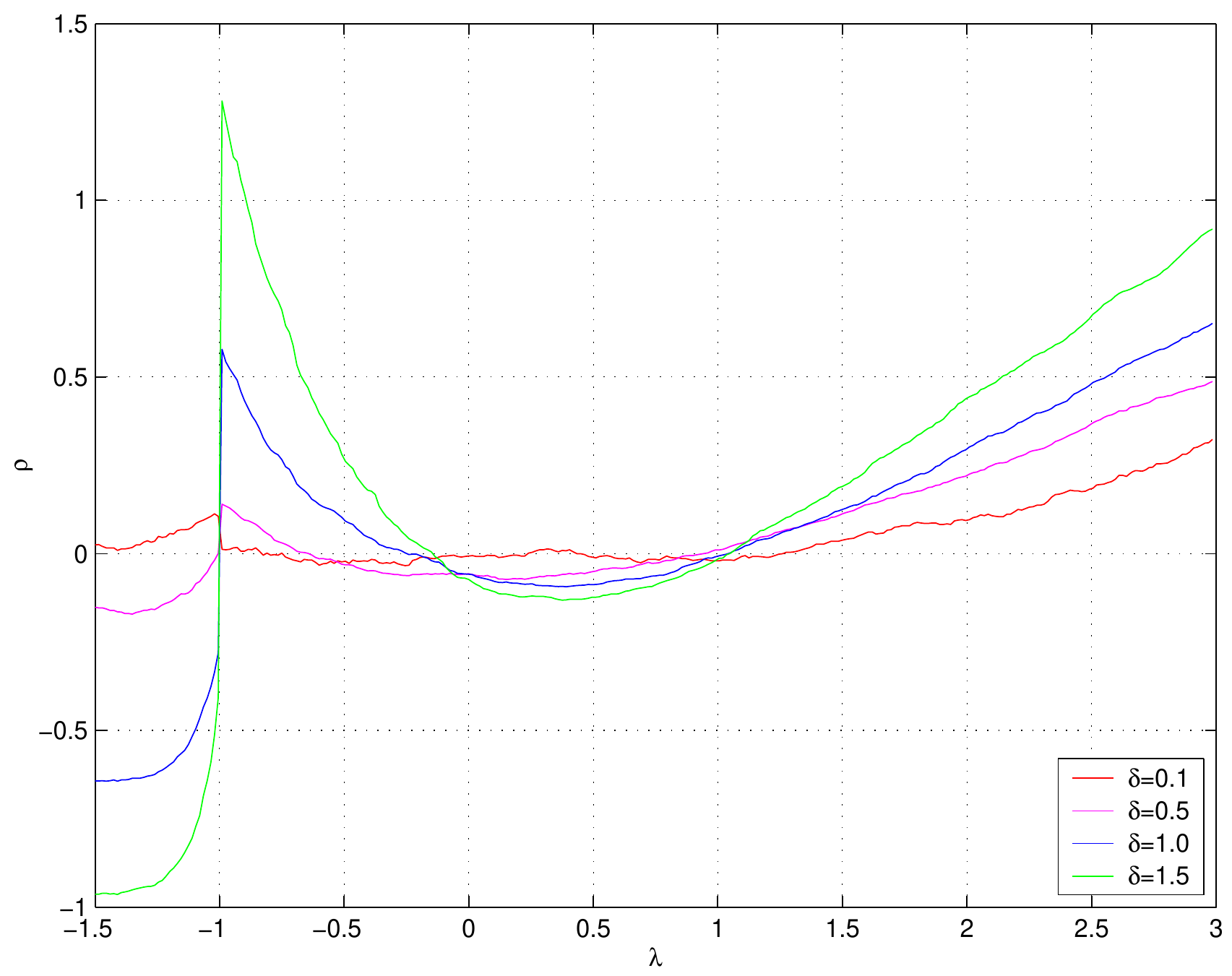}%
}%
&
{\includegraphics[
height=2.4016in,
width=3.0286in
]%
{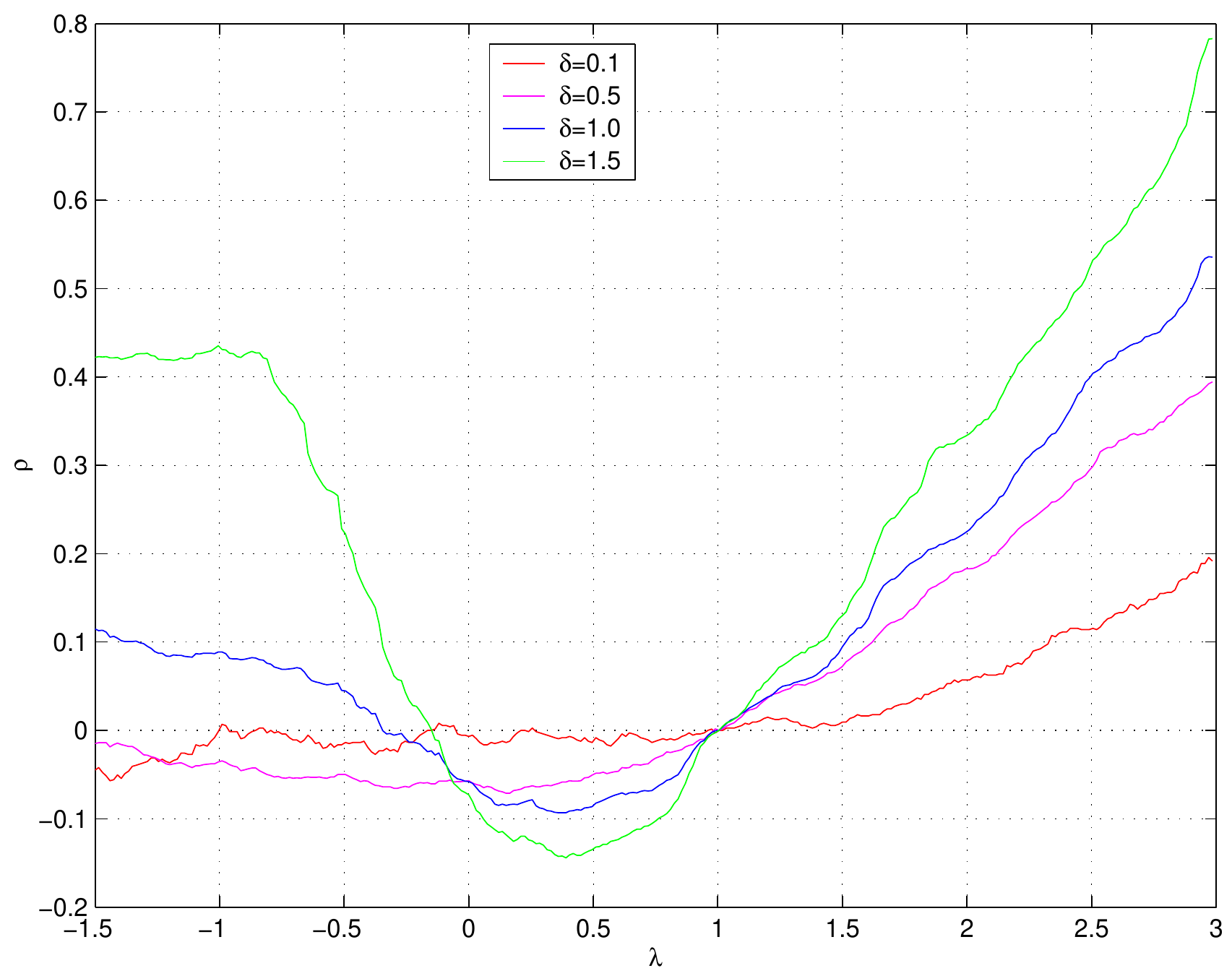}%
}%
\end{tabular}
\caption{Power and relative local efficiencies for $T_{\lambda}$ and $S_{\lambda}$ in scenario 3. \label{fig3}}%
\end{figure}%
%

\begin{figure}[htbp]  \tabcolsep2.8pt  \centering
\begin{tabular}
[c]{cc}%
${T_{\lambda}}$ & ${S_{\lambda}}$\\%
{\includegraphics[
height=2.4016in,
width=3.0286in
]%
{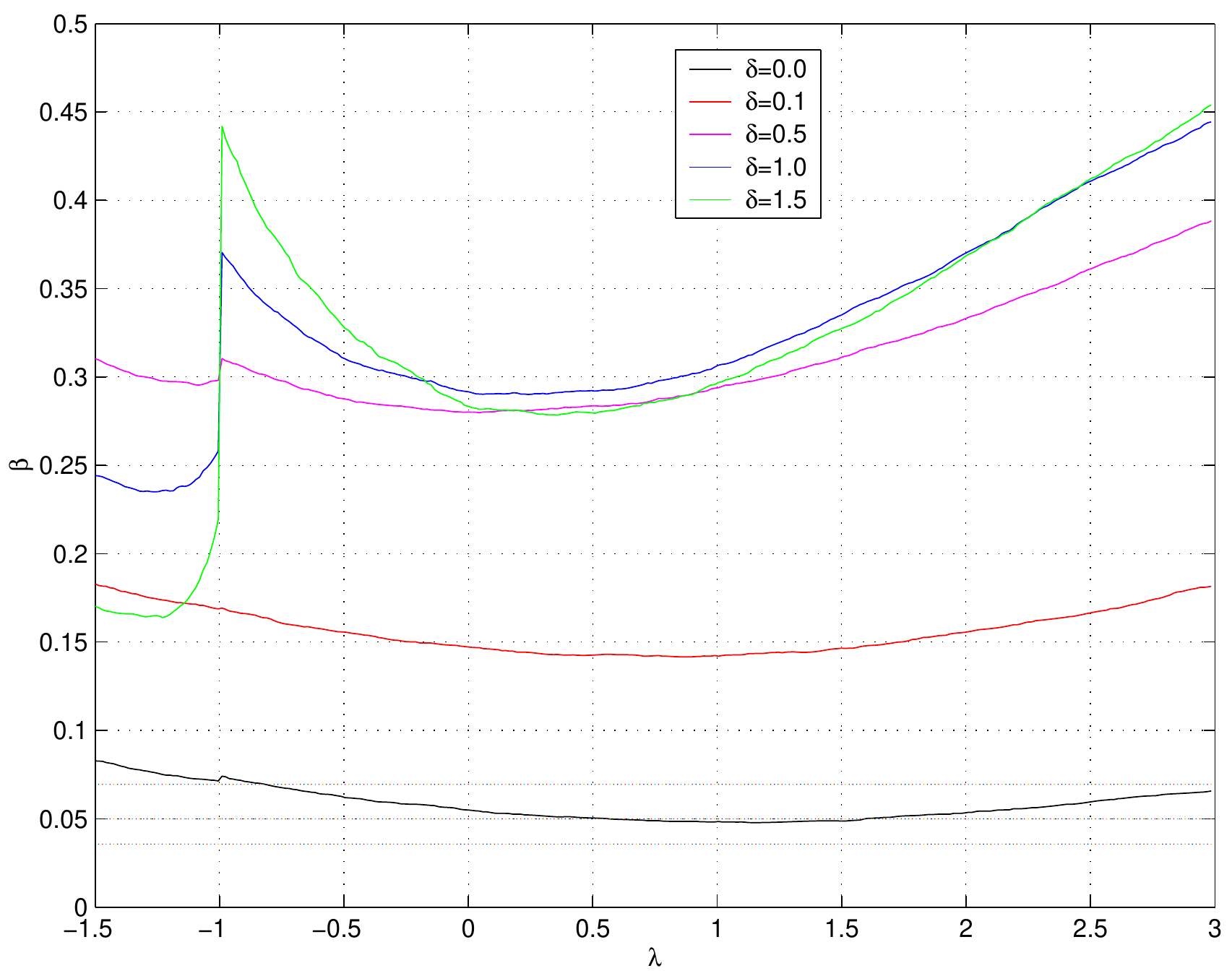}%
}%
&
{\includegraphics[
height=2.4016in,
width=3.0286in
]%
{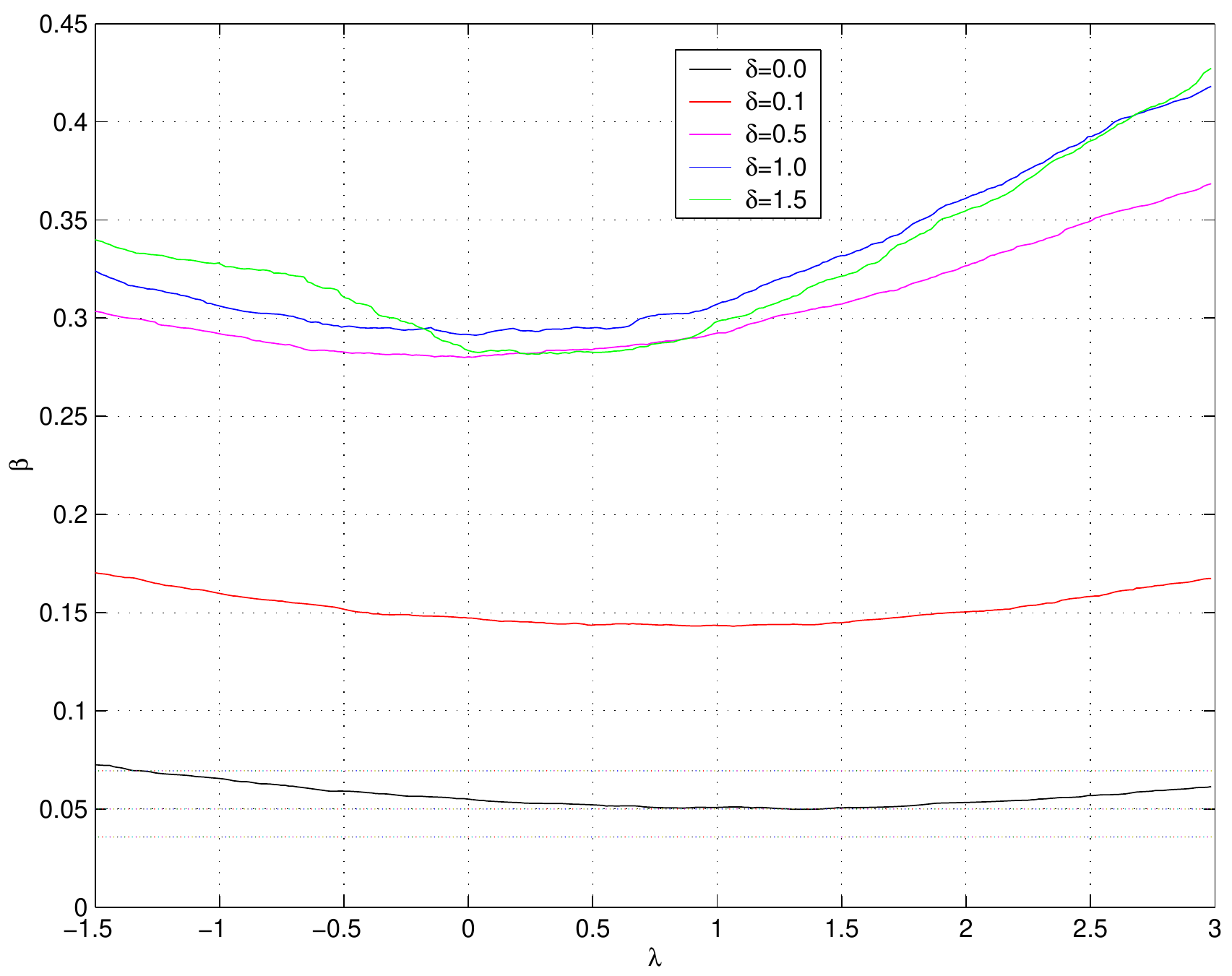}%
}%
\\%
{\includegraphics[
height=2.4016in,
width=3.0286in
]%
{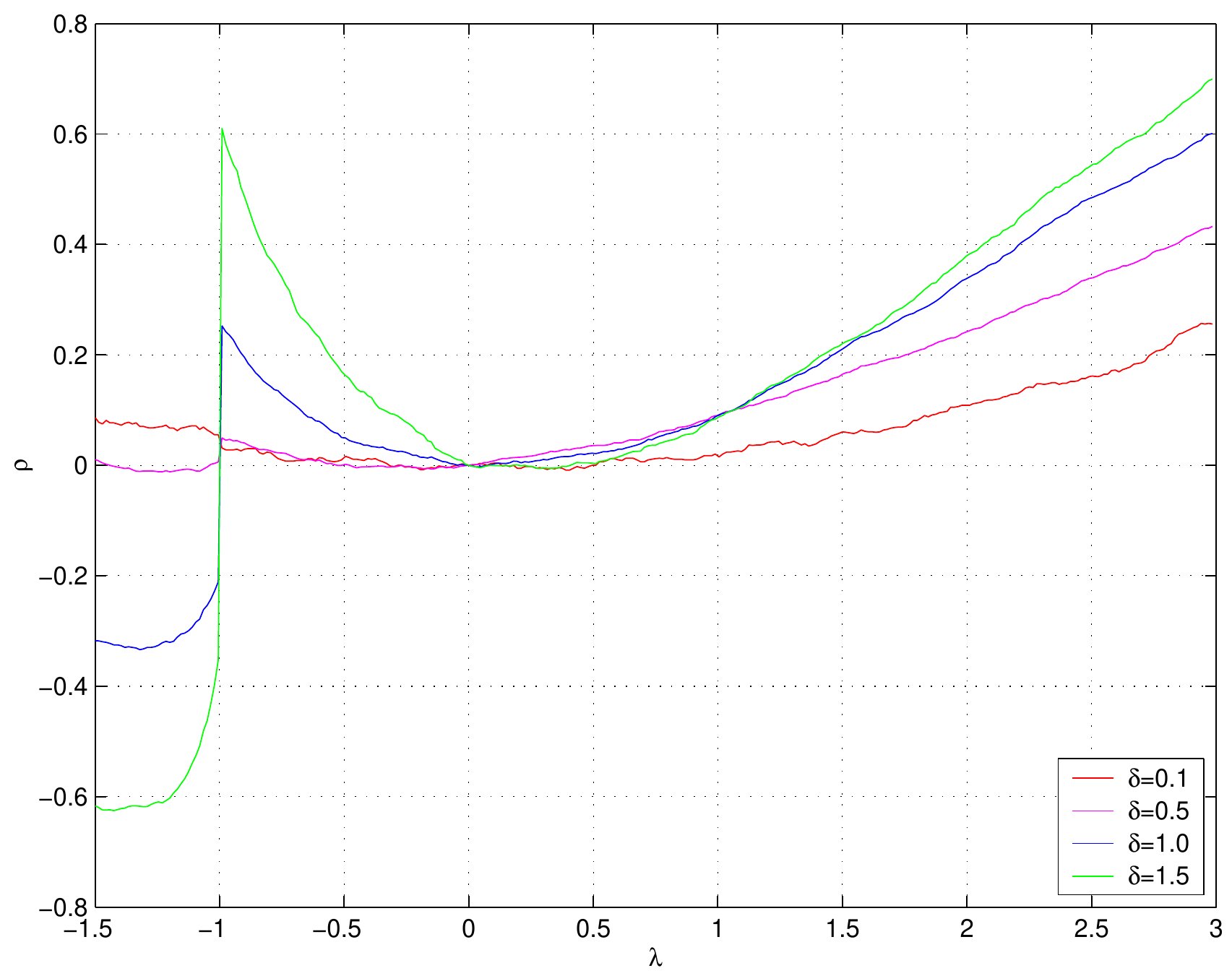}%
}%
&
{\includegraphics[
height=2.4016in,
width=3.0286in
]%
{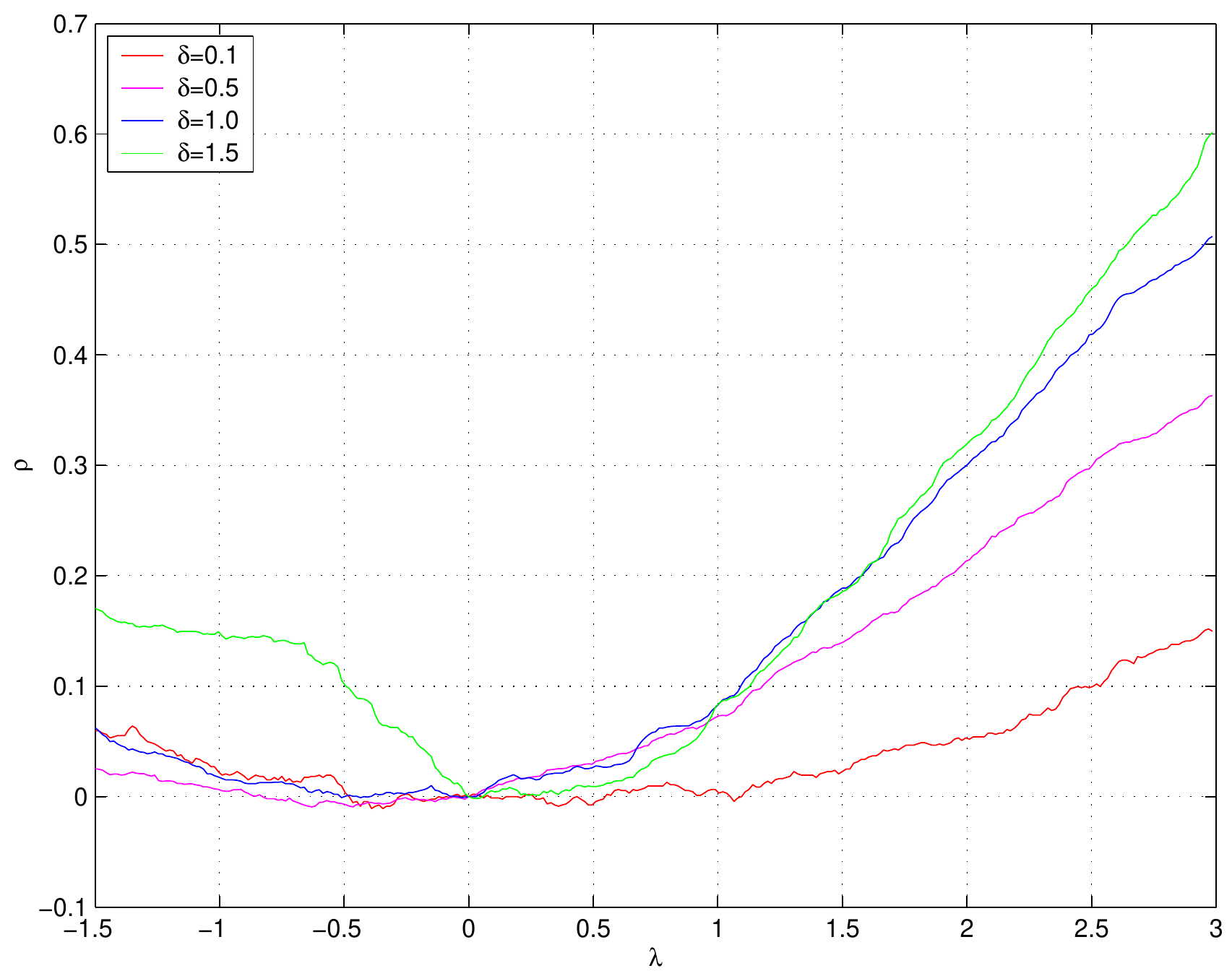}%
}%
\\%
{\includegraphics[
height=2.4016in,
width=3.0286in
]%
{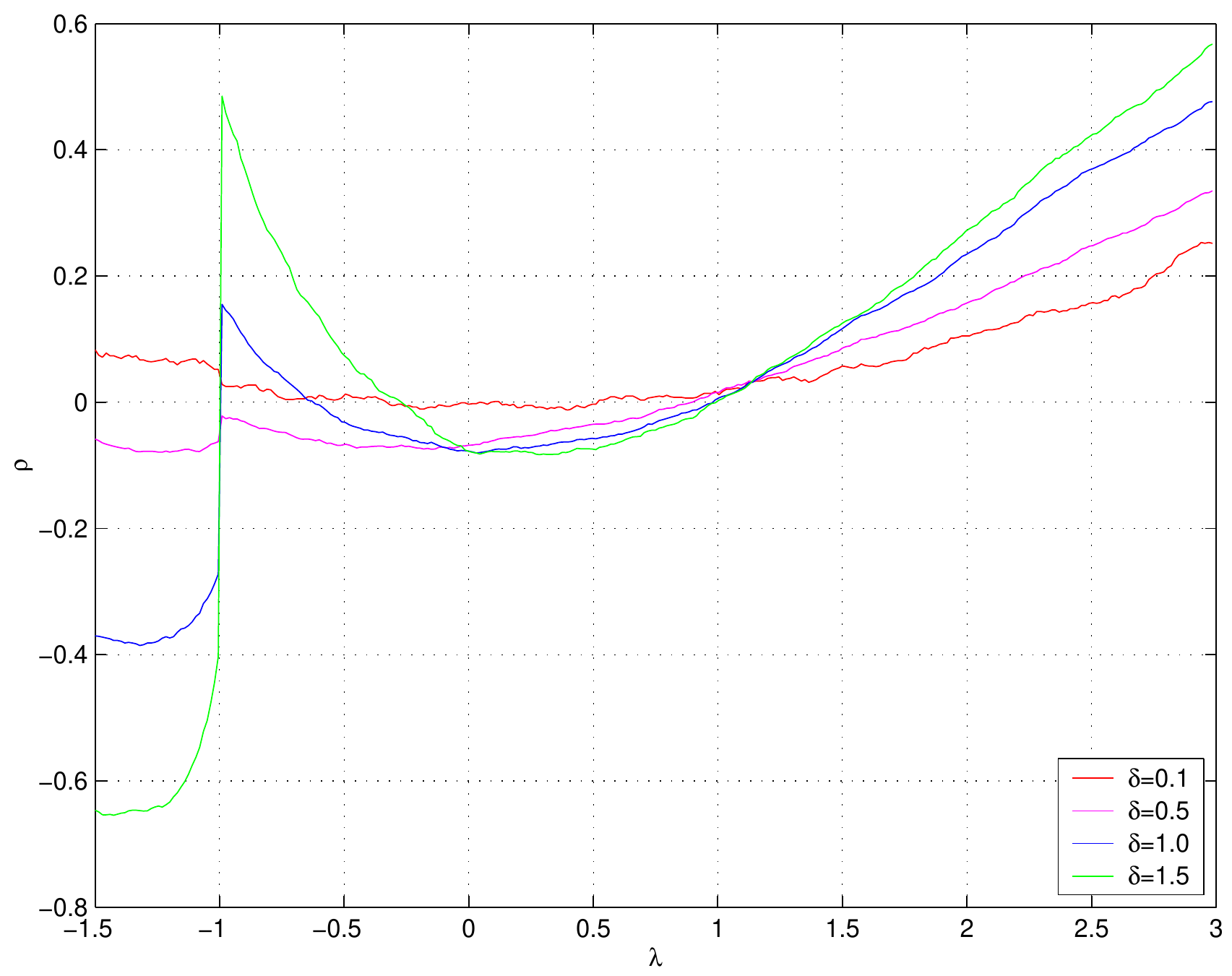}%
}%
&
{\includegraphics[
height=2.4016in,
width=3.0286in
]%
{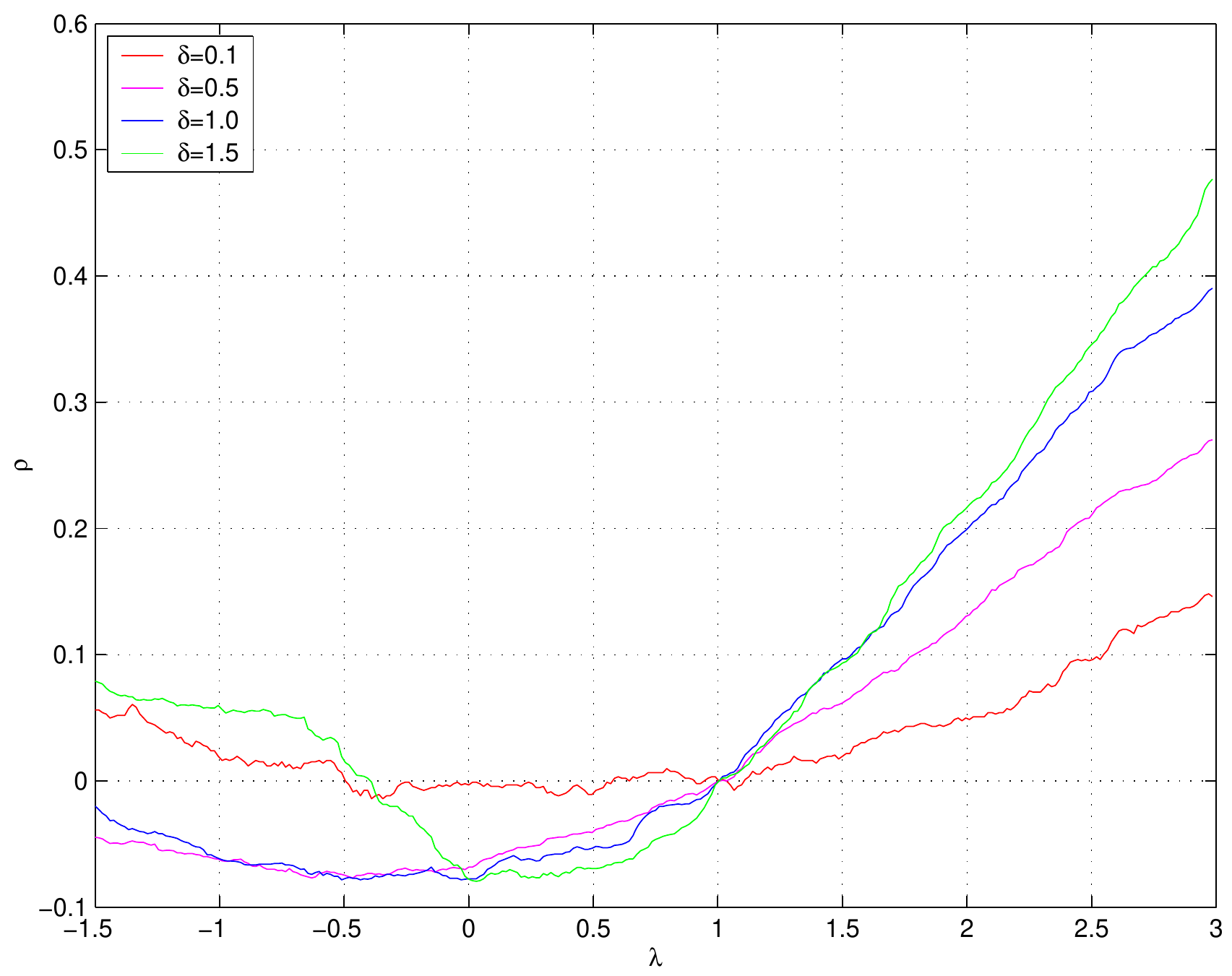}%
}%
\end{tabular}
\caption{Power and relative local efficiencies for $T_{\lambda}$ and $S_{\lambda}$ in scenario 4. \label{fig4}}%
\end{figure}%

\section{Summary and conclusion}

We have proposed and studied two new families of test-statistics, useful for
testing if there exists homogeneity in $I$ independent multinomial samples or
on the contrary, likelihood ordering. Their asymptotic chi-bar distribution is
common, with weights easy to be estimated, using the matrix $\boldsymbol{H}%
(\widehat{\boldsymbol{\theta}})$, and have a simple interpretation in terms of
log-linear modeling. Two algorithms provide the procedure for computing the
estimation of the weights and asymptotic $p$-values of the test-statistics. In
the literature of likelihood ratio ordering, using asymptotic tecniques, the
likelihood ratio test has solely been considered. The simulation study shows
that such a test-statistic has a poor performance for small and moderate
sample sizes and we have seen that it is much better using other
test-statistics such as ${T_{\lambda}}$ and ${S_{\lambda}}$ with $\lambda=2$.

%

\appendix

\section{Appendix}

Suppose we are interested in testing
\begin{equation}
H_{0}:\quad\boldsymbol{R\theta}=\boldsymbol{0}_{(I-1)(J-1)}\quad\text{vs}\quad
H_{1}:\quad\boldsymbol{R}(S)\boldsymbol{\theta}=\boldsymbol{0}_{\mathrm{card}%
(S)}\quad\text{and}\quad\boldsymbol{R\theta}\neq\boldsymbol{0}_{(I-1)(J-1)}.
\label{TB}%
\end{equation}
Under $H_{0}$, the parameter space is $\Theta_{0}=\left\{  \boldsymbol{\theta
}\in%
\mathbb{R}
^{I(J-1)}:\boldsymbol{R\theta}=\boldsymbol{0}_{(I-1)(J-1)}\right\}  $ and the
maximum likelihood estimator of $\boldsymbol{\theta}$ in $\Theta_{0}$ is
$\widehat{\boldsymbol{\theta}}=\arg\max_{\boldsymbol{\theta\in}\Theta_{0}}%
\ell(\boldsymbol{N};\boldsymbol{\theta})$. Under the alternative hypothesis
the parameter space is $\Theta(S)-\Theta_{0}$, where $\Theta(S)=\left\{
\boldsymbol{\theta}\in%
\mathbb{R}
^{I(J-1)}:\boldsymbol{R}(S)\boldsymbol{\theta}=\boldsymbol{0}_{(I-1)(J-1)}%
\right\}  $, that is, under both hypotheses, $H_{0}$\ and $H_{1}$, the
parameter space is $\Theta(S)=\left\{  \boldsymbol{\theta}\in%
\mathbb{R}
^{I(J-1)}:\boldsymbol{R}(S)\boldsymbol{\theta}=\boldsymbol{0}_{(I-1)(J-1)}%
\right\}  $ and the maximum likelihood estimator of $\boldsymbol{\theta}$ in
$\Theta(S)$ is $\widehat{\boldsymbol{\theta}}(S)=\arg\max_{\boldsymbol{\theta
\in}\Theta(S)}\ell(\boldsymbol{N};\boldsymbol{\theta})$. By following the same
idea we used for building test-statistics (\ref{5a})-(\ref{5b}) we shall
consider two family of test-statistics based on $\phi$-divergence measures,%
\begin{equation}
T_{\phi}(\overline{\boldsymbol{p}},\boldsymbol{p}(\widehat{\boldsymbol{\theta
}}(S)),\boldsymbol{p}(\widehat{\boldsymbol{\theta}}))=\frac{2n}{\phi^{\prime
}(1)}(d_{\phi}(\overline{\boldsymbol{p}},\boldsymbol{p}%
(\widehat{\boldsymbol{\theta}}))-d_{\phi}(\overline{\boldsymbol{p}%
},\boldsymbol{p}(\widehat{\boldsymbol{\theta}}(S)))) \label{5aB}%
\end{equation}
and%
\begin{equation}
S_{\phi}(\boldsymbol{p}(\widehat{\boldsymbol{\theta}}(S)),\boldsymbol{p}%
(\widehat{\boldsymbol{\theta}}))=\frac{2n}{\phi^{\prime}(1)}d_{\phi
}(\boldsymbol{p}(\widehat{\boldsymbol{\theta}}(S)),\boldsymbol{p}%
(\widehat{\boldsymbol{\theta}})). \label{5bB}%
\end{equation}

\subsection{Proposition\label{prop}}

Under $H_{0}$,
\begin{equation}
S_{\phi}(\boldsymbol{p}(\widehat{\boldsymbol{\theta}}(S)),\boldsymbol{p}%
(\widehat{\boldsymbol{\theta}}))=T_{\phi}(\overline{\boldsymbol{p}%
},\boldsymbol{p}(\widehat{\boldsymbol{\theta}}(S)),\boldsymbol{p}%
(\widehat{\boldsymbol{\theta}}))+\mathrm{o}_{p}(1), \label{D}%
\end{equation}
the asymptotic distribution of (\ref{5aB})\ and (\ref{5bB}) is $\chi_{df}^{2}$
with $df=(I-1)(J-1)-$\textrm{$card$}$(S)$.

\begin{proof}
The second order Taylor expansion of function $\mathrm{d}_{\phi}%
(\boldsymbol{\theta})=\mathrm{d}_{\phi}(\boldsymbol{p}(\boldsymbol{\theta
}),\boldsymbol{p}(\widehat{\boldsymbol{\theta}}))$ about
$\widehat{\boldsymbol{\theta}}$ is%
\begin{equation}
\mathrm{d}_{\phi}(\boldsymbol{\theta})=\mathrm{d}_{\phi}%
(\widehat{\boldsymbol{\theta}})+(\boldsymbol{\theta}%
-\widehat{\boldsymbol{\theta}})^{T}\left.  \frac{\partial}{\partial
\boldsymbol{\theta}}\mathrm{d}_{\phi}(\boldsymbol{\theta})\right\vert
_{\boldsymbol{\theta=}\widehat{\boldsymbol{\theta}}}+\frac{1}{2}%
(\boldsymbol{\theta}-\widehat{\boldsymbol{\theta}})^{T}\left.  \frac
{\partial^{2}}{\partial\boldsymbol{\theta}\partial\boldsymbol{\theta}^{T}%
}\mathrm{d}_{\phi}(\boldsymbol{\theta})\right\vert _{\boldsymbol{\theta
=}\widehat{\boldsymbol{\theta}}}(\boldsymbol{\theta}%
-\widehat{\boldsymbol{\theta}})+\mathrm{o}\left(  \left\Vert
\boldsymbol{\theta}-\widehat{\boldsymbol{\theta}}\right\Vert ^{2}\right)  ,
\label{eq16}%
\end{equation}
where%
\begin{align*}
\left.  \frac{\partial}{\partial\boldsymbol{\theta}}\mathrm{d}_{\phi
}(\boldsymbol{\theta})\right\vert _{\boldsymbol{\theta=}%
\widehat{\boldsymbol{\theta}}}  &  =\boldsymbol{0}_{(I-1)(J-1)},\\
\left.  \frac{\partial^{2}}{\partial\boldsymbol{\theta}\partial
\boldsymbol{\theta}^{T}}\mathrm{d}_{\phi}(\boldsymbol{\theta})\right\vert
_{\boldsymbol{\theta=}\widehat{\boldsymbol{\theta}}}  &  =\phi^{\prime\prime
}\left(  1\right)  \mathcal{I}_{F}^{(n_{1},...,n_{I})}%
(\widehat{\boldsymbol{\theta}}),
\end{align*}
and $\mathcal{I}_{F}^{(n_{1},n_{2})}(\boldsymbol{\theta})$ was defined at the
beginning of Section \ref{Sec2}. Let $\overline{\boldsymbol{\theta}}$ such
that $\overline{\boldsymbol{p}}=\boldsymbol{p}(\overline{\boldsymbol{\theta}%
})$, where $\boldsymbol{p}(\overline{\boldsymbol{\theta}})=\exp
\{\boldsymbol{1}_{IJ}\overline{u}+\boldsymbol{W}_{1}\overline
{\boldsymbol{\theta}}_{1}+\boldsymbol{W}\overline{\boldsymbol{\theta}%
}\boldsymbol{\}}$, with $\overline{u}=u(\overline{\boldsymbol{\theta}})$,
$\overline{\theta}_{1(i)}=\theta_{1(i)}(\overline{\boldsymbol{\theta}})$, is
the the saturated log-linear model. In particular, for $\boldsymbol{\theta
=}\overline{\boldsymbol{\theta}}$ we have%
\[
\mathrm{d}_{\phi}(\boldsymbol{p}(\overline{\boldsymbol{\theta}}%
),\boldsymbol{p}(\widehat{\boldsymbol{\theta}}))=\frac{\phi^{\prime\prime
}\left(  1\right)  }{2}(\overline{\boldsymbol{\theta}}%
-\widehat{\boldsymbol{\theta}})^{T}\mathcal{I}_{F}^{(n_{1},...,n_{I}%
)}(\widehat{\boldsymbol{\theta}})(\overline{\boldsymbol{\theta}}%
-\widehat{\boldsymbol{\theta}})+\mathrm{o}\left(  \left\Vert \overline
{\boldsymbol{\theta}}-\widehat{\boldsymbol{\theta}}\right\Vert ^{2}\right)  .
\]
In a similar way it is obtained%
\[
\mathrm{d}_{\phi}(\boldsymbol{p}(\overline{\boldsymbol{\theta}}%
),\boldsymbol{p}(\widehat{\boldsymbol{\theta}}(S)))=\frac{\phi^{\prime\prime
}\left(  1\right)  }{2}(\overline{\boldsymbol{\theta}}%
-\widehat{\boldsymbol{\theta}}(S))^{T}\mathcal{I}_{F}^{(n_{1},...,n_{I}%
)}(\widehat{\boldsymbol{\theta}}(S))(\overline{\boldsymbol{\theta}%
}-\widehat{\boldsymbol{\theta}}(S))+\mathrm{o}\left(  \left\Vert
\overline{\boldsymbol{\theta}}-\widehat{\boldsymbol{\theta}}(S)\right\Vert
^{2}\right)  .
\]
Multiplying both sides of the equality by $\frac{2n}{\phi^{\prime\prime
}\left(  1\right)  }$ and taking the difference in both sides of the equality%
\begin{align*}
T_{\phi}(\overline{\boldsymbol{p}},\boldsymbol{p}(\widehat{\boldsymbol{\theta
}}(S)),\boldsymbol{p}(\widehat{\boldsymbol{\theta}}))  &  =\frac{2n}%
{\phi^{\prime\prime}(1)}\left(  \mathrm{d}_{\phi}(\boldsymbol{p}%
(\overline{\boldsymbol{\theta}}),\boldsymbol{p}(\widehat{\boldsymbol{\theta}%
}))-\mathrm{d}_{\phi}(\boldsymbol{p}(\overline{\boldsymbol{\theta}%
}),\boldsymbol{p}(\widehat{\boldsymbol{\theta}}(S)))\right) \\
&  =\sqrt{n}(\overline{\boldsymbol{\theta}}-\widehat{\boldsymbol{\theta}}%
)^{T}\mathcal{I}_{F}^{(n_{1},...,n_{I})}(\widehat{\boldsymbol{\theta}}%
)\sqrt{n}(\overline{\boldsymbol{\theta}}-\widehat{\boldsymbol{\theta}%
})+\mathrm{o}\left(  \left\Vert \sqrt{n}\left(  \overline{\boldsymbol{\theta}%
}-\widehat{\boldsymbol{\theta}}\right)  \right\Vert ^{2}\right) \\
&  -\sqrt{n}(\overline{\boldsymbol{\theta}}-\widehat{\boldsymbol{\theta}%
}(S))^{T}\mathcal{I}_{F}^{(n_{1},...,n_{I})}(\widehat{\boldsymbol{\theta}%
}(S))\sqrt{n}(\overline{\boldsymbol{\theta}}-\widehat{\boldsymbol{\theta}%
}(S))+\mathrm{o}\left(  \left\Vert \sqrt{n}\left(  \overline
{\boldsymbol{\theta}}-\widehat{\boldsymbol{\theta}}(S)\right)  \right\Vert
^{2}\right)  .
\end{align*}
Now we are going to generalize the three types of estimators by
$\widehat{\boldsymbol{\theta}}(\bullet)$, understanding that for
$\bullet=\varnothing$, $\widehat{\boldsymbol{\theta}}(\varnothing
)=\overline{\boldsymbol{\theta}}$, $\boldsymbol{R}(\varnothing\mathbf{)=0}%
_{(I-1)(J-1)\times(IJ-1)}$, for $\bullet=E$, $\widehat{\boldsymbol{\theta}%
}(E)=\widehat{\boldsymbol{\theta}}$, $\boldsymbol{R}(E\mathbf{)=}%
\boldsymbol{R}$, and $\bullet=S$, $\widehat{\boldsymbol{\theta}}(S)$ and
$\boldsymbol{R}(S\mathbf{)}$ as originally defined. It is well-known that%
\begin{equation}
\sqrt{n}(\widehat{\boldsymbol{\theta}}(\bullet)-\boldsymbol{\theta}%
_{0})=\boldsymbol{P}(\boldsymbol{\theta}_{0},\bullet\mathbf{)}\frac{1}%
{\sqrt{n}}\left.  \frac{\partial}{\partial\boldsymbol{\theta}}\ell
(\boldsymbol{N};\boldsymbol{\theta})\right\vert _{\boldsymbol{\theta
}=\boldsymbol{\theta}_{0}}+\mathrm{o}_{p}(\boldsymbol{1}_{k}), \label{eq14}%
\end{equation}
where $\boldsymbol{\theta}_{0}$ is the true and unknown value of the
parameter,%
\[
\boldsymbol{P}(\boldsymbol{\theta}_{0},\bullet\mathbf{)}\mathbf{=}%
\mathcal{I}_{F}^{-1}(\boldsymbol{\theta}_{0})-\mathcal{I}_{F}^{-1}%
(\boldsymbol{\theta}_{0})\boldsymbol{R}^{T}(\bullet\mathbf{)}\left(
\boldsymbol{R}(\bullet\mathbf{)}\mathcal{I}_{F}^{-1}(\boldsymbol{\theta}%
_{0})\boldsymbol{R}^{T}(\bullet\mathbf{)}\right)  ^{-1}\boldsymbol{R}%
(\bullet\mathbf{)}\mathcal{I}_{F}^{-1}(\boldsymbol{\theta}_{0}),
\]
is the variance covariance matrix of $\widehat{\boldsymbol{\theta}}(\bullet
)$,\ and $\frac{1}{\sqrt{n}}\left.  \frac{\partial}{\partial\boldsymbol{\theta
}}\ell(\boldsymbol{N};\boldsymbol{\theta})\right\vert _{\boldsymbol{\theta
}=\boldsymbol{\theta}_{0}}\underset{n\rightarrow\infty}{\overset{\mathcal{L}%
}{\longrightarrow}}\mathcal{N}(\boldsymbol{0}_{k},\mathcal{I}_{F}%
(\boldsymbol{\theta}_{0}))$ by the Central Limit Theorem. We shall denote%
\[
\boldsymbol{P}(\boldsymbol{\theta}_{0}\mathbf{)=}\boldsymbol{P}%
(\boldsymbol{\theta}_{0},E\mathbf{)=}\mathcal{I}_{F}^{-1}(\boldsymbol{\theta
}_{0})-\mathcal{I}_{F}^{-1}(\boldsymbol{\theta}_{0})\boldsymbol{R}^{T}\left(
\boldsymbol{R}\mathcal{I}_{F}^{-1}(\boldsymbol{\theta}_{0})\boldsymbol{R}%
^{T}\right)  ^{-1}\boldsymbol{R}\mathcal{I}_{F}^{-1}(\boldsymbol{\theta}%
_{0}).
\]
Taking the differences of both sides of the equality in (\ref{eq14}) with
cases $\bullet=\varnothing$\ and $\bullet=E$, we obtain%
\begin{equation}
\sqrt{n}(\overline{\boldsymbol{\theta}}-\widehat{\boldsymbol{\theta}})=\left(
\mathcal{I}_{F}^{-1}(\boldsymbol{\theta}_{0})-\boldsymbol{P}%
(\boldsymbol{\theta}_{0}\mathbf{)}\right)  \frac{1}{\sqrt{n}}\left.
\frac{\partial}{\partial\boldsymbol{\theta}}\ell(\boldsymbol{N}%
;\boldsymbol{\theta})\right\vert _{\boldsymbol{\theta}=\boldsymbol{\theta}%
_{0}}+\mathrm{o}_{p}(\boldsymbol{1}_{k}), \label{A}%
\end{equation}
with cases $\bullet=\varnothing$\ and $\bullet=S$,%
\begin{equation}
\sqrt{n}(\overline{\boldsymbol{\theta}}-\widehat{\boldsymbol{\theta}%
}(S))=\left(  \mathcal{I}_{F}^{-1}(\boldsymbol{\theta}_{0})-\boldsymbol{P}%
(\boldsymbol{\theta}_{0},S\mathbf{)}\right)  \frac{1}{\sqrt{n}}\left.
\frac{\partial}{\partial\boldsymbol{\theta}}\ell(\boldsymbol{N}%
;\boldsymbol{\theta})\right\vert _{\boldsymbol{\theta}=\boldsymbol{\theta}%
_{0}}+\mathrm{o}_{p}(\boldsymbol{1}_{k}), \label{B}%
\end{equation}
and taking into account $\mathcal{I}_{F}(\widehat{\boldsymbol{\theta}%
})\underset{n\rightarrow\infty}{\overset{P}{\longrightarrow}}\mathcal{I}%
_{F}(\boldsymbol{\theta}_{0})$,%
\begin{align}
&  T_{\phi}(\overline{\boldsymbol{p}},\boldsymbol{p}%
(\widehat{\boldsymbol{\theta}}(S)),\boldsymbol{p}(\widehat{\boldsymbol{\theta
}}))\nonumber\\
&  =\frac{1}{\sqrt{n}}\left.  \frac{\partial}{\partial\boldsymbol{\theta}^{T}%
}\ell(\boldsymbol{N};\boldsymbol{\theta})\right\vert _{\boldsymbol{\theta
}=\boldsymbol{\theta}_{0}}\left(  \boldsymbol{P}(\boldsymbol{\theta}%
_{0},S\mathbf{)}-\boldsymbol{P}(\boldsymbol{\theta}_{0}\mathbf{)}\right)
^{T}\mathcal{I}_{F}(\boldsymbol{\theta}_{0})\left(  \boldsymbol{P}%
(\boldsymbol{\theta}_{0},S\mathbf{)}-\boldsymbol{P}(\boldsymbol{\theta}%
_{0}\mathbf{)}\right)  \frac{1}{\sqrt{n}}\left.  \frac{\partial}%
{\partial\boldsymbol{\theta}}\ell(\boldsymbol{N};\boldsymbol{\theta
})\right\vert _{\boldsymbol{\theta}=\boldsymbol{\theta}_{0}}+\mathrm{o}%
_{p}(1)\nonumber\\
&  =\boldsymbol{Y}^{T}\boldsymbol{Y}+\mathrm{o}_{p}(1), \label{C}%
\end{align}
where%
\[
\boldsymbol{Y}=\boldsymbol{A}(\boldsymbol{\theta}_{0}\mathbf{)}\left(
\boldsymbol{P}(\boldsymbol{\theta}_{0},S\mathbf{)}-\boldsymbol{P}%
(\boldsymbol{\theta}_{0}\mathbf{)}\right)  \boldsymbol{A}(\boldsymbol{\theta
}_{0}\mathbf{)}^{T}\boldsymbol{Z}\text{,}%
\]
with $\boldsymbol{Z}\sim\mathcal{N}(\boldsymbol{0}_{(I-1)(J-1)},\boldsymbol{I}%
_{(I-1)(J-1)}\mathbf{)}$ and $\boldsymbol{A}(\boldsymbol{\theta}_{0}%
\mathbf{)}$\ is the Cholesky's factorization matrix for a non singular matrix
such a Fisher information matrix, that is $\mathcal{I}_{F}(\boldsymbol{\theta
}_{0})=\boldsymbol{A}(\boldsymbol{\theta}_{0}\mathbf{)}^{T}\boldsymbol{A}%
(\boldsymbol{\theta}_{0}\mathbf{)}$. In other words%
\[
\boldsymbol{Y}\sim\mathcal{N}(\boldsymbol{0}_{k},\boldsymbol{A}%
(\boldsymbol{\theta}_{0}\mathbf{)}\left(  \boldsymbol{P}(\boldsymbol{\theta
}_{0},S\mathbf{)}-\boldsymbol{P}(\boldsymbol{\theta}_{0}\mathbf{)}\right)
\boldsymbol{A}(\boldsymbol{\theta}_{0}\mathbf{)}^{T})\text{,}%
\]
where the variance covariance matrix is idempotent and symmetric. Following
Lemma 3 in Ferguson (1996, page 57), $\boldsymbol{A}(\boldsymbol{\theta}%
_{0}\mathbf{)}\left(  \boldsymbol{P}(\boldsymbol{\theta}_{0},S\mathbf{)}%
-\boldsymbol{P}(\boldsymbol{\theta}_{0}\mathbf{)}\right)  \boldsymbol{A}%
(\boldsymbol{\theta}_{0}\mathbf{)}^{T}$ is idempotent and symmetric, if only
if $T_{\phi}(\overline{\boldsymbol{p}},\boldsymbol{p}%
(\widehat{\boldsymbol{\theta}}(S)),\boldsymbol{p}(\widehat{\boldsymbol{\theta
}}))$ is a chi-square random variable with degrees of freedom
\[
df=\mathrm{rank}(\boldsymbol{A}(\boldsymbol{\theta}_{0}\mathbf{)}\left(
\boldsymbol{P}(\boldsymbol{\theta}_{0},S\mathbf{)}-\boldsymbol{P}%
(\boldsymbol{\theta}_{0}\mathbf{)}\right)  \boldsymbol{A}(\boldsymbol{\theta
}_{0}\mathbf{)}^{T})=\mathrm{trace}(\boldsymbol{A}(\boldsymbol{\theta}%
_{0}\mathbf{)}\left(  \boldsymbol{P}(\boldsymbol{\theta}_{0},S\mathbf{)}%
-\boldsymbol{P}(\boldsymbol{\theta}_{0}\mathbf{)}\right)  \boldsymbol{A}%
(\boldsymbol{\theta}_{0}\mathbf{)}^{T}).
\]
Since%
\[
\left(  \boldsymbol{P}(\boldsymbol{\theta}_{0},S\mathbf{)}-\boldsymbol{P}%
(\boldsymbol{\theta}_{0}\mathbf{)}\right)  ^{T}\mathcal{I}_{F}%
(\boldsymbol{\theta}_{0})\left(  \boldsymbol{P}(\boldsymbol{\theta}%
_{0},S\mathbf{)}-\boldsymbol{P}(\boldsymbol{\theta}_{0}\mathbf{)}\right)
=\boldsymbol{P}(\boldsymbol{\theta}_{0},S\mathbf{)}-\boldsymbol{P}%
(\boldsymbol{\theta}_{0}\mathbf{),}%
\]
the condition is reached. The effective degrees of freedom are given by%
\begin{align*}
df  &  =\mathrm{trace}(\boldsymbol{P}(\boldsymbol{\theta}_{0},S\mathbf{)}%
\boldsymbol{A}(\boldsymbol{\theta}_{0}\mathbf{)}^{T}\boldsymbol{A}%
(\boldsymbol{\theta}_{0}\mathbf{)})-\mathrm{trace}(\boldsymbol{P}%
(\boldsymbol{\theta}_{0}\mathbf{)}\boldsymbol{A}(\boldsymbol{\theta}%
_{0}\mathbf{)}^{T}\boldsymbol{A}(\boldsymbol{\theta}_{0}\mathbf{)}%
)=\mathrm{trace}(\boldsymbol{P}(\boldsymbol{\theta}_{0},S\mathbf{)}%
\mathcal{I}_{F}(\boldsymbol{\theta}_{0}))-\mathrm{trace}(\boldsymbol{P}%
(\boldsymbol{\theta}_{0}\mathbf{)}\mathcal{I}_{F}(\boldsymbol{\theta}_{0}))\\
&  =\mathrm{trace}(-\left(  \boldsymbol{R}(S\mathbf{)}\mathcal{I}_{F}%
^{-1}(\boldsymbol{\theta}_{0})\boldsymbol{R}^{T}(S\mathbf{)}\right)
^{-1}\boldsymbol{R}(S\mathbf{)}\mathcal{I}_{F}^{-1}(\boldsymbol{\theta}%
_{0})\boldsymbol{R}^{T}(S\mathbf{)})\\
&  -\mathrm{trace}(-\left(  \boldsymbol{R}\mathcal{I}_{F}^{-1}%
(\boldsymbol{\theta}_{0})\boldsymbol{R}^{T}\right)  ^{-1}\boldsymbol{R}%
\mathcal{I}_{F}^{-1}(\boldsymbol{\theta}_{0})\boldsymbol{R}^{T})\\
&  =(I-1)(J-1)-\mathrm{card}(S).
\end{align*}
Regarding the other test-statistic $S_{\phi}(\boldsymbol{p}%
(\widehat{\boldsymbol{\theta}}(S)),\boldsymbol{p}(\widehat{\boldsymbol{\theta
}}))$, observe that if we take (\ref{eq16}), in particular for
$\boldsymbol{\theta=}\widehat{\boldsymbol{\theta}}(S)$ it is obtained%
\[
\mathrm{d}_{\phi}(\widehat{\boldsymbol{\theta}}(S))=\frac{\phi^{\prime\prime
}\left(  1\right)  }{2}(\widehat{\boldsymbol{\theta}}%
(S)-\widehat{\boldsymbol{\theta}})^{T}\mathcal{I}_{F}%
(\widehat{\boldsymbol{\theta}})(\widehat{\boldsymbol{\theta}}%
(S)-\widehat{\boldsymbol{\theta}})+\mathrm{o}\left(  \left\Vert
\widehat{\boldsymbol{\theta}}(S)-\widehat{\boldsymbol{\theta}}\right\Vert
^{2}\right)  .
\]
In addition, (\ref{A})$-$(\ref{B}) is%
\[
\sqrt{n}(\widehat{\boldsymbol{\theta}}(S)-\widehat{\boldsymbol{\theta}%
})=\left(  \boldsymbol{P}(\boldsymbol{\theta}_{0},S\mathbf{)}-\boldsymbol{P}%
(\boldsymbol{\theta}_{0}\mathbf{)}\right)  \frac{1}{\sqrt{n}}\left.
\frac{\partial}{\partial\boldsymbol{\theta}}\ell(\boldsymbol{N}%
;\boldsymbol{\theta})\right\vert _{\boldsymbol{\theta}=\boldsymbol{\theta}%
_{0}}+\mathrm{o}_{p}(\boldsymbol{1}_{k}),
\]
and taking into account $\mathcal{I}_{F}(\widehat{\boldsymbol{\theta}%
})\underset{n\rightarrow\infty}{\overset{P}{\longrightarrow}}\mathcal{I}%
_{F}(\boldsymbol{\theta}_{0})$ and (\ref{C}), it follows (\ref{D}), which
means from Slutsky's Theorem that both test-statistics have the same
asymptotic distribution.
\end{proof}

\subsection{Lemma \label{LemContrA}}

Let $\boldsymbol{Y}$ be a $k$-dimensional random variable with normal
distribution $\mathcal{N}\left(  \boldsymbol{0}_{k},\boldsymbol{I}\right)  $
with $\boldsymbol{Q}$ being a projection matrix, that is idempotent and
symmetric, and let fixed $k$-dimensional vectors $\boldsymbol{d}_{i}$ such
that for them either $\boldsymbol{Qd}_{i}=\boldsymbol{0}_{k}$ or
$\boldsymbol{Qd}_{i}=\boldsymbol{d}_{i}$, $i=1,...,k$, is true. Then $\left(
\boldsymbol{Y}^{T}\boldsymbol{Q}\boldsymbol{Y}\left\vert \boldsymbol{d}%
_{i}^{T}\boldsymbol{Y}\geq0,i=1,...,k\right.  \right)  \sim\chi_{df}^{2}$,
where $df=\mathrm{rank}(\boldsymbol{Q})$.

\begin{proof}
This result can be found in several sources, for instance in Kud\^{o} (1963,
page 414), Barlow et al. (1972, page 128) and Shapiro (1985, page 139).
\end{proof}

\subsection{Proof of Theorem \ref{Th1}\label{ProofTh1ContrA}}

We shall perform the proof for $S_{\phi}(\boldsymbol{p}%
(\widetilde{\boldsymbol{\theta}}),\boldsymbol{p}(\widehat{\boldsymbol{\theta}%
}))$. It suppose that it is true $\boldsymbol{R\theta}\geq\boldsymbol{0}%
_{(I-1)(J-1)}$ and we want to test $\boldsymbol{R\theta}=\boldsymbol{0}%
_{(I-1)(J-1)}$ ($H_{0}$). It is clear that if $H_{0}$ is not true is because
there exists some index $i\in E$ such that $\boldsymbol{R}(\{i\}\mathbf{)}%
\boldsymbol{\theta}>0$. Let us consider the family of all possible subsets in
$E$, denoted by $\mathcal{F}(E)$, then\ we shall specify more thoroughly
$\widetilde{\boldsymbol{\theta}}$ by $\widetilde{\boldsymbol{\theta}}(S)$ when
there exists $S\in\mathcal{F}(E)$ such that%
\[
\boldsymbol{R}(S)\widetilde{\boldsymbol{\theta}}=\boldsymbol{0}_{\mathrm{card}%
(S)}\qquad\text{and}\qquad\boldsymbol{R}(S^{C})\widetilde{\boldsymbol{\theta}%
}>\boldsymbol{0}_{(I-1)(J-1)-\mathrm{card}(S)}.
\]
It is clear that for a sample $\widetilde{\boldsymbol{\theta}}%
=\widetilde{\boldsymbol{\theta}}(S)$ can be true only for a unique set of
indices $S\in\mathcal{F}(E)$, and thus by applying the Theorem of Total
Probability%
\[
\Pr\left(  S_{\phi}(\boldsymbol{p}(\widetilde{\boldsymbol{\theta}%
}),\boldsymbol{p}(\widehat{\boldsymbol{\theta}}))\leq x\right)  =\sum
_{S\in\mathcal{F}(E)}\Pr\left(  S_{\phi}(\boldsymbol{p}%
(\widetilde{\boldsymbol{\theta}}),\boldsymbol{p}(\widehat{\boldsymbol{\theta}%
}))\leq x,\widetilde{\boldsymbol{\theta}}=\widetilde{\boldsymbol{\theta}%
}(S)\right)  .
\]
From the Karush-Khun-Tucker necessary conditions (see for instance Theorem
4.2.13 in Bazaraa et al. (2006)) to solve the optimization problem $\max
\ell(\boldsymbol{N};\boldsymbol{\theta})$ s.t. $\boldsymbol{R\theta}%
\geq\boldsymbol{0}_{(I-1)(J-1)}$, associated with
$\widetilde{\boldsymbol{\theta}}$,%
\begin{subequations}
\begin{align}
\frac{\partial}{\partial\boldsymbol{\theta}}\ell(\boldsymbol{N}%
;\boldsymbol{\theta})+\sum_{i=1}^{(I-1)(J-1)}\lambda_{i}\boldsymbol{R}%
^{T}(\{i\}\mathbf{)}  &  =0\text{, }i=1,...,(I-1)(J-1),\label{KKT1}\\
\lambda_{i}\boldsymbol{R}(\{i\}\mathbf{)}\boldsymbol{\theta}  &  =0\text{,
}i=1,...,(I-1)(J-1),\label{KKT2}\\
\lambda_{i}  &  \leq0\text{, }i=1,...,(I-1)(J-1), \label{KKT3}%
\end{align}
the only conditions which characterize the MLE $\widetilde{\boldsymbol{\theta
}}=\widetilde{\boldsymbol{\theta}}(S)$ with a specific $S\in\mathcal{F}(E)$,
are the complementary slackness conditions $\boldsymbol{R}(\{i\}\mathbf{)}%
\boldsymbol{\theta}>0$, for $i\in S$ and $\lambda_{i}<0$, for $i\in S^{C}$,
since $\frac{\partial}{\partial\boldsymbol{\theta}}\ell(\boldsymbol{N}%
;\boldsymbol{\theta})+\lambda_{i}\boldsymbol{R}^{T}(\{i\}\mathbf{)}=0$,
$i=1,...,(I-1)(J-1)$,$\ \boldsymbol{R}(\{i\}\mathbf{)}\boldsymbol{\theta}=0$,
for $i\in S^{C}$ and $\lambda_{i}=0$, for $i\in S$ are redundant conditions
once we know that the Karush-Khun-Tucker necessary conditions are true for all
the possible sets $S\in\mathcal{F}(E)$ which define
$\widetilde{\boldsymbol{\theta}}=\widetilde{\boldsymbol{\theta}}(S)$. For this
reason we can consider%
\end{subequations}
\begin{align*}
&  \Pr\left(  S_{\phi}(\boldsymbol{p}(\widetilde{\boldsymbol{\theta}%
}),\boldsymbol{p}(\widehat{\boldsymbol{\theta}}))\leq
x,\widetilde{\boldsymbol{\theta}}=\widetilde{\boldsymbol{\theta}}(S)\right)
=\\
&  \Pr\left(  S_{\phi}(\boldsymbol{p}(\widetilde{\boldsymbol{\theta}%
}),\boldsymbol{p}(\widehat{\boldsymbol{\theta}}))\leq
x,\widetilde{\boldsymbol{\lambda}}(S)<\boldsymbol{0}_{\mathrm{card}%
(S)},\boldsymbol{R}(S^{C})\widetilde{\boldsymbol{\theta}}(S)>\boldsymbol{0}%
_{(I-1)(J-1)-\mathrm{card}(S)}\right)  ,
\end{align*}
where $\widetilde{\boldsymbol{\lambda}}(S)$ is the vector of the vector of
Karush-Khun-Tucker multipliers associated with estimator
$\widetilde{\boldsymbol{\theta}}(S)$. Furthermore, under $H_{0}$,
$\boldsymbol{R}\widetilde{\boldsymbol{\theta}}(S)=\boldsymbol{R}%
\widetilde{\boldsymbol{\theta}}(S)-\boldsymbol{R\theta}_{0}$, because
$\boldsymbol{R\theta}_{0}=\boldsymbol{0}_{(I-1)(J-1)}$, hence%
\[
\Pr\left(  S_{\phi}(\boldsymbol{p}(\widetilde{\boldsymbol{\theta}%
}),\boldsymbol{p}(\widehat{\boldsymbol{\theta}}))\leq x\right)  =\sum
_{S\in\mathcal{F}(E)}\Pr\left(  S_{\phi}(\boldsymbol{p}%
(\widetilde{\boldsymbol{\theta}}),\boldsymbol{p}(\widehat{\boldsymbol{\theta}%
}))\leq x,\widetilde{\boldsymbol{\lambda}}(S)<\boldsymbol{0}_{\mathrm{card}%
(S)},\boldsymbol{R}(S^{C})\widetilde{\boldsymbol{\theta}}(S)-\boldsymbol{R}%
(S^{C})\boldsymbol{\theta}_{0}>\boldsymbol{0}_{\mathrm{card}(S^{C})}\right)
,
\]
where $\mathrm{card}(S^{C})=(I-1)(J-1)-\mathrm{card}(S)$. On the other hand,
(\ref{KKT1}) and (\ref{KKT2}) are also true for $(\widehat{\boldsymbol{\theta
}}^{T}(S),\widehat{\boldsymbol{\lambda}}^{T}(S))^{T}$ according to the
Lagrange multipliers method. Hence, $\widetilde{\boldsymbol{\theta}%
}(S)=\widehat{\boldsymbol{\theta}}(S)$ and $\widetilde{\boldsymbol{\lambda}%
}(S)=\widehat{\boldsymbol{\lambda}}(S)$. It follows that:\newline$\bullet$
under $\widetilde{\boldsymbol{\theta}}=\widehat{\boldsymbol{\theta}}(S)$,
$S_{\phi}(\boldsymbol{p}(\widetilde{\boldsymbol{\theta}}),\boldsymbol{p}%
(\widehat{\boldsymbol{\theta}}))=S_{\phi}(\boldsymbol{p}%
(\widehat{\boldsymbol{\theta}}(S)),\boldsymbol{p}(\widehat{\boldsymbol{\theta
}}))$ and taking into account the Proposition given in Section \ref{prop}
\begin{align*}
&  S_{\phi}(\boldsymbol{p}(\widetilde{\boldsymbol{\theta}}),\boldsymbol{p}%
(\widehat{\boldsymbol{\theta}}))=T_{\phi}(\overline{\boldsymbol{p}%
},\boldsymbol{p}(\widehat{\boldsymbol{\theta}}(S)),\boldsymbol{p}%
(\widehat{\boldsymbol{\theta}}))+\mathrm{o}_{p}(1)\\
&  =\left(  \boldsymbol{A}(\boldsymbol{\theta}_{0}\mathbf{)}\left(
\boldsymbol{P}(\boldsymbol{\theta}_{0},S\mathbf{)}-\boldsymbol{P}%
(\boldsymbol{\theta}_{0}\mathbf{)}\right)  \boldsymbol{A}(\boldsymbol{\theta
}_{0}\mathbf{)}^{T}\boldsymbol{Z}\right)  ^{T}\left(  \boldsymbol{A}%
(\boldsymbol{\theta}_{0}\mathbf{)}\left(  \boldsymbol{P}(\boldsymbol{\theta
}_{0},S\mathbf{)}-\boldsymbol{P}(\boldsymbol{\theta}_{0}\mathbf{)}\right)
\boldsymbol{A}(\boldsymbol{\theta}_{0}\mathbf{)}^{T}\boldsymbol{Z}\right)
+\mathrm{o}_{p}(1),\\
&  =\boldsymbol{Z}^{T}\boldsymbol{A}(\boldsymbol{\theta}_{0}\mathbf{)}\left(
\boldsymbol{P}(\boldsymbol{\theta}_{0},S\mathbf{)}-\boldsymbol{P}%
(\boldsymbol{\theta}_{0}\mathbf{)}\right)  \boldsymbol{A}(\boldsymbol{\theta
}_{0}\mathbf{)}^{T}\boldsymbol{Z}+\mathrm{o}_{p}(1).
\end{align*}
where $\boldsymbol{Z}\sim\mathcal{N}\left(  \boldsymbol{0}_{k},\boldsymbol{I}%
_{k}\right)  $.$\newline\bullet$ under $\widetilde{\boldsymbol{\lambda}%
}(S)=\widehat{\boldsymbol{\lambda}}(S)$ and from Sen et al. (2010, page 267
formula (8.6.28))%
\begin{align*}
\frac{1}{\sqrt{n}}\widetilde{\boldsymbol{\lambda}}(S)  &  =\sqrt
{n}\boldsymbol{Q}^{T}(\boldsymbol{\theta}_{0},S\mathbf{)}\frac{1}{\sqrt{n}%
}\left.  \frac{\partial}{\partial\boldsymbol{\theta}}\ell(\boldsymbol{N}%
;\boldsymbol{\theta})\right\vert _{\boldsymbol{\theta}=\boldsymbol{\theta}%
_{0}}+\mathrm{o}_{p}(\boldsymbol{1}_{\mathrm{card}(S)})\\
&  =\boldsymbol{Q}^{T}(\boldsymbol{\theta}_{0},S\mathbf{)}\boldsymbol{A}%
(\boldsymbol{\theta}_{0}\mathbf{)}^{T}\boldsymbol{Z}+\mathrm{o}_{p}%
(\boldsymbol{1}_{\mathrm{card}(S)}),
\end{align*}
where%
\[
\boldsymbol{Q}(\boldsymbol{\theta}_{0},S\mathbf{)}\mathbf{=}-\mathcal{I}%
_{F}^{-1}(\boldsymbol{\theta}_{0})\boldsymbol{R}^{T}(S\mathbf{)}%
\boldsymbol{L}(\boldsymbol{\theta}_{0},S\mathbf{)}\left(  \boldsymbol{R}%
(S\mathbf{)}\mathcal{I}_{F}^{-1}(\boldsymbol{\theta}_{0})\boldsymbol{R}%
^{T}(S\mathbf{)}\right)  ^{-1};
\]
$\bullet$ under $\widetilde{\boldsymbol{\theta}}=\widehat{\boldsymbol{\theta}%
}(S)$ and from (\ref{eq14})%
\begin{align*}
\sqrt{n}\left(  \boldsymbol{R}(S^{C})\widetilde{\boldsymbol{\theta}%
}(S)-\boldsymbol{R}(S^{C})\boldsymbol{\theta}_{0}\right)   &  =\sqrt
{n}\boldsymbol{R}(S^{C})\boldsymbol{P}(\boldsymbol{\theta}_{0},S\mathbf{)}%
\frac{1}{\sqrt{n}}\left.  \frac{\partial}{\partial\boldsymbol{\theta}}%
\ell(\boldsymbol{N};\boldsymbol{\theta})\right\vert _{\boldsymbol{\theta
}=\boldsymbol{\theta}_{0}}+\mathrm{o}_{p}(\boldsymbol{1}_{\mathrm{card}%
(S^{C})})\\
&  =\boldsymbol{R}(S^{C})\boldsymbol{P}(\boldsymbol{\theta}_{0},S\mathbf{)}%
\boldsymbol{A}(\boldsymbol{\theta}_{0}\mathbf{)}^{T}\boldsymbol{Z}%
+\mathrm{o}_{p}(\boldsymbol{1}_{\mathrm{card}(S^{C})}).
\end{align*}
That is,%
\begin{align*}
&  \lim_{n\rightarrow\infty}\Pr\left(  S_{\phi}(\boldsymbol{p}%
(\widetilde{\boldsymbol{\theta}}),\boldsymbol{p}(\widehat{\boldsymbol{\theta}%
}))\leq x\right)  =\sum_{S\in\mathcal{F}(E)}\Pr\left(  \boldsymbol{Z}_{3}%
^{T}(S)\boldsymbol{Z}_{3}(S)\leq x,\boldsymbol{Z}_{1}(S)\geq\boldsymbol{0}%
_{\mathrm{card}(S)},\boldsymbol{Z}_{2}(S)\geq\boldsymbol{0}_{\mathrm{card}%
(S^{C})}\right) \\
&  =\sum_{S\in\mathcal{F}(E)}\Pr\left(  \boldsymbol{Z}_{3}^{T}%
(S)\boldsymbol{Z}_{3}(S)\leq x|\boldsymbol{Z}_{1}(S)\geq\boldsymbol{0}%
_{\mathrm{card}(S)},\boldsymbol{Z}_{2}(S)\geq\boldsymbol{0}_{\mathrm{card}%
(S^{C})}\right)  \Pr\left(  \boldsymbol{Z}_{1}(S)\geq\boldsymbol{0}%
_{\mathrm{card}(S)},\boldsymbol{Z}_{2}(S)\geq\boldsymbol{0}_{\mathrm{card}%
(S^{C})}\right) \\
&  =\sum_{S\in\mathcal{F}(E)}\Pr\left(  \boldsymbol{Z}_{3}^{T}%
(S)\boldsymbol{Z}_{3}(S)\leq x\left\vert \left(  \boldsymbol{Z}_{1}%
^{T}(S),\boldsymbol{Z}_{2}^{T}(S)\right)  ^{T}\geq\boldsymbol{0}%
_{(I-1)(J-1)}\right.  \right)  \Pr\left(  \boldsymbol{Z}_{1}(S)\geq
\boldsymbol{0}_{\mathrm{card}(S)},\boldsymbol{Z}_{2}(S)\geq\boldsymbol{0}%
_{\mathrm{card}(S^{C})}\right)  ,
\end{align*}
where%
\begin{align*}
\boldsymbol{Z}_{3}(S)  &  =\boldsymbol{M}_{3}(\boldsymbol{\theta}%
_{0},S\mathbf{)}\boldsymbol{Z,}\qquad\boldsymbol{M}_{3}(\boldsymbol{\theta
}_{0},S\mathbf{)=}\boldsymbol{A}(\boldsymbol{\theta}_{0}\mathbf{)}\left(
\boldsymbol{P}(\boldsymbol{\theta}_{0},S\mathbf{)}-\boldsymbol{P}%
(\boldsymbol{\theta}_{0}\mathbf{)}\right)  \boldsymbol{A}(\boldsymbol{\theta
}_{0}\mathbf{)}^{T},\\
\boldsymbol{Z}_{1}(S)  &  =\boldsymbol{M}_{1}(\boldsymbol{\theta}%
_{0},S\mathbf{)}\boldsymbol{Z},\qquad\boldsymbol{M}_{1}(\boldsymbol{\theta
}_{0},S\mathbf{)=-}\boldsymbol{Q}^{T}(\boldsymbol{\theta}_{0},S\mathbf{)}%
\boldsymbol{A}(\boldsymbol{\theta}_{0}\mathbf{)}^{T},\\
\boldsymbol{Z}_{2}(S)  &  =\boldsymbol{M}_{2}(\boldsymbol{\theta}%
_{0},S\mathbf{)}\boldsymbol{Z},\qquad\boldsymbol{M}_{2}(\boldsymbol{\theta
}_{0},S\mathbf{)=}\boldsymbol{R}(S^{C})\boldsymbol{P}(\boldsymbol{\theta}%
_{0},S\mathbf{)}\boldsymbol{A}(\boldsymbol{\theta}_{0}\mathbf{)}^{T}.
\end{align*}
Taking into account that $\boldsymbol{M}_{3}(\boldsymbol{\theta}%
_{0},S\mathbf{)}\boldsymbol{M}_{2}^{T}(\boldsymbol{\theta}_{0},S\mathbf{)=}%
\boldsymbol{M}_{2}^{T}(\boldsymbol{\theta}_{0},S\mathbf{)}$ and
$\boldsymbol{M}_{3}(\boldsymbol{\theta}_{0},S\mathbf{)}\boldsymbol{M}_{1}%
^{T}(\boldsymbol{\theta}_{0},S\mathbf{)=}\boldsymbol{0}_{(I-1)(J-1)\times
\mathrm{card}(S)}$, by applying the lemma given in Section \ref{LemContrA}%
\[
\Pr\left(  \boldsymbol{Z}_{3}^{T}(S)\boldsymbol{Z}_{3}(S)\leq x\left\vert
\left(  \boldsymbol{Z}_{1}^{T}(S),\boldsymbol{Z}_{2}^{T}(S)\right)  ^{T}%
\geq\boldsymbol{0}_{(I-1)(J-1)}\right.  \right)  =\Pr\left(  \chi_{df}^{2}\leq
x\right)
\]
where%
\begin{align*}
df  &  =\mathrm{rank}\left(  \boldsymbol{A}(\boldsymbol{\theta}_{0}%
\mathbf{)}\left(  \boldsymbol{P}(\boldsymbol{\theta}_{0},S\mathbf{)}%
-\boldsymbol{P}(\boldsymbol{\theta}_{0}\mathbf{)}\right)  \boldsymbol{A}%
(\boldsymbol{\theta}_{0}\mathbf{)}^{T}\right)  =\mathrm{trace}\left(
\boldsymbol{A}(\boldsymbol{\theta}_{0}\mathbf{)}\left(  \boldsymbol{P}%
(\boldsymbol{\theta}_{0},S\mathbf{)}-\boldsymbol{P}(\boldsymbol{\theta}%
_{0}\mathbf{)}\right)  \boldsymbol{A}(\boldsymbol{\theta}_{0}\mathbf{)}%
^{T}\right) \\
&  =(I-1)(J-1)-\mathrm{card}(S).
\end{align*}
Finally,%
\begin{align*}
&  \lim_{n\rightarrow\infty}\Pr\left(  S_{\phi}(\boldsymbol{p}%
(\widetilde{\boldsymbol{\theta}}),\boldsymbol{p}(\widehat{\boldsymbol{\theta}%
}))\leq x\right) \\
&  =\sum_{S\in\mathcal{F}(E)}\Pr\left(  \chi_{(I-1)(J-1)-\mathrm{card}(S)}%
^{2}\leq x\right)  \Pr\left(  \boldsymbol{Z}_{1}(S)\geq\boldsymbol{0}%
_{\mathrm{card}(S)},\boldsymbol{Z}_{2}(S)\geq\boldsymbol{0}_{\mathrm{card}%
(S^{C})}\right) \\
&  =\sum_{j=0}^{(I-1)(J-1)}\Pr\left(  \chi_{(I-1)(J-1)-j}^{2}\leq x\right)
\sum_{S\in\mathcal{F}(E),\mathrm{card}(S)=j}\Pr\left(  \boldsymbol{Z}%
_{1}(S)\geq\boldsymbol{0}_{\mathrm{card}(S)},\boldsymbol{Z}_{2}(S)\geq
\boldsymbol{0}_{\mathrm{card}(S^{C})}\right)  ,
\end{align*}
and since $\boldsymbol{Q}^{T}(\boldsymbol{\theta}_{0},S\mathbf{)}%
\mathcal{I}_{F}(\boldsymbol{\theta}_{0})\boldsymbol{P}(\boldsymbol{\theta}%
_{0},S\mathbf{)=}\boldsymbol{0}_{\mathrm{card}(S)\times(I-1)(J-1)}$, it holds
$\boldsymbol{M}_{1}(\boldsymbol{\theta}_{0},S\mathbf{)}\boldsymbol{M}_{2}%
^{T}(\boldsymbol{\theta}_{0},S\mathbf{)}=\boldsymbol{0}_{\mathrm{card}%
(S)\times\mathrm{card}(S^{C})}$ which means that $\boldsymbol{Z}_{1}(S)$ and
$\boldsymbol{Z}_{2}(S)$ are independent, that is%
\[
\lim_{n\rightarrow\infty}\Pr\left(  S_{\phi}(\boldsymbol{p}%
(\widetilde{\boldsymbol{\theta}}),\boldsymbol{p}(\widehat{\boldsymbol{\theta}%
}))\leq x\right)  =\sum_{j=0}^{(I-1)(J-1)}\Pr\left(  \chi_{(I-1)(J-1)-j}%
^{2}\leq x\right)  w_{j}(\boldsymbol{\theta}_{0})
\]
where the expression of $w_{j}(\boldsymbol{\theta}_{0})$ is (\ref{eqw}). We
have also,
\[
\mathrm{Var}(\boldsymbol{Z}_{1}(S))=\boldsymbol{M}_{1}(\boldsymbol{\theta}%
_{0},S\mathbf{)}\boldsymbol{M}_{1}^{T}(\boldsymbol{\theta}_{0},S\mathbf{)}%
=\boldsymbol{Q}^{T}(\boldsymbol{\theta}_{0},S\mathbf{)}\mathcal{I}%
_{F}(\boldsymbol{\theta}_{0})\boldsymbol{Q}(\boldsymbol{\theta}_{0}%
,S\mathbf{)}=\left(  \boldsymbol{R}(S\mathbf{)}\mathcal{I}_{F}^{-1}%
(\boldsymbol{\theta}_{0})\boldsymbol{R}^{T}(S\mathbf{)}\right)  ^{-1}%
=\boldsymbol{\Sigma}_{1}(\boldsymbol{\theta}_{0},S),
\]%
\begin{align*}
\mathrm{Var}(\boldsymbol{Z}_{2}(S))  &  =\boldsymbol{M}_{2}(\boldsymbol{\theta
}_{0},S\mathbf{)}\boldsymbol{M}_{2}^{T}(\boldsymbol{\theta}_{0},S\mathbf{)}%
=\boldsymbol{R}(S^{C})\boldsymbol{P}(\boldsymbol{\theta}_{0},S\mathbf{)}%
\mathcal{I}_{F}(\boldsymbol{\theta}_{0})\boldsymbol{P}^{T}(\boldsymbol{\theta
}_{0},S\mathbf{)}\boldsymbol{R}^{T}(S^{C})\\
&  =\boldsymbol{R}(S^{C})\boldsymbol{P}(\boldsymbol{\theta}_{0},S\mathbf{)}%
\boldsymbol{R}^{T}(S^{C})=\boldsymbol{\Sigma}_{2}(\boldsymbol{\theta}_{0},S).
\end{align*}
The proof of $T_{\phi}(\overline{\boldsymbol{p}},\boldsymbol{p}%
(\widetilde{\boldsymbol{\theta}}),\boldsymbol{p}(\widehat{\boldsymbol{\theta}%
}))$ is almost immediate from the proof for $S_{\phi}(\boldsymbol{p}%
(\widetilde{\boldsymbol{\theta}}),\boldsymbol{p}(\widehat{\boldsymbol{\theta}%
}))$ and taking into account that for some $S\in\mathcal{F}(E)$%
\[
T_{\phi}(\overline{\boldsymbol{p}},\boldsymbol{p}%
(\widetilde{\boldsymbol{\theta}}),\boldsymbol{p}(\widehat{\boldsymbol{\theta}%
}))=T_{\phi}(\overline{\boldsymbol{p}},\boldsymbol{p}%
(\widehat{\boldsymbol{\theta}}(S)),\boldsymbol{p}(\widehat{\boldsymbol{\theta
}}))+\mathrm{o}_{p}(1)=S_{\phi}(\boldsymbol{p}(\widetilde{\boldsymbol{\theta}%
}),\boldsymbol{p}(\widehat{\boldsymbol{\theta}})).
\]

\end{document}